\documentclass[journal]{IEEEtran}
\usepackage{amsmath,amsfonts,amssymb}
\usepackage{algorithm}
\usepackage{algorithmic}

\usepackage[table]{xcolor}

\usepackage{array}
\usepackage{textcomp}
\usepackage{stfloats}
\usepackage{url}
\usepackage{verbatim}
\usepackage{graphicx}
\usepackage{xcolor}
\usepackage{cite}
\usepackage{hyperref,url}
\usepackage{xcolor}
\usepackage{bbm}
\usepackage{subcaption}
\usepackage{enumitem}
\usepackage{float}
\usepackage[normalem]{ulem}
\usepackage{amsthm}
\theoremstyle{plain}

\newtheorem{definition}{Definition}
\newtheorem{assumption}{Assumption}
\newtheorem{theorem}{Theorem}
\newtheorem{corollary}{Corollary}
\newtheorem{lemma}{Lemma}
\theoremstyle{definition}
\newtheorem{remark}{Remark}
\newtheorem*{remark*}{Remark}

\begin{document}

\title{Diffusion Stochastic Learning Over Adaptive Competing Networks}

\author{Yike Zhao{$^*$}, \IEEEmembership{Student Member, IEEE}, Haoyuan Cai{$^*$}, \IEEEmembership{Student Member, IEEE}, and Ali H. Sayed{$^*$}, \IEEEmembership{Fellow, IEEE}\\
$^*$\'Ecole Polytechnique F\'ed\'erale de Lausanne, Switzerland
\thanks{
The authors are with the Institute of Electrical and Micro Engineering, EPFL, Switzerland. Emails: yike.zhao@epfl.ch, haoyuan.cai@epfl.ch, ali.sayed@epfl.ch}}




\maketitle

\begin{abstract}
This paper studies a stochastic dynamic game between two competing teams, each consisting of a network of collaborating agents. Unlike fully cooperative settings, where all agents share a common objective, each team in this game aims to minimize its own distinct objective. In the adversarial setting, their objectives could be conflicting as in zero-sum games. Throughout the competition, agents share strategic information within their own team while simultaneously inferring and adapting to the strategies of the opposing team.
 We propose diffusion learning algorithms to address two important classes of this network game:  i) a zero-sum game characterized by weak cross-team subgraph interactions, and ii) a general non-zero-sum game exhibiting strong cross-team subgraph interactions.
We analyze the stability performance of the proposed algorithms under reasonable assumptions and illustrate the theoretical results through experiments on Cournot team competition and decentralized GAN training.
\end{abstract}


\section{Introduction}

\IEEEPARstart{A}{daptive} cooperative networks using diffusion learning have achieved considerable success in addressing distributed optimization problems \cite{sayed2014adaptive, sayed2022inference1}.
Although many existing works concentrate primarily on single-model optimization tasks with a single-objective function, many practical applications involve more complex game-theoretic scenarios, such as Cournot team competition in economics \cite{bischi2000global, elettreby2006dynamical, ahmed2006multi, raab2009cournot}, multi-GAN systems \cite{hoang2018mgan}, and competitive e-sports games \cite{jaderberg2019human, vinyals2019grandmaster}. As a result, these applications require networks to operate in non-cooperative environments.


Motivated by this gap, this paper extends the {\em stochastic competing networks} problem from \cite{vlaski2021competing}, where each network aims to minimize its own global objective function defined as follows:
\begin{subequations}
\begin{align}
    \min_{x \in \mathbb{R}^{M_1}}  J^{(1)}\left(x,y\right) &\triangleq \sum_{k \in \mathcal{N}^{(1)}} p_k J_k\left(x,y\right)\label{eq:global_game_ya}\\
    \min_{y \in \mathbb{R}^{M_2}}  J^{(2)}\left(x,y\right) &\triangleq \sum_{k \in \mathcal{N}^{(2)}} p_{k} J_k\left(x,y\right) \label{eq:global_game_yb}\\
    \text{ where } J_k \left(x,y\right) &\triangleq \mathbb{E} Q_k \left(x,y; \boldsymbol{\xi}_k\right), k \in \mathcal{N}
\label{eq:global_game_yc}
\end{align}
\end{subequations}
where $\mathcal{N}^{(1)} \triangleq \{1, \ldots, K_1\}$ and $\mathcal{N}^{(2)} \triangleq \{K_1 + 1, \ldots, K\}$ denote the index sets of agents in teams 1 and 2, respectively, with $K = K_1 + K_2$, and $\mathcal{N} \triangleq \{\mathcal{N}^{(1)}, \mathcal{N}^{(2)}\}$ represents all agents in the game. The global objectives of teams 1 and 2 are given by $J^{(1)}(\cdot,\cdot)$ and $J^{(2)}(\cdot,\cdot)$,  
and $J_{k}(\cdot, \cdot)$ is 
the local objective of agent $k$.
This problem involves two distinct teams, whose strategies (i.e., solution vectors) are represented by $x $ and $y$, respectively. 
The weights $p_k > 0$ satisfy $\sum_{k \in \mathcal{N}^{(t)}} p_k = 1$ for each team $t \in \{1,2\}$. 
Note that formulation
\eqref{eq:global_game_ya}-\eqref{eq:global_game_yc} includes two-network zero-sum games as a special case when
$J^{(1)}\left(x,y\right) = - J^{(2)}\left(x,y\right)$.
This zero-sum formulation encompasses the important class of minimax problems studied in the literature (e.g. \cite{lin2020gradient,cai2024accelerated,cai2024diffusion}).
In this work, we specifically focus on the stochastic scenario, where each agent computes its local loss $Q_{k}(\cdot; \boldsymbol{\xi})$ using only local samples $\{\boldsymbol{\xi}_{k}\}$.
 This scenario is important in large-scale and real-world applications, where evaluating the full gradient of the true risk function is impractical due to
 communication constraints, unknown data distributions, or environmental uncertainty. Furthermore,  each team operates as a fully decentralized network, with agents interacting solely with their immediate neighbors, ensuring scalability and robustness without relying on centralized coordination.

Although problem \eqref{eq:global_game_ya}-\eqref{eq:global_game_yc}  is not entirely new, only a limited number of studies have explored a similar setup \cite{lou2015nash, huang2024no, zhao2025diffusion, meng2023linear,zimmermann2021solving}. Moreover,  existing studies exhibit two main limitations: i) they assume a bipartite graph structure to model cross-team interactions, which is often less practical in real-world applications \cite{lou2015nash, huang2024no, zhao2025diffusion}, and ii) they focus on deterministic settings \cite{meng2023linear,zimmermann2021solving}, thereby excluding more realistic scenarios that involve continuously streaming stochastic data.
In comparison, we relax these stringent requirements and develop novel algorithms that allow the network to continuously learn in an adaptive manner while searching for the equilibrium state.
In the following, we review relevant literature in distributed optimization and game theory, which provides the foundational basis for developing our proposed framework to address the  competing networks problem.

\subsection{Related works}

Distributed optimization problems commonly rely on a communication graph where agents collaborate to optimize a common objective \cite{chen2012diffusion, sayed2022inference1}. 
In this setting, the performance of distributed gradient algorithms has been extensively analyzed
\cite{chen2012diffusion, jakovetic2014fast, nedic2009distributed}. Beyond standard minimization tasks, distributed min-max optimization, which is often interpreted as a cooperative game, has also been studied in various contexts \cite{cai2024diffusion, lin2020gradient, beznosikov2022decentralized}.

In practical scenarios, networked games often involve noncooperative players optimizing their own objectives, each of which depends not only on their own actions but also on those of other players. As a result, agents must track the strategies of others and adapt their decisions accordingly.
When agents are only partially connected, existing approaches typically rely on gradient-based methods for updating their own actions, while employing leader-following consensus protocols \cite{ye2017distributed, ye2020distributed, fang2020distributed} or gossip-based methods \cite{salehisadaghiani2016distributed, salehisadaghiani2018distributed} to track the actions of other players. Other methods involve designing augmented gradient dynamics to ensure stability and convergence \cite{gadjov2018passivity} or applying Nesterov-type acceleration techniques to improve convergence speed \cite{tatarenko2020geometric}.
In certain settings, such as energy consumption games \cite{ye2016game}, each agent’s objective is a function of its own action and a weighted sum of all agents’ actions. This structure allows algorithms to track the aggregated actions of other players rather than the full action profile \cite{ye2016game, deng2018distributed, koshal2016distributed, belgioioso2020distributed, lei2022distributed}, reducing the complexity of information sharing. 
Furthermore, some studies \cite{parise2019variational, yu2017distributed} consider only local interactions, where agents adjust their strategies based on neighboring information without requiring full knowledge of the entire network’s action profile.

In networked systems, cooperative and noncooperative interactions often occur simultaneously. Agents may work together within a team to achieve shared objectives while simultaneously competing against other teams with conflicting interests.
This setting applies to a wide range of real-world applications, including Cournot-team competition and smart grid power management \cite{zimmermann2021solving, yu2023distributed}. More broadly, such game problems can be modeled using theoretical frameworks like two-network zero-sum games \cite{gharesifard2013distributed, lou2015nash, huang2024no, shi2019nash} and multicluster games, where multiple networks compete simultaneously \cite{ye2018nash1, ye2019unified, ye2017distributed1, zeng2019generalized, meng2023linear, zimmermann2021solving, pang2023distributed, zhou2022distributed, chen2023generalized, zhu2025accelerated, wang2025gradient, yu2023distributed, yu2024distributed}.

For two-network zero-sum games, many works \cite{gharesifard2013distributed, lou2015nash, huang2024no} assume a bipartite network without isolated nodes for cross-team information exchange. Specifically, the work \cite{gharesifard2013distributed} analyzes the convergence of a distributed gradient dynamics, while the work \cite{lou2015nash} develops a distributed subgradient-based algorithm for time-varying graphs. The work \cite{huang2024no} introduces a distributed no-regret mirror descent method that also achieves convergence. To relax the strong cross-team subgraph assumption, the work \cite{shi2019nash} 
 proposes an incremental strategy requiring only a leader-to-leader connection. However, it assumes each agent knows its network's Laplacian and introduces additional consensus steps per iteration. Note that all these algorithms employ 
 decaying step size to achieve convergence. Beyond the zero-sum setting, the work \cite{vlaski2021competing} addresses the same stochastic competing networks problem as ours, and develops a diffusion learning algorithm utilizing a constant step size. However, it focuses on consensus analysis without providing convergence arguments.

For multicluster games, most existing studies focus on continuous-time algorithms \cite{ ye2018nash1, ye2019unified, ye2017distributed1, zeng2019generalized}. In fact, discrete-time algorithms are more interesting, as practical data-driven systems typically operate with discrete input-output dynamics.
The work \cite{meng2023linear} introduces a distributed gradient-tracking algorithm with linear convergence. 
The main drawback of this work 
is that 
it restricts the cross-team interactions to team leaders.
The work \cite{zimmermann2021solving}
considers a directed graph and 
employs identical constant step size for all agents,
while the work \cite{pang2023distributed} 
allows agents to adopt a local learning rate.
Additionally, the work \cite{zhou2022distributed} modifies the local objectives to the one that depends on all agents' actions, which introduces increased memory and communication overhead.
Some other works explore time-varying graphs \cite{zhu2025accelerated, wang2025gradient}, with time-varying objectives \cite{yu2023distributed, yu2024distributed}, or coupled constraints \cite{chen2023generalized}. 
Note that the aforementioned works \cite{meng2023linear, chen2023generalized, yu2023distributed} assume undirected graphs for communication. Although some studies address directed graphs \cite{zimmermann2021solving, pang2023distributed, zhou2022distributed,zhu2025accelerated, wang2025gradient, yu2024distributed}, they  use additional right-stochastic matrices for within-team cooperation, which require additional steps \cite{zimmermann2021solving}. Furthermore, nearly all these works \cite{zimmermann2021solving, pang2023distributed, zhou2022distributed, chen2023generalized, zhu2025accelerated, wang2025gradient, yu2023distributed, yu2024distributed} require the local cost to be convex, which limits their applicability to more complex scenarios. Moreover, all these works focus on {\em deterministic settings} excluding the online scenario that data is continuously streaming in.

\subsection{Contributions}
The contributions of this work are summarized as follows: 
\begin{enumerate}
    \item 
We propose novel diffusion learning strategies to solve stochastic two-network game problems and analyze their convergence.
Our work relaxes the strong assumption of a bipartite graph without isolated nodes for the cross-team graph structure, as considered in previous studies (e.g. \cite{gharesifard2013distributed, lou2015nash, huang2024no}). Unlike existing methods, our algorithm employs a constant step size, enabling agents to continuously adapt toward an equilibrium state.
Furthermore, simulation results illustrate that the proposed methods achieve faster convergence and offer an improved stability range for the step size compared to the baseline \cite{vlaski2021competing}.

    

    \item 
  Our algorithm is simple and efficient to implement, as it requires only left-stochastic matrices for within-team communication, avoiding the additional step of constructing right-stochastic matrices \cite{zimmermann2021solving, pang2023distributed, zhou2022distributed, zhu2025accelerated, wang2025gradient, yu2024distributed}. Moreover, we do not impose convexity on local objective functions and instead assume that the global gradient operator satisfies strong monotonicity.

    \item 
We apply our algorithms to important practical applications, including the Cournot team competition in economics and decentralized GAN training. Simulation results confirm that our proposed algorithms outperform existing baselines.

\end{enumerate}
\vspace{0.3em}

\begin{table}[!htbp]  
    \centering
    \captionsetup{labelformat=empty}
    \caption{{\bf Notation:} The table below summarizes some of the key symbols in this article.}

    \begin{tabular}{|c||c|} 
    
    \hline
    $K_t$ & Number of agents in Team $t$, $K= K_1 + K_2$ \\\hline
    
    $\mathcal{N}^{(t)}$ &
    $\mathcal{N}^{(1)} \triangleq \{1, \ldots, K_1\}$, $\mathcal{N}^{(2)} \triangleq \{K_1, \ldots, K\}$
    \\\hline
$\mathcal{N}$ & $\mathcal{N} \triangleq \{\mathcal{N}^{(1)}, \mathcal{N}^{(2)}\}$
 \\
 \hline
    $J^{(t)} (\cdot)$ & Global risk function of Team/network $t$ \\\hline

    $J_k (\cdot;\cdot), Q_k(\cdot;\cdot)$ & Local risk, and loss functions of agent $k$ \\\hline

    $\boldsymbol{\xi}_k$ & Random sample of agent $k$ \\
    \hline
   $\mathbb{E}[\cdot]$ & Expectation operator
    \\ \hline 

$\otimes$ & Kronecker product\\
\hline
    \end{tabular}
    
    \label{tab:notations}
\end{table}



\section{Network Model and Algorithms}



Most distributed optimization methods rely on a single connected graph to exchange local estimates.  However, in competitive environments, agents must coordinate their actions within the team while also processing adversarial inputs to respond effectively to their opponents.
To understand how information flows in such settings, we introduce the network topology and the combination and inference matrices, which define communication links and their strengths for within-team sharing and cross-team estimation.
Building on the network structure, we develop diffusion learning algorithms that allow agents to continuously adapt their strategies, and ultimately reach an equilibrium state.

\subsection{Network Model}

\begin{figure*}[htbp]
\centering
\begin{subfigure}[b]{0.30\textwidth}
\centering
\includegraphics[width=\textwidth]{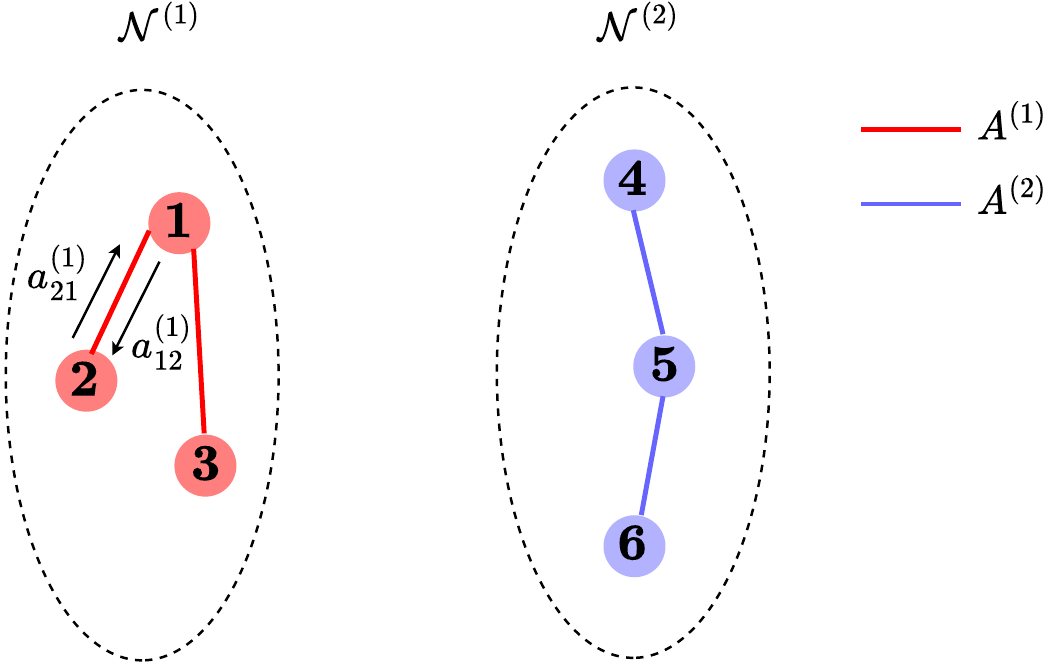}
\caption{Within-team subgraphs.}
\label{fig:subfig1}
\end{subfigure}
\hfill
\begin{subfigure}[b]{0.30\textwidth}
\centering
\includegraphics[width=\textwidth]{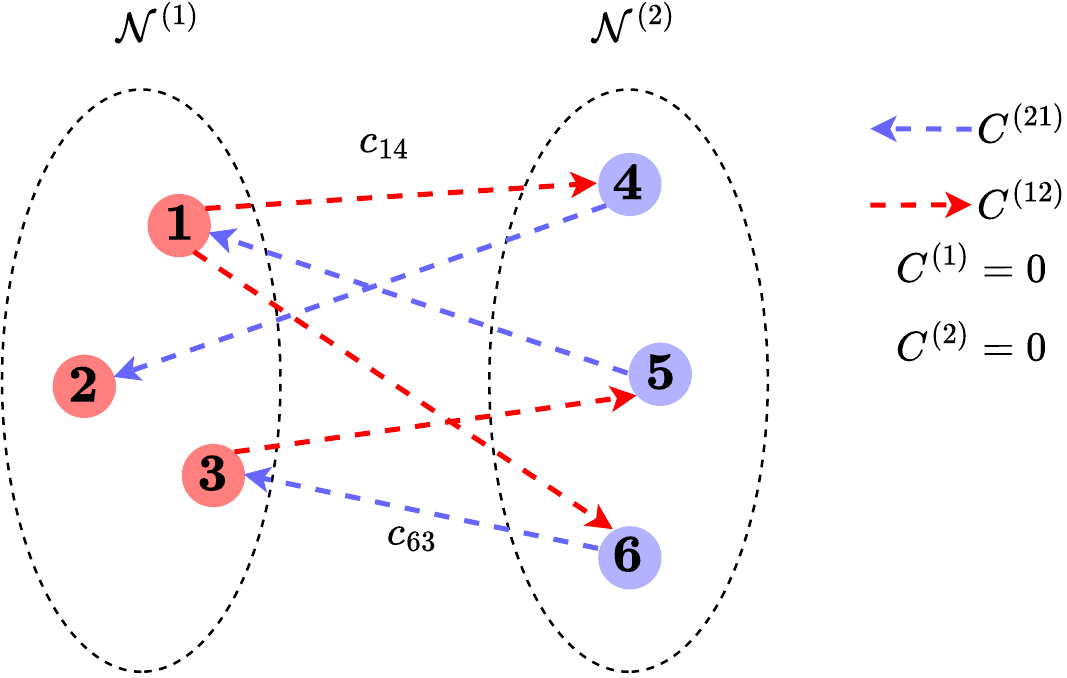}
\caption{Strong cross-team subgraph.}
\label{fig:subfig2}
\end{subfigure}
\hfill
\begin{subfigure}[b]{0.30\textwidth}
\centering
\includegraphics[width=\textwidth]{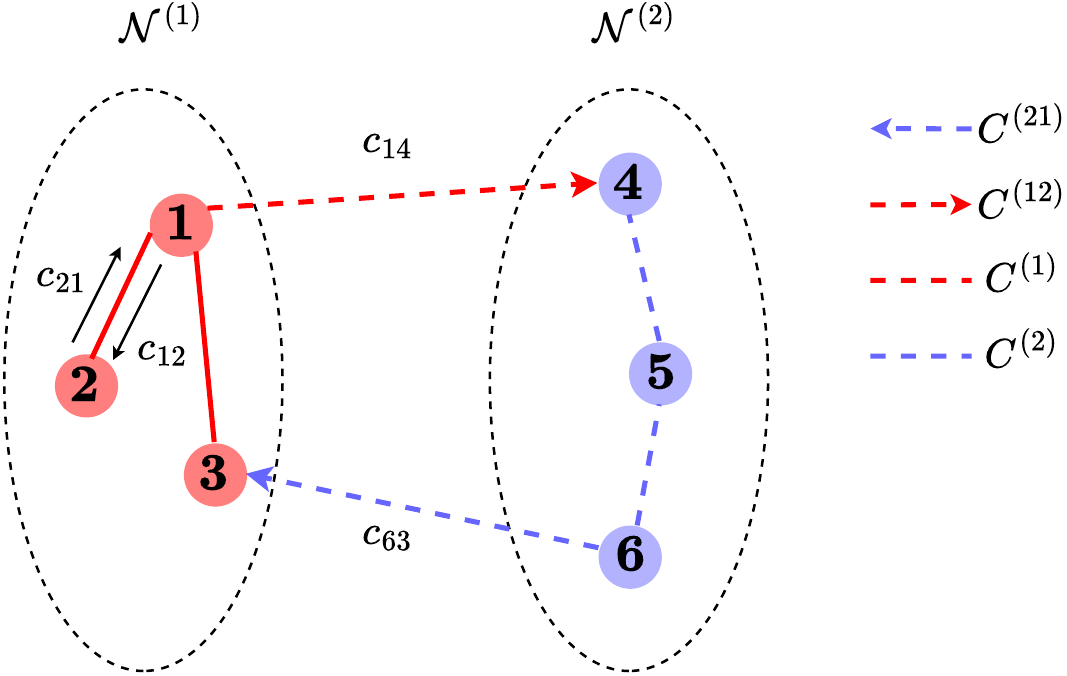}
\caption{Weak cross-team subgraph.}
\label{fig:subfig3}
\end{subfigure}
\caption{Illustration of within-team subgraphs, cross-team subgraphs, and associated combination and {\color{black} inference} matrices.}
\label{fig:main}
\end{figure*}

 Following standard graph theory notation  \cite{sayed2022inference1, godsil2013algebraic},  
we represent the interactions among all game agents as a graph $\mathcal{G} = ({\mathcal{N}, \mathcal{E}}$), where $\mathcal{E} \subseteq \mathcal{N} \times \mathcal{N}$ is the set of edges representing communication links. 
This $\mathcal{G}$ consists of the following subgraphs:
\begin{itemize}
\item \textbf{Within-team subgraphs} $\mathcal{G}^{(1)}$ and $\mathcal{G}^{(2)}$: Each team communicates internally through its respective subgraph, defined as:
\begin{equation}
\mathcal{G}^{(t)} = \{ \mathcal{N}^{(t)}, \mathcal{E}^{(t)} \}, \quad t \in {1,2}
\end{equation}
where $\mathcal{N}^{(t)}$ is the set of agents in Team $t$ and $\mathcal{E}^{(t)}$ defines the communication links within the team. 
\item \textbf{Cross-team subgraph} $\mathcal{G}^{(c)}$: The interactions between the two teams are modeled by the subgraph:
\begin{equation}
    \mathcal{G}^{(c)} = \{ \mathcal{N}^{(1)} \cup \mathcal{N}^{(2)}, \mathcal{E}^{(c)} \}
\end{equation}
where $\mathcal{E}^{(c)} \subseteq (\mathcal{N}^{(1)} \times \mathcal{N}^{(2)}) \cup (\mathcal{N}^{(2)} \times \mathcal{N}^{(1)})$ represents  inter-team communication links for sharing the adversarial information.
\end{itemize}

For within-team cooperation, agents in Team $t$ share information with teammates through $\mathcal{G}^{(t)}$, which is characterized by combination matrix $A^{(t)} \in \mathbb{R}^{K_t \times K_t}$. For $A^{(1)}$, the entry $a_{lk}^{(1)} \geq 0$ represents the weight scaling the information flowing from agent $l$ to agent $k$ within the team, ensuring proper mixing of information. $A^{(2)}$ is defined similarly for Team 2. An illustration of the within-team subgraphs and the combination matrices is shown in Figure~\ref{fig:subfig1}.

In addition to sharing cooperative strategies within their respective teams through the subgraphs $\mathcal{G}^{(1)}$ and $\mathcal{G}^{(2)}$,
agents must also infer the adversary strategy. As this process could rely on the whole graph $\mathcal{G}$, we 
introduce the following inference matrix $C \in \mathbb{R}^{K \times K}$:

\begin{equation}
    C \triangleq \begin{bmatrix}
       C^{(1)}  & C^{(12)} \\
       C^{(21)}  & C^{(2)}
    \end{bmatrix}
\end{equation}
where $C^{(12)}$ and $C^{(21)}$ account for direct cross-team communication links from Team 1 to Team 2 and from Team 2 to Team 1, while $C^{(1)}$ and $C^{(2)}$ correspond to internal links within each team used to refine cross-team information, as the discussion will reveal.



The cross-team subgraph 
$\mathcal{G}^{(c)}$
  is crucial in competitive settings, as it enables interactions between opposing teams. Without such interconnections, competition cannot take place. In this work, we examine both a special case and a general case of $\mathcal{G}^{(c)}$.
\begin{definition}[{\bf Strong cross-team subgraph}] \label{def:strong_def}
A cross-team subgraph $\mathcal{G}^{(c)}$ is said to be strong if every agent in one team $t$ has at least one incoming link from an agent in the other team $t'$, i.e., every $k \in \mathcal{N}^{(t)}$ has at least one incoming link from $\ell \in \mathcal{N}^{(t^{\prime})}$, so that $(\ell,k) \in \mathcal{E}^{(c)}$.
\end{definition}
The above structure is also referred to as a bipartite graph without isolated nodes \cite{lou2015nash}, ensuring that every agent has at least one neighboring opponent. In a strong cross-team subgraph, each agent can directly observe their adversaries’ strategies through cross-team connections. Since direct observations are more reliable than indirect estimates obtained through teammates, the need for internal information refinement diminishes. As a result, the matrices capturing within-team communication,
$C^{(1)}, C^{(2)}$, become unnecessary and are set to
 $C^{(1)} = 0$ and $C^{(2)} = 0$. We illustrate the strong cross-team subgraph and the associated inference matrix in Figure \ref{fig:subfig2}.
\begin{definition}[{\bf Weak cross-team subgraph}] \label{def:weak_def}
 A cross-team subgraph $\mathcal{G}^{(c)}$ is is said to be weak if each team $t$ has at least one incoming link from the other team $t^{\prime}$, i.e., 
$\mathcal{E}^{(c)} \cap (\mathcal{N}^{(t)} \times \mathcal{N}^{(t^{\prime})}) \neq \emptyset$. 
\end{definition}

A weak cross-team subgraph requires at least one directed edge from Team \(1\) to Team \(2\) and at least one directed edge in the other direction from Team \(2\) to Team \(1\), ensuring minimal cross-team connectivity.
 In this setting, only a subset of agents has direct access to the opponent's strategy, making matrices $C^{(12)}$ and $C^{(21)}$
insufficient for fully observing the opponent’s actions. 
To compensate, teammates must communicate internally within their teams via
$C^{(1)}$ and $C^{(2)}$  $(C^{(1)}\not =0, C^{(2)} \not =0)$
to learn about the opponent's actions.
We illustrate the weak cross-team subgraph and the associated  inference matrix in Figure \ref{fig:subfig3}.

\subsection{Algorithm Development}
In this section, we develop algorithms to address two important classes of the network game. Our method  aims to relax the reliance on the strong cross-team subgraph structure, a common assumption in previous two-network zero-sum frameworks \cite{lou2015nash, gharesifard2013distributed, huang2024no}. Additionally, we extend the applicability of our methods by generalizing the cost function so that we can also handle the situation
$ J^{(1)}\left(x,y\right) \not= - J^{(2)}\left(x,y\right)$. 
A key challenge in these settings is designing a framework that integrates learning algorithms while respecting network constraints on information flow.  Our proposed algorithms build upon the Adapt-then-Combine (ATC) diffusion learning framework \cite{sayed2014adaptation}, where each node first updates its model using local data and then shares these updates with neighboring nodes.

\textbf{Zero-sum network-game under a weak cross-team subgraph}:
For a zero-sum game satisfying
\begin{equation}
\label{eq:zero-sum}
    J^{(1)}\left(x,y\right) = - J^{(2)}\left(x,y\right)
\end{equation}
the learning process can be divided into two phases: within-team diffusion learning and cross-team diffusion learning.
For instance, let us consider the update procedure of a single agent $k \in \mathcal{N}^{(1)}$ in Team 1 to describe this process. At the start of each iteration $i$, agent $k$ uses  
its past strategy (or iterate)
$\boldsymbol{x}_{k, i-1}$
and the available estimated adversarial strategy (or iterate)
$\boldsymbol{y}_{k, i-1}$ to carry out an update based on the following ATC strategy:
\begin{align}
\boldsymbol{\phi}_{k, i} &= \boldsymbol{x}_{k, i-1}- \mu \widehat{\nabla_x J}_k(\boldsymbol{x}_{k, i-1}, \boldsymbol{y}_{k, i-1}) \quad \text{(Adaptation)}
\label{eq:within_phi}
\\
\boldsymbol{x}_{k, i} &= \sum_{\ell \in \mathcal{N}^{(1)}} a_{\ell k}^{(1)} \boldsymbol{\phi}_{ \ell, i}
  \quad  \hspace{3em} \text{(Within-team diffusion)}  \label{eq:within_x}
\end{align}
Importantly, agent $k$ also needs to update its estimate of Team 2’s information (i.e., to update $\boldsymbol{y}_{k,i-1}$ to $\boldsymbol{y}_{k,i}$). 
A straightforward approach would be to passively receive data transmitted through the cross-team subgraph ${\cal G}^{(c)}$. 
However, due to the weak cross-team topology, this information may need to traverse a multi-hop communication path before reaching agent $k$ and  become outdated, leading to inaccurate observation.
To mitigate the impact of this delay, 
 we exploit the structure of the zero-sum formulation.
When  
$\nabla_y J_k \left(x, y\right)$ is close to $\nabla_y J^{(1)} \left(x, y\right)$ for a positive
constant $c$ (see condition \eqref{eq:bdis} further ahead):
\begin{equation}
    \label{eq:alg_bdis}
    \left\| \nabla_y J_k \left(x, y\right) - \nabla_y J^{(1)} \left(x, y\right)\right\| \le c
\end{equation}
We can deduce from \eqref{eq:zero-sum} that  
\begin{equation}
    \begin{aligned}
        &\quad \left\| \nabla_y J_k \left(x, y\right) + \nabla_y J^{(2)} \left(x, y\right)\right\| \le c
    \end{aligned}
\end{equation}
The above relation suggests that the local gradient $- \nabla_y J_k \left(x, y\right)$ can approximate $\nabla_y J^{(2)} \left(x, y\right)$.
This enables agent $k$ to infer the adversarial strategy through a hybrid approach that combines direct observations with estimates obtained via an additional estimation step added to \eqref{eq:within_phi}--\eqref{eq:within_x}:
\begin{align}
\!\!\!\!\boldsymbol{\psi}_{k,i} &= \boldsymbol{y}_{k,i-1} + \mu \widehat{\nabla_y J}_k(\boldsymbol{x}_{k, i-1}, \boldsymbol{y}_{k, i-1}) \quad \hspace{0.5em}\text{(Inference)} \label{eq:psi_update}\\
\!\!\!\!\boldsymbol{y}_{k,i} &= \!\!
\sum_{\ell \in \mathcal{N}^{(2)}} \!\!\!c_{\ell k} \boldsymbol{y}_{\ell,i} + \!\!\!\!\sum_{\ell \in \mathcal{N}^{(1)}} \!\!\!c_{\ell k} \boldsymbol{\psi}_{\ell,i} \hspace{0.3em}\text{(Cross-team diffusion)}
\label{eq:cross_diffusion}
\end{align}
Note that the first term in \eqref{eq:cross_diffusion} vanishes if agent $k$ does not have an incoming link from any node $\ell \in \mathcal{N}^{(2)}$, meaning that  $\boldsymbol{y}_{k,i}$ is  entirely determined by the estimated information from its neighbors. According to weakly-connected graph theory \cite{ying2016information, salami2017social}, a network that receives information from another network without reciprocal links will, in the long run, be dominated by the sending network. However, the update step \eqref{eq:psi_update} is particularly useful for accelerating convergence in the transient stage when the adversarial information has not yet fully propagated in Team 1, and the gradient calculation in \eqref{eq:psi_update} can be removed after a sufficient number of iterations to reduce the computational complexity. Inspired by this combination strategy, we name the proposed algorithm  
\textbf{A}dapt-then-\textbf{C}ombine and {\bf I}nfer-then-\textbf{C}ombine (\textbf{ATC-ITC})
algorithm.


To summarize the above steps concisely, we introduce the following network notation for Team 1:
\begin{align}
\boldsymbol{\mathcal{X}}^1_{i} 
&\triangleq \mbox{col} \{\boldsymbol{x}_{1,i}, \ldots, \boldsymbol{x}_{K_1,i} \} \in  \mathbb{R}^{K_1 M_1 \times 1} \label{eq:X1i}
\\
\boldsymbol{\mathcal{Y}}^1_{i} 
&\triangleq \mbox{col} \{\boldsymbol{y}_{1,i}, \ldots, \boldsymbol{y}_{K_1,i} \} \in  \mathbb{R}^{K_1 M_2 \times 1}\label{eq:Y1i}
\end{align}
and $\boldsymbol{\mathcal{Y}}^2_{i},\boldsymbol{\mathcal{X}}^2_{i}$
are defined similarly for Team 2. 
Here, we use a superscript to distinguish the network quantities associated with Teams 1 and 2.
Therefore,
$\boldsymbol{\mathcal{Y}}^1_{i}$
and 
$\boldsymbol{\mathcal{X}}^2_{i}$
represent inferred adversarial information by Teams 1 and 2, respectively.
In addition, the networked stochastic gradients are defined as follows:
\begin{align}
    \boldsymbol{\mathcal{G}}^1_{x, i} &\triangleq \mbox{col} \Big\{\widehat{\nabla_x J}_k(\boldsymbol{x}_{k, i-1}, \boldsymbol{y}_{k, i-1}) \vphantom{\sum} \Big\}{}_{\vphantom{A^A} k=1}^{K_1}\label{eq:G1xi} \\
    \boldsymbol{\mathcal{G}}^1_{y, i} &\triangleq \mbox{col} \Big\{- \widehat{\nabla_y J}_k(\boldsymbol{x}_{k, i-1}, \boldsymbol{y}_{k, i-1})\Big\}{}_{\vphantom{A^A} k=1}^{K_1}\label{eq:G1yi}
\end{align} 
where $\widehat{\nabla_x J}_k (\cdot, \cdot)$ and $\widehat{\nabla_y J}_k (\cdot, \cdot)$ are stochastic gradients associated with $x$ and $y$ at agent $k$;
$\boldsymbol{\mathcal{G}}^2_{x, i}$ and $\boldsymbol{\mathcal{G}}^2_{y, i}$ are defined similarly. 
Finally, the augmented combination matrices are defined as follows:
\begin{align}
    \mathcal{C}^{(1)} &= C^{(1)} \otimes I_{M_2}, \quad 
    &\mathcal{C}^{(2)} &= C^{(2)} \otimes I_{M_1} \notag \\
    \mathcal{C}^{(21)} &= C^{(21)} \otimes I_{M_2}, \quad 
    &\mathcal{C}^{(12)} &= C^{(12)} \otimes I_{M_1} \notag \\    \mathcal{A}^{(1)} &= A^{(1)} \otimes I_{M_1}, \quad 
    &\mathcal{A}^{(2)} &= A^{(2)} \otimes I_{M_2}
\end{align}
With the notation so defined, we summarize the preceding procedures in \textbf{Algorithm \ref{alg:network_learning_zerosum}}.

\begin{algorithm}
\caption{
\textbf{A}dapt-then-\textbf{C}ombine and \textbf{I}nfer-then-\textbf{C}ombine (\textbf{ATC-ITC}) algorithm}  \label{alg:network_learning_zerosum}
\begin{algorithmic}[1]
\STATE \textbf{Initialize:} strategies $\boldsymbol{\mathcal{X}}^1_{-1},  \boldsymbol{\mathcal{Y}}^1_{-1},  \boldsymbol{\mathcal{X}}^2_{-1}, \boldsymbol{\mathcal{Y}}^2_{-1} \leftarrow 0$, step size $\mu$

\FOR{$i = 0,\cdots$}
\STATE \underline{\texttt{Within-team adapt-then-combine} }\\$\boldsymbol{\mathcal{X}}^1_{i} = {\mathcal{A}^{(1)}}^{\top} \left(\boldsymbol{\mathcal{X}}^1_{i-1} - \mu \boldsymbol{\mathcal{G}}^1_{x, i} \right)$\\ $\boldsymbol{\mathcal{Y}}^2_{i} = {\mathcal{A}^{(2)}}^{\top} (\boldsymbol{\mathcal{Y}}^2_{i-1} - \mu \boldsymbol{\mathcal{G}}^2_{y, i} )$

\STATE \underline{\texttt{Cross-team infer-then-combine}}

\STATE $\boldsymbol{\mathcal{\psi}}^1_{i} = \boldsymbol{\mathcal{Y}}^1_{i-1} - \mu \boldsymbol{\mathcal{G}}^1_{y, i}$
\vspace{0.2em}
\STATE $\boldsymbol{\mathcal{\psi}}^2_{i} = \boldsymbol{\mathcal{X}}^2_{i-1} - \mu \boldsymbol{\mathcal{G}}^2_{x, i}$

\STATE $\boldsymbol{\mathcal{Y}}^1_{i} = {\mathcal{C}^{(21)}}^{\top} \boldsymbol{\mathcal{Y}}^2_{i} + {\mathcal{C}^{(1)}}^{\top} \boldsymbol{\mathcal{\psi}}^1_{i}$

\STATE$\boldsymbol{\mathcal{X}}^2_{i} = {\mathcal{C}^{(12)}}^{\top} \boldsymbol{\mathcal{X}}_{i}^1 + {\mathcal{C}^{(2)}}^{\top} \boldsymbol{\mathcal{\psi}}^2_{i}$
\ENDFOR
\end{algorithmic}
\end{algorithm}

\textbf{General network-game under a strong cross-team subgraph}:
For a general game $J^{(1)} \left(x, y\right) \not = J^{(2)} \left(x, y\right)$,
the estimation  procedure in 
recursion
\eqref{eq:psi_update} may not yield accurate results since $\left\| \nabla_y J_k \left(x, y\right) + \nabla_y J^{(2)} \left(x, y\right)\right\|$
can be large.
However,  competition is still possible when a strong cross-team subgraph is present. In this scenario, agents can observe the adversary strategy by directly combining neighboring opponents' information. Since now $C^{(1)} = C^{(2)} = 0$, agent $k$ fuses Team 2's information as follows:
\begin{align}
\boldsymbol{y}_{k, i} = \sum_{\ell \in \mathcal{N}^{(2)}} c_{\ell k}  \boldsymbol{y}_{\ell, i}
\label{cross_tema_learn}
\end{align} 
The above step is the only change compared to \textbf{Algorithm 1}.
Using the notation from 
\eqref{eq:X1i}-\eqref{eq:G1xi} for Team 1 and the same notation for Team 2,
we summarize the new procedure 
in \textbf{Algorithm 2}
and we name it 
\textbf{A}dapt-then-\textbf{C}ombine and \textbf{C}ombine (\textbf{ATC-C}).
In this setup, each agent accesses the latest adversarial information and computes the stochastic gradient only once per iteration, reducing computational costs and improving efficiency.


\begin{algorithm}
\caption{\textbf{A}dapt-then-\textbf{C}ombine and \textbf{C}ombine (\textbf{ATC-C})} \label{alg:network_learning_general_strong}
\begin{algorithmic}[1]
\STATE \textbf{Initialize:} $i$=0, actions $\boldsymbol{\mathcal{X}}_{-1}, \boldsymbol{\mathcal{Y}}_{-1} \leftarrow 0$, step size $\mu$

\WHILE{not done}
    \STATE \underline{\texttt{Within-team adapt-then-combine}} \\$\boldsymbol{\mathcal{X}}^1_{i} = {\mathcal{A}^{(1)}}^{\top} (\boldsymbol{\mathcal{X}}^1_{i-1} - \mu \boldsymbol{\mathcal{G}}^1_{x, i})$\\ $\boldsymbol{\mathcal{Y}}^2_{i} = {\mathcal{A}^{(2)}}^{\top} (\boldsymbol{\mathcal{Y}}^2_{i-1} - \mu \boldsymbol{\mathcal{G}}^2_{y, i})$
    
    \STATE \underline{\texttt{Cross-team combine}}\\ $\boldsymbol{\mathcal{Y}}^1_{i} = {\mathcal{C}^{(21)}}^{\top} \boldsymbol{\mathcal{Y}}^2_{i}$\\
    $\boldsymbol{\mathcal{X}}^2_{i} = {\mathcal{C}^{(12)}}^{\top} \boldsymbol{\mathcal{X}}^1_{i}$
    \STATE $i \gets i + 1$
\ENDWHILE
\end{algorithmic}
\end{algorithm}

\begin{remark*}
     Both {\bf ATC-ITC} and {\bf ATC-C} can solve zero-sum network-games under strong cross-team subgraphs since every strong cross-team subgraph is also a weak one, and general network games encompass zero-sum network games. General network-games under weak cross-team subgraphs can be addressed by {\bf ATC-ITC}, assuming partially observable stochastic gradients (see Assumption~\ref{ass:add_grad} and Corollary~\ref{corollary:ATC-ETC}). 
\end{remark*}


\section{Convergence Analysis}
In this section, we conduct a stability performance analysis for the proposed algorithms. To support the theoretical analysis, we introduce several mild assumptions that are commonly used in the literature on distributed learning. Our objective is to provide convergence guarantees, ensuring that the algorithms converge to a Nash equilibrium at which both networks achieve an equilibrium state.
The notion of a Nash equilibrium is defined as follows.
 \begin{definition}[\textbf{Nash equilibrium}]
  A point $(x^\star, y^\star)$ is a Nash equilibrium if the following condition holds
\begin{equation}
    J^{(1)}(x^\star, y^\star) \leq J^{(1)}(x, y^\star), \; J^{(2)}(x^\star, y^\star) \leq J^{(2)}(x^\star, y)
\end{equation}  
\end{definition}
\subsection{Assumptions}

We introduce assumptions regarding the structural properties of the combination matrices $A^{(1)}$, $A^{(2)}$ and the  inference matrix $C$, and conditions related to the cost function.

\begin{assumption}[\textbf{Within-team combination matrix}] 
\label{ass:within_conn}
    For $t \in \{1,2\}$, the combination matrix $A^{(t)} \in \mathbb{R}^{K_t \times K_t}$ is primitive and left-stochastic, i.e. $\mathbbm{1}^{\top} A^{(t)} = \mathbbm{1}^{\top}$. \hfill\qed
\end{assumption}
The above assumption is standard in distributed optimization to guarantee consensus among agents\cite{sayed2022inference1}.
According to \cite{sayed2014adaptation}, 
$A^{(1)}$ and $A^{(2)}$ have Perron eigenvectors $p^{(1)}$ and $p^{(2)}$, with positive entries that satisfy $A^{(1)}p^{(1)} = p^{(1)}$,  $A^{(2)}p^{(2)} = p^{(2)}$, and $\mathbbm{1}_{K_1}^{\top} p^{(1)}= \mathbbm{1}_{K_2}^{\top} p^{(2)} = 1$. We denote 
\begin{equation}
    p \triangleq \left[\begin{array}{c}p^{(1)}  \\p^{(2)} \end{array}\right]
\end{equation}
where the $\{p_k\}_{k=1}^{K_1}$ and $\{p_k\}_{k=K_1 +1}^{K}$ correspond to the coefficients used in \eqref{eq:global_game_ya} and \eqref{eq:global_game_yb}, respectively.


For the inference matrix $C$, we consider the following assumption.

\begin{assumption}[\textbf{Cross-team  inference matrix}] \label{ass:cross-conn}
The inference matrix $C$ is partitioned into blocks  $
[C^{(1)}, C^{(12)}; C^{(21)}, C^{(2)}]$
and satisfies either one of the following conditions depending on whether we are dealing with a weak (case (b)) or strong (case (a)) cross-team subgraph:  
\begin{enumerate}[label=(\alph*), ref=\theassumption\alph*]  
    \item \label{ass:cross-conn_strong} For every $t \in \{1,2\}$, the diagonal blocks satisfy $C^{(t)} = 0$. Moreover, for any $t, t^{\prime} \in \{1,2\}, t \neq t^{\prime}$, the off-diagonal block
    $C^{(t^{\prime} t)} \in \mathbb{R}^{K_{t^{\prime}} \times K_t}$ is left-stochastic, i.e.,  $
    \mathbbm{1}^{\top} C^{(t^{\prime} t)} = \mathbbm{1}^{\top}.$
    
    \item \label{ass:cross-conn_weak} $C$ is left-stochastic, and for every $t \in \{1,2\}$, the diagonal block matrices $C^{(t)}$ are irreducible. Moreover, for $t, t^{\prime} \in \{1,2\}, t \neq t^{\prime}$, there exists at least one $k \in \mathcal{N}^{(t)}$ and $\ell \in \mathcal{N}^{(t^{\prime})}$ such that $c_{\ell k} > 0$.
\end{enumerate}
\hfill\qed
\end{assumption}

Assumption \ref{ass:cross-conn_strong} corresponds to the strong cross-team subgraph case discussed earlier and is widely adopted in two-network zero-sum literature \cite{lou2015nash, talebi2019distributed}. This assumption ensures that every agent within a team receives direct incoming links from at least some opponent agents. While somewhat restrictive, it plays a crucial role in facilitating cross-team adversarial information observation and proving the convergence of \textbf{Algorithm 2}.
Furthermore, we introduce Assumption \ref{ass:cross-conn_weak}, requiring only that at least one node per team directly accesses information from the opponent agent. Meanwhile, the block matrices
$C^{(1)}$
and $C^{(2)}$
must be 
irreducible, which implies that the underlying graphs are connected \cite{meyer2023matrix};  otherwise, some nodes would remain isolated and unable to participate in the network game.

For convenience of analysis,
we define the following concatenated {\color{black} gradient} vectors:
\begin{align} 
    F(z) &\triangleq \left[\begin{array}{c}
    \nabla_x J^{(1)} \left(x, y\right) \\
    \nabla_y J^{(2)} \left(x, y\right) 
\end{array}\right],
\text{ where }     z \triangleq \left[\begin{array}{c} x \\ y \end{array}\right]\label{eq:Fz_defn}
\end{align}
This gradient mapping $F(z)$ captures the gradients of each team's global cost function with respect to its own strategy, and is commonly used in the context of game theory \cite{zimmermann2021solving, pang2023distributed}.
\begin{assumption}[\textbf{Strong monotonicity}]
\label{ass:strong mono}
    The global gradient operator \( F(z) \) is \(\nu\)-strongly monotone, i.e., for every $z_1, z_2 \in \mathbb{R}^{M}$, we have 
    \begin{equation}
        \left(F\left(z_1\right) - F\left(z_2\right)\right)^\top \left(z_1 - z_2\right) \geq \nu \|z_1 - z_2\|^2
        \label{strongmonotone}
    \end{equation}
for a positive constant $\nu$.
\hfill\qed%
\end{assumption}
Assumption \ref{ass:strong mono} is commonly used in the context of the variational inequality \cite{komlosi1999stampacchia, ryu2016primer}.  Condition \eqref{strongmonotone} generalizes the setup of
strongly-convex strongly-concave minmax problems.

\begin{assumption}[\textbf{Lipschitz gradients}]
\label{ass:lip}
    For each $t \in \{1,2\}$ and $k \in \mathcal{N}^{(t)}$, we assume the gradients associated with each local risk function $J_k (\cdot, \cdot)$ are $L_f$-Lipschitz, i.e, for any $x_1, x_2 \in \mathbb{R}^{M_1}, y_1, y_2 \in \mathbb{R}^{M_2}$:
    \begin{equation}
    \label{eq:lip}
        \begin{aligned}
            &\| \nabla_w J_k (x_1, y_1) - \nabla_w J_k (x_2, y_2)\|\\
            & \leq L_f \left(\|x_1 - x_2\| + \|y_1 - y_2\|\right),
        \end{aligned}
    \end{equation}
    where  $w = x \text{ or } y$.\hfill\qed
\end{assumption}

\begin{assumption}[\textbf{Bounded gradient disagreement}]
\label{ass:bdis}
    For each $t \in \{1,2\}$ \textcolor{black}{and $k \in \mathcal{N}^{(t)}$,} the gradient disagreement between the local risk functions and the global risk function is bounded, i.e., for any $x_1 \in \mathbb{R}^{M_1}, y_1 \in \mathbb{R}^{M_2}$:
    \begin{equation}
    \label{eq:bdis}
        \| \nabla_w J_k (x_1, y_1) - \nabla_w J^{(t)} (x_1, y_1)\| \leq G
    \end{equation}
    where  $w = x \text{ or } y$.\hfill\qed
\end{assumption}

\begin{assumption}[\textbf{Gradient noise process}]
\label{ass:gradnoise}
    We define the filtration generated by the random processes as $\boldsymbol{\mathcal{F}}_i = \{(\boldsymbol{x}_{k, j}, \boldsymbol{y}_{k, j}) \mid k = 1, \dots, K, j = -1,\dots, i \}.$ For each $t \in \{1,2\}$ and $k \in \mathcal{N}^{(t)}$, we assume the stochastic gradients are unbiased with bounded variance conditioned on $\boldsymbol{\mathcal{F}}_i$, i.e., for any $\boldsymbol{x}, \boldsymbol{y} \in \boldsymbol{\mathcal{F}}_i$, 
    \textcolor{black}{
    \begin{equation}
    \label{eq:unbiased}
        \mathbb{E} \{ \widehat{\nabla_w J}_k (\boldsymbol{x}, \boldsymbol{y}) \mid \boldsymbol{\mathcal{F}}_i\} = \nabla_w J_k (\boldsymbol{x}, \boldsymbol{y})
    \end{equation}
    \begin{equation}
    \label{eq:b_var}
        \mathbb{E} \{ \| \widehat{\nabla_w J}_k (\boldsymbol{x}, \boldsymbol{y}) - \nabla_w J_k (\boldsymbol{x}, \boldsymbol{y}) \|^2 \mid \boldsymbol{\mathcal{F}}_i \} \leq \sigma^2
    \end{equation}
    }
    where  $w = x \text{ or } y$.\hfill\qed
\end{assumption}
Note that similar gradient assumptions are also adopted in distributed stochastic optimization works such as \cite{vlaski2021distributed, cai2024diffusion}.

\subsection{Main Results}

Under these assumptions, the proposed algorithms are proven to converge to a Nash equilibrium. We first verify the existence and uniqueness of the Nash equilibrium by leveraging results from \cite{facchinei2003finite}.
\begin{lemma}({\bf \cite[Proposition 1.4.2, Theorem 2.3.3]{facchinei2003finite}})
Under Assumption \ref{ass:strong mono}, there exists a unique Nash equilibrium $z^\star = [x^\star;y^\star]$ for problem \eqref{eq:global_game_ya}-\eqref{eq:global_game_yb}.
\end{lemma}

Before presenting the main results, we introduce the following notation: 
\begin{align}
    \boldsymbol{x}_{c,i} &\triangleq  \sum_{k \in \mathcal{N}^{(1)}} p_k \boldsymbol{x}_{k,i}, \quad &\boldsymbol{y}_{c,i} &\triangleq  \sum_{k \in \mathcal{N}^{(2)}} p_k \boldsymbol{y}_{k,i}\label{eq:centroid_x}  \\ \boldsymbol{\mathcal{X}}_{c,i} &\triangleq \mathbbm{1}_{K_1} \otimes \boldsymbol{x}_{c,i}, \quad &\boldsymbol{\mathcal{Y}}_{c,i} &\triangleq \mathbbm{1}_{K_2} \otimes \boldsymbol{y}_{c,i}\\
    \boldsymbol{\mathcal{X}}^{\prime}_{c,i} &\triangleq \mathbbm{1}_{K_2} \otimes \boldsymbol{x}_{c,i}, \quad &\boldsymbol{\mathcal{Y}}^{\prime}_{c,i} &\triangleq \mathbbm{1}_{K_1} \otimes \boldsymbol{y}_{c,i} \label{eq:centroid_network_x} 
\end{align}
By taking the weighted average of the strategies within each team, we define the variables $\boldsymbol{x}_{c,i}$ and $\boldsymbol{y}_{c,i}$ as the centroid strategies for Teams 1 and 2, respectively. The variable $\boldsymbol{\cal X}_{c,i}$ is a repeated copy of $\boldsymbol{x}_{c,i}$ and has the same dimension as Team 1’s networked strategy $\boldsymbol{\cal X}^1_{i}$, while $\boldsymbol{\cal X}^\prime_{c,i}$ matches the size  of Team 2's inferred strategy $\boldsymbol{\cal X}^2_{i}$. The quantities $\boldsymbol{\mathcal{Y}}_{c,i}$, $\boldsymbol{\cal Y}^\prime_{c,i}$ are defined in a similar way. These variables are introduced solely to support our analysis and do not appear in the proposed algorithms. For both algorithms, we can establish the following recursion for 
$\boldsymbol{\mathcal{X}}_{c,i}$:
\begin{align}
\label{eq:centroid_X_evolution}
\boldsymbol{\mathcal{X}}_{c,i} &= \mathbbm{1}_{K_1} \otimes \boldsymbol{x}_{c,i} = \left(\mathbbm{1}_{K_1} {p^{(1)}}^{\top} \otimes I_{M_1}\right) \boldsymbol{\mathcal{X}}^1_{i} \notag \\
    &= \left(\mathbbm{1}_{K_1} {p^{(1)}}^{\top} \otimes I_{M_1}\right) {\mathcal{A}^{(1)}}^{\top} \left(\boldsymbol{\mathcal{X}}^1_{i-1} - \mu \boldsymbol{\mathcal{G}}^1_{x, i} \right) \notag \\
    &= \left(\mathbbm{1}_{K_1} {p^{(1)}}^{\top} \otimes I_{M_1}\right) \left(\boldsymbol{\mathcal{X}}^1_{i-1} - \mu \boldsymbol{\mathcal{G}}^1_{x, i} \right) \notag \\
    &= \boldsymbol{\mathcal{X}}_{c,i-1} - \mu \left(\mathbbm{1}_{K_1} {p^{(1)}}^{\top} \otimes I_{M_1}\right)\boldsymbol{\mathcal{G}}^1_{x, i}
\end{align}
The above relation indicates that the network centroid $\boldsymbol{x}_{c,i}$
is updated through the following recursion
\begin{equation}
\label{eq:centroid_x_evolution}
    \boldsymbol{x}_{c,i} = \boldsymbol{x}_{c,i-1} - \mu \sum_{k \in \mathcal{N}^{(1)}} p_k \widehat{\nabla_x J}_k (\boldsymbol{x}_{k,i-1}, \boldsymbol{y}_{k,i-1})
\end{equation}
The recursion in \eqref{eq:centroid_x_evolution} is used for our convergence analysis of \textbf{ATC-ITC} and \textbf{ATC-C} since the local models of each team will converge to a neighborhood around their corresponding network centroid. In other words, the network centroid $\boldsymbol{x}_{c,i}$ acts as a ``proxy" for studying the network behavior. 


\begin{lemma}
[{\bf Within-team and cross-team consensus}]
\label{lemma:zerosum_within_cross}
Let
Assumptions  \ref{ass:within_conn}, \ref{ass:cross-conn_weak},  \ref{ass:lip}, \ref{ass:bdis}, \ref{ass:gradnoise}
hold for \textbf{ATC-ITC} under zero-sum objectives \eqref{eq:zero-sum},
and
Assumptions \ref{ass:within_conn}, \ref{ass:cross-conn_strong}, \ref{ass:lip}, \ref{ass:bdis}, \ref{ass:gradnoise} hold for 
\textbf{ATC-C}.
The iterates 
$\boldsymbol{\mathcal{X}}^1_{i}$,$\boldsymbol{\mathcal{Y}}^2_{i}$, and the inferred adversary iterates 
$\boldsymbol{\mathcal{Y}}^1_{i}$, $\boldsymbol{\mathcal{X}}^2_{i}$ of both teams, 
cluster within $\mathcal{O}(\mu^2)$ around the corresponding
$\boldsymbol{x}_{c,i}$
and $\boldsymbol{y}_{c,i}$, namely,
 \begin{equation}
        \begin{aligned}
            &\mathbb{E} \{\|\boldsymbol{\mathcal{X}}^1_{i} - \boldsymbol{\mathcal{X}}_{c,i}\|^2 + \|\boldsymbol{\mathcal{Y}}^2_{i} - \boldsymbol{\mathcal{Y}}_{c,i}\|^2 \\
            &\quad + \|\boldsymbol{\mathcal{X}}^2_{i} - \boldsymbol{\mathcal{X}}^\prime_{c,i}\|^2 + \|\boldsymbol{\mathcal{Y}}^1_{i} - \boldsymbol{\mathcal{Y}}^\prime_{c,i}\|^2\} \leq O\left(\mu^2\right) 
        \end{aligned}
    \end{equation}
    when the step size $\mu$ is sufficiently small and after sufficient iterations (i.e., $i \geq i_\alpha$ for \textbf{ATC-ITC}, $i \geq i_\beta$ for \textbf{ATC-C}), where
    \begin{align}
        i_\alpha &=  \frac{\log \left(O\left(\mu^2\right)\right)}{\log \left(\alpha\right)} \\
        i_\beta &=  \frac{\log \left(O\left(\mu^2\right)\right)}{\log \left(\beta\right)}
    \end{align}
      and $\alpha < 1$ and $\beta < 1$ are constants depending on $A^{(1)}, A^{(2)}$, and $C$.
     \hfill\qed 
\end{lemma}

\begin{proof}
    The proof can be found in
      Appendix \ref{proof:lemma2_atcetc} for \textbf{ATC-ITC} and Appendix \ref{proof:lemma2_atcc} for \textbf{ATC-C}.
\end{proof}
 Lemma \ref{lemma:zerosum_within_cross} shows that each team’s local strategies cluster around its own network centroid, while each team’s inferred strategies cluster around the opposing team's network centroid. With a sufficiently small step size, the deviations from the centroids remain small, thereby justifying the use of the network centroids
$\boldsymbol{x}_{c,i}$ and $ \boldsymbol{y}_{c,i}$
as reliable proxies for decision-making.
Regarding the recursion of $\boldsymbol{z}_{c, i} = [\boldsymbol{x}_{c,i};\boldsymbol{y}_{c,i}]$, we present the following lemma, which characterizes the learning dynamics of the network centroid after sufficiently large iterations.
\begin{lemma}[{\bf Learning dynamics}]
\label{lemma:zerosum_learning_dynamic}
Let
Assumptions  \ref{ass:within_conn}, \ref{ass:cross-conn_weak}, \ref{ass:lip}, \ref{ass:bdis}, \ref{ass:gradnoise}
hold for \textbf{ATC-ITC} under zero-sum objectives \eqref{eq:zero-sum},
and
Assumptions \ref{ass:within_conn}, \ref{ass:cross-conn_strong}, \ref{ass:lip}, \ref{ass:bdis}, \ref{ass:gradnoise} hold for 
\textbf{ATC-C}. The 
centroid $\boldsymbol{z}_{c,i} = [\boldsymbol{x}_{c,i}; \boldsymbol{y}_{c,i}]$ 
follows the following dynamics 
\begin{equation}
        \boldsymbol{z}_{c,i} = \boldsymbol{z}_{c,i-1} - \mu F(\boldsymbol{z}_{c-1,i}) + \boldsymbol{d}_{c,i}
    \end{equation}  
    where the noisy term $\mathbb{E} \{\|\boldsymbol{d}_{c,i}\|^2\} \leq O(\mu^2)$,
for 
sufficiently small
$\mu$ 
 and sufficient iterations
i.e., $i \geq i_\alpha$ for \textbf{ATC-ITC}, $i \geq i_\beta$ for \textbf{ATC-C}.\hfill\qed
\end{lemma}

The proof of Lemma \ref{lemma:zerosum_learning_dynamic} is given in Appendix \ref{proof:lemma3}. The lemma  characterizes the long-term behavior of 
$\boldsymbol{z}_{c,i}$, showing that it approximately follows a gradient descent trajectory with an added random perturbation proportional to $\mathcal{O}(\mu^2)$. Note that Lemmas \ref{lemma:zerosum_within_cross} and \ref{lemma:zerosum_learning_dynamic} hold without requiring the assumption of strong monotonicity.

\begin{theorem}[{\bf Mean-square-error stability}]
\label{thm:zerosum_conv}
Let
Assumptions  \ref{ass:within_conn}, \ref{ass:cross-conn_weak},  \ref{ass:strong mono}, \ref{ass:lip}, \ref{ass:bdis}, \ref{ass:gradnoise}
hold for \textbf{ATC-ITC} under zero-sum objectives \eqref{eq:zero-sum},
and
Assumptions \ref{ass:within_conn}, \ref{ass:cross-conn_strong}, \ref{ass:strong mono}, \ref{ass:lip}, \ref{ass:bdis}, \ref{ass:gradnoise} hold for 
\textbf{ATC-C}.
The mean square error
decays at the following 
form for 
sufficiently small
$\mu$ 
 and sufficient iterations
(i.e., $i \geq i_\alpha$ for \textbf{ATC-ITC}, $i \geq i_\beta$ for \textbf{ATC-C}):
\begin{equation}
    \mathbb{E} \left[\|\boldsymbol{z}_{c,i} - z^\star\|^2\right] \leq 
    d \, \mathbb{E}\left[\|\boldsymbol{z}_{c,i-1} - z^\star\|^2\right] + O\left(\mu^2\right)
\end{equation}
where $d < 1$ is a constant depending on $\mu, \nu$, and $L_f$.
Furthermore,
the centroid $\boldsymbol{z}_{c,i}$ asymptotically converges to 
the Nash equilibrium $\boldsymbol{z}^\star$ 
in the mean-square-error sense
    \begin{equation}
        \limsup_{i \rightarrow \infty} \mathbb{E}\{\|\boldsymbol{z}_{c,i} - z^{\star}\|^2\} \leq O(\mu)
    \end{equation}
for sufficiently small step size $\mu$. \hfill\qed
\end{theorem}

The
proof of Theorem \ref{thm:zerosum_conv} can be found in Appendix \ref{proof:thm1}. From Lemma \ref{lemma:zerosum_within_cross}, each $\boldsymbol{z}_{k,i} = [\boldsymbol{x}_{k,i}; \boldsymbol{y}_{k,i}]$ moves towards the network centroids $\boldsymbol{z}_{c,i}$. Using the results of Theorem \ref{thm:zerosum_conv}, we conclude that $\forall k \in \mathcal{N}$,
\begin{equation}
    \limsup_{i \rightarrow \infty} \mathbb{E}\{\|\boldsymbol{z}_{k,i} - z^{\star}\|^2\} \leq O(\mu)
\end{equation}

\begin{remark*}
Although \textbf{ATC-ITC} is primarily concerned with a zero-sum network-game under weak cross-team subgraphs, we argue that it can be extended to the general game setting (i.e., nonzero sum games). 
This holds 
when 
each agent $k \in \mathcal{N}^{(1)}$ has access to a stochastic gradient, say
$\widehat{\nabla_y J}{}_k^{(2)}(\cdot,\cdot)$ which can approximate $\nabla_y J^{(2)}(\cdot,\cdot)$. For example, a natural choice could be any $\widehat{\nabla_y J}_\ell(\cdot,\cdot), \ell \in \mathcal{N}^{(2)}$ from Team 2.
This assumption is reasonable in practical scenarios where competitors have partial knowledge of how their opponents compute their costs. Such situations arise in strategic interactions where cost structures are partially transparent due to industry norms, regulatory requirements, or shared market intelligence.
\end{remark*}

\begin{assumption}[\textbf{Partially observable  stochastic gradient}]
\label{ass:add_grad}
    For $t, t^{\prime} \in \{1,2\}, t \neq t^{\prime}$, each agent $k \in \mathcal{N}^{(t)}$ has access to an additional stochastic gradient $\widehat{\nabla_{w^{(t)}} J}{}^{(t^{\prime})}_k\left(\cdot, \cdot\right)$, where $w^{(1)} = y$ and $ w^{(2)} = x$. The expectation equals to $\nabla_{w^{(t)}} J^{(t^{\prime})}_k (\boldsymbol{x}, \boldsymbol{y})$ conditioned on $\boldsymbol{\mathcal{F}}_i$ similar to \eqref{eq:unbiased}, with bounded gradient noise similar to \eqref{eq:b_var}. The $\nabla_{w^{(t)}} J^{(t^{\prime})}_k (\boldsymbol{x}, \boldsymbol{y})$ also has bounded gradient disagreement and is Lipschitz similar to \eqref{eq:bdis} and \eqref{eq:lip}.
    \hfill\qed
\end{assumption}


\begin{corollary}
\label{corollary:ATC-ETC}
     Under Assumption \ref{ass:add_grad}, replacing the Inference step in \eqref{eq:psi_update} with 
    \begin{equation}
    \label{eq:new_cross_predict}
        \boldsymbol{\psi}_{k,i} = \boldsymbol{y}_{k,i-1} - \mu \widehat{\nabla_y J}{}^{(2)}_k(\boldsymbol{x}_{k, i-1}, \boldsymbol{y}_{k, i-1}) 
    \end{equation}
    ensures that Lemmas \ref{lemma:zerosum_within_cross} and \ref{lemma:zerosum_learning_dynamic}, as well as Theorem \ref{thm:zerosum_conv}, remain valid for \textbf{ATC-ITC}.
\end{corollary}
\begin{proof}
\label{proof:corollary}
The proof follows by substituting the stochastic gradient in  
$\boldsymbol{\mathcal{G}}^2_{x, i}$ with  
\[
\mbox{col} \left\{\widehat{\nabla_x J}{}^{(1)}_k(\boldsymbol{x}_{k, i-1}, \boldsymbol{y}_{k, i-1})\right\}\!{}_{\vphantom{A^{[]}}k=K_1+1}^{K}
\]
and applying the same modification to  
$\boldsymbol{\mathcal{G}}^1_{y, i}$. The steps in Lemmas \ref{lemma:zerosum_within_cross} and \ref{lemma:zerosum_learning_dynamic}, as well as Theorem \ref{thm:zerosum_conv}, then proceed identically. 
\end{proof}

\section{Computer Simulations}

The proposed algorithms are suitable for diverse applications involving network competition, such as market modeling, decentralized GAN training, and power resource allocation.
To illustrate the effectiveness of the proposed algorithms, we explore two specific applications: Cournot team competition, which models the strategic interactions of firms operating within a shared market, and decentralized Wasserstein GAN training, which enhances training efficiency by utilizing distributed machine systems.

\subsection{Nonzero-sum Game: Cournot team-competition}
Cournot  game has been widely considered in the economics literature
\cite{bischi2000global, elettreby2006dynamical, ahmed2006multi, raab2009cournot}. 
We consider a network competition scenario in which two teams, consisting of $K_1, K_2$ firms respectively,
produce homogeneous goods
and compete in the same market to maximize their team's profits.

Let the Team \(1\)'s collaborative strategy $x$ be 
represented as:
\begin{equation}
    x = \begin{bmatrix}
        x(1)\\\vdots\\x(K_1)
    \end{bmatrix}
\end{equation}
The agent $\ell$ then chooses 
component $x(\ell)$ from the above vector
as its practical production quantity.
Consequently, the cost for agent $\ell$
for producing 
$x(\ell)$ volume of goods is given by:
\begin{equation}
    \boldsymbol{C}_\ell(x)  = \left(c_\ell+ \boldsymbol{v}_\ell\right) \, x^2(\ell)
\end{equation}
where $c_\ell > 0$ is the cost parameter, and $\boldsymbol{v}_\ell$  denotes identically, and independently distributed (i.i.d) zero-mean random  perturbation.  Furthermore, the pricing function of the product is modeled as:
\begin{equation}
    \boldsymbol{P}(x, y) = P - (w + \boldsymbol{z})\, (\mathbbm{1}^\top x + \mathbbm{1}^\top y)
\end{equation}
where $P>0$ and $w>0$ are pricing parameters, and $\boldsymbol{z}$  denotes i.i.d. zero-mean random perturbation. With the production cost and market price so defined, the expected profit 
for a single agent $\ell$ can be determined by $-J_\ell(x,y)$, where $J_\ell(x,y)$ is:
\begin{equation}
    J_\ell(x,y) = \mathbb{E} \{ \boldsymbol{C}_\ell(x)
-
x(\ell) \, \boldsymbol{P}(x, y) \}
\end{equation}
We can now formally introduce the optimization objective for the Cournot team-competition model as follows:
\begin{equation}
\label{eq:cournot_prob}
    \begin{aligned}
        \min_x \quad & J^{(1)}(x,y) = \sum_{\ell \in \mathcal{N}^{(1)}} p_\ell J_\ell(x,y) \\
        \min_y \quad & J^{(2)}(x,y) = \sum_{r \in \mathcal{N}^{(2)}} p_r J_r(x,y) 
    \end{aligned}
\end{equation}
A similar formulation for 
team-based Cournot game 
can be found in \cite{ahmed2006multi, yu2017distributed}.


{\em Experimental setup:
}
We consider the Cournot network competition problem with three agents for each team ($K_1 = K_2 = 3$). 
Agents 1 and 4 are the actual players involved in the market, therefore they can naturally observe each other's strategies through the market (i.e., agent 1 and 4 can observe $\boldsymbol{y}_{4,i}$ and $\boldsymbol{x}_{1,i}$, respectively). 
Depending on the requirement of the algorithm, we implicitly construct two cross-team subgraphs for  our algorithms, \textbf{ATC-ITC} and \textbf{ATC-C}:
\begin{enumerate}
    \item \textbf{ATC-C}:
    All agents can observe the market and collect the opponent team's strategy. The information flow for this case is illustrated at the top of Figure \ref{fig:cournot_explain}, which follows a strong cross-team subgraph.

    \item \textbf{ATC-ITC}: Only agent 1 and 4 can collect the opponent team's strategy. The information flow is depicted at the bottom of Figure \ref{fig:cournot_explain}, which follows a weak cross-team subgraph.
\end{enumerate}




\begin{figure}[htbp]
\centerline{\includegraphics[width=\columnwidth]{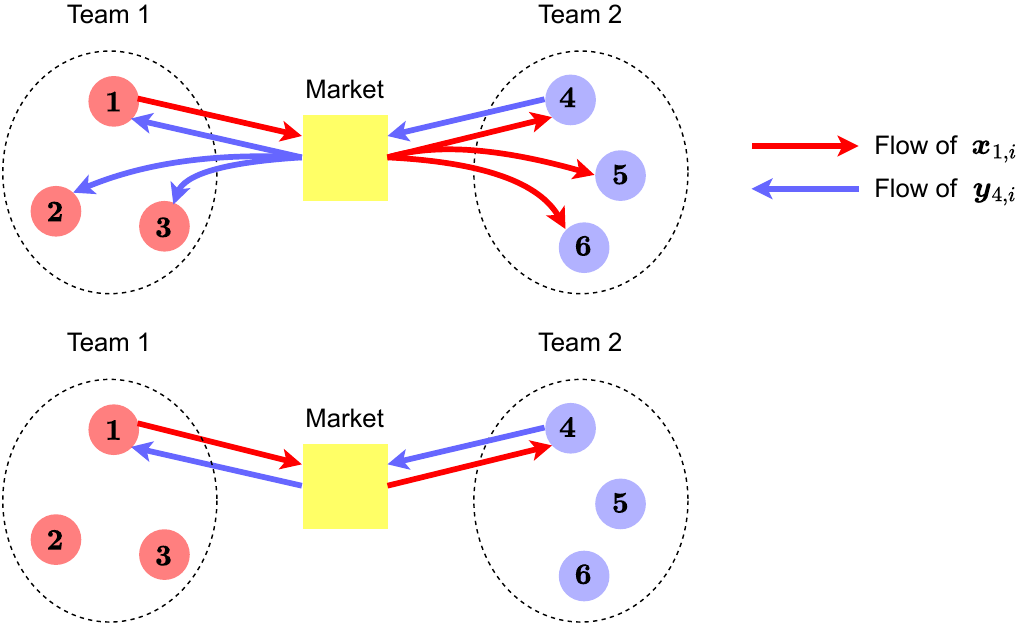}}
\caption{Cournot cross-team information flow with two representatives in each team. The top subfigure corresponds to \textbf{ATC-C}, while the bottom corresponds to \textbf{ATC-ITC}.}
\label{fig:cournot_explain}
\end{figure}

For the simulation, the cost parameters are set at $c_1 = c_4 = 5$ and $c_2 = c_3 = c_5 = c_6 = 3$. The price parameters are set at $P = 5$ and $w = 3$.  The i.i.d. random variables $\boldsymbol{v}_l, \boldsymbol{v}_r$, and $\boldsymbol{z}$ are uniformly  sampled from $[-0.1, 0.1]$. Under this setting, it can be verified that problem \eqref{eq:cournot_prob} satisfies Assumptions \ref{ass:strong mono} and \ref{ass:lip}.
Both algorithms \textbf{ATC-ITC} and \textbf{ATC-C} employ identical within-team topologies, with corresponding combination matrices $A^{(1)}$ and 
$A^{(2)}$ constructed as in \eqref{eq:within_combination_sample}, ensuring Assumption \ref{ass:within_conn} is satisfied. The combination matrix $C$ for \textbf{ATC-ITC} is provided in \eqref{eq:cross_sample}, satisfying Assumption \ref{ass:cross-conn_weak}, whereas the matrix $C$ for \textbf{ATC-C} satisfies Assumption \ref{ass:cross-conn_strong}.

 \begin{subequations}
\label{eq:combination_matrices}
    \begin{equation}
    \label{eq:within_combination_sample}
        A^{(1)} = \left[\renewcommand{\arraystretch}{1.2}\begin{array}{ccc}
        \frac{1}{3} & \frac{1}{2} & \frac{1}{2} \\
        \frac{1}{3} & \frac{1}{2} & 0 \\
        \frac{1}{3} & 0 & \frac{1}{2}
    \end{array}\right], 
        A^{(2)} = \left[ \renewcommand{\arraystretch}{1.2}
\begin{array}{ccc}
    \frac{1}{2} & \frac{1}{3} & 0 \\
    \frac{1}{2} & \frac{1}{3} & \frac{1}{2} \\
    0 & \frac{1}{3} & \frac{1}{2}
\end{array}
\right],
    \end{equation}
    \vspace{-0.5em}
    \begin{equation}
    \label{eq:cross_sample}
        C = \left[\renewcommand{\arraystretch}{1.2}\begin{array}{ccc|ccc}
            \frac{3}{10} & \frac{1}{2} & \frac{1}{2} & \frac{1}{10} & 0 & 0  \\
             \frac{3}{10} & \frac{1}{2} & 0 & 0 & 0 & 0  \\
            \frac{3}{10} & 0 & \frac{1}{2} & 0 & 0 & 0  \\
            \noalign{\hrule height 0.3pt}
            \frac{1}{10} & 0 & 0 & \frac{9}{20} & \frac{1}{3} & 0  \\
            0 & 0 & 0 & \frac{9}{20} & \frac{1}{3} & \frac{1}{2}  \\
            0 & 0 & 0 & 0 & \frac{1}{3} & \frac{1}{2}  
        \end{array}\right].
    \end{equation}
\end{subequations}

We evaluate the performance of our proposed algorithms, \textbf{ATC-ITC} and \textbf{ATC-C}, against the \textbf{Competing Diffusion (CD)} algorithm introduced in \cite{vlaski2021competing}.
Note that \textbf{CD} utilizes the same matrices $A^{(1)}, A^{(2)}$ and $C$ as \textbf{ATC-ITC}. In this nonzero-sum game, for \textbf{ATC-ITC}, we assume that at each iteration, in addition to its own stochastic loss function, each firm has access to the local loss of an adversary agent, enabling it to compute an additional gradient as required by Assumption \ref{ass:add_grad}.
From Figure \ref{fig:cournot},
we observe that 
\textbf{ATC-ITC} 
and 
 \textbf{ATC-C} converge faster than
 \textbf{CD}.
 The performance of 
\textbf{ATC-ITC}  is comparable to that of 
  \textbf{ATC-C}  even under a weaker cross-team connection.
The benefit of \textbf{ATC-ITC} over \textbf{CD}  comes from the estimate 
$\boldsymbol{\psi}_{l,i}$ instead of plain $\boldsymbol{y}_{l,i}$ in \eqref{eq:new_cross_predict}.
Furthermore, simulation results confirm that a smaller step size guarantees convergence of network centroid to a tighter $O(\mu)$-neighborhood of $z^\star$.

\begin{figure}[htbp]
\centerline{\includegraphics[width=.9\columnwidth]{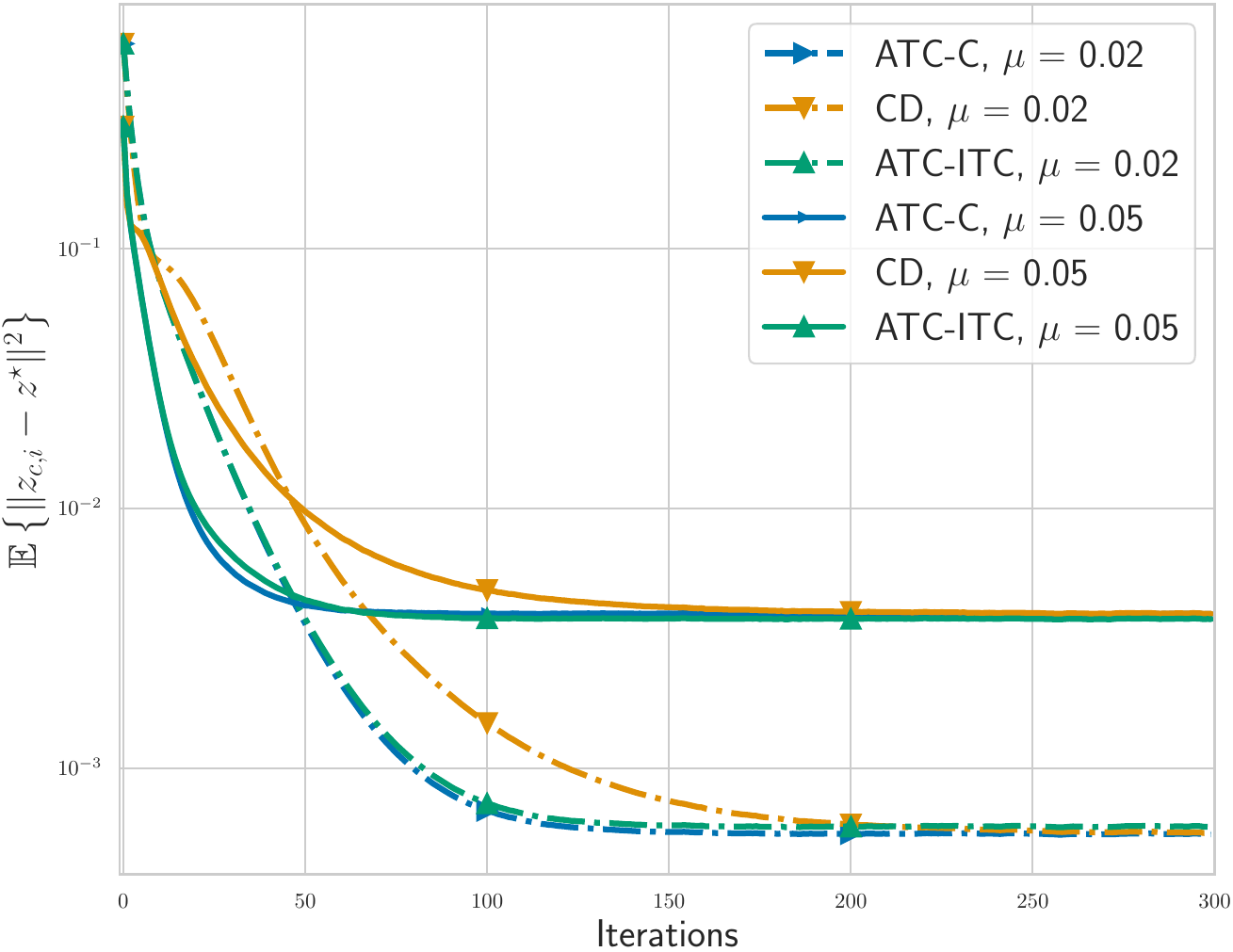}}
\caption{Performance of \textbf{ATC-ITC} (with Assumption \ref{ass:add_grad}), \textbf{CD}, \textbf{ATC-C} in Cournot team-competition}
\label{fig:cournot}
\end{figure}


\subsection{Zero-Sum Game: Decentralized Wasserstein GAN Training}
\label{sec:GAN}
In this example, we consider the application of decentralized training of an $\ell_2$-regularized Wasserstein GAN (WGAN) for learning the mean and variance of a one-dimensional Gaussian \cite{cai2024diffusion, yang2022faster}. 
 Using our framework, the training process for both the discriminator and generator can be decentralized into multiple agents of separate networks, potentially reducing computation requirements compared to the traditional distributed approach \cite{cai2024diffusion}.
The optimization problem is modeled as a two-network zero-sum game, where two teams optimize the following objectives:
\begin{equation}
\label{eq:experiment_problem}
    \begin{aligned}
        \min_x \quad & J^{(1)}(x,y) = \sum_{\ell \in \mathcal{N}^{(1)}} p_\ell J_\ell(x,y) \\
        \min_y \quad & J^{(2)}(x,y) = \sum_{r \in \mathcal{N}^{(2)}} p_r J_r(x,y) 
    \end{aligned}
\end{equation}
where local costs are defined as: 
\begin{subequations}
\renewcommand{\theequation}{\theparentequation\alph{equation}}
\begin{align}
J_\ell(x,y) \triangleq& \
 \mathbb{E}_{(\boldsymbol{u}_\ell, \boldsymbol{z}_\ell) \sim \mathcal{D}} \left[ D(y; \boldsymbol{u}_\ell) -  D(y; G(x; \boldsymbol{z}_\ell))  \notag \right.\\
& \left. + \lambda_x \| x \|^2 - \lambda_y \| y \|^2 \right]\\
J_r(x,y) \triangleq&  -
 \mathbb{E}_{(\boldsymbol{u}_r, \boldsymbol{z}_r) \sim \mathcal{D}} \left[ D(y; \boldsymbol{u}_r) - D(y; G(x; \boldsymbol{z}_r)) \right. \notag
\\
& \left. + \lambda_x \| x \|^2 - \lambda_y \| y \|^2\right]
\end{align}
\end{subequations}

\begin{figure*}[!htbp]
    \centering
    \begin{subfigure}[b]{0.38\linewidth}
        \centering
        \includegraphics[width=\linewidth]{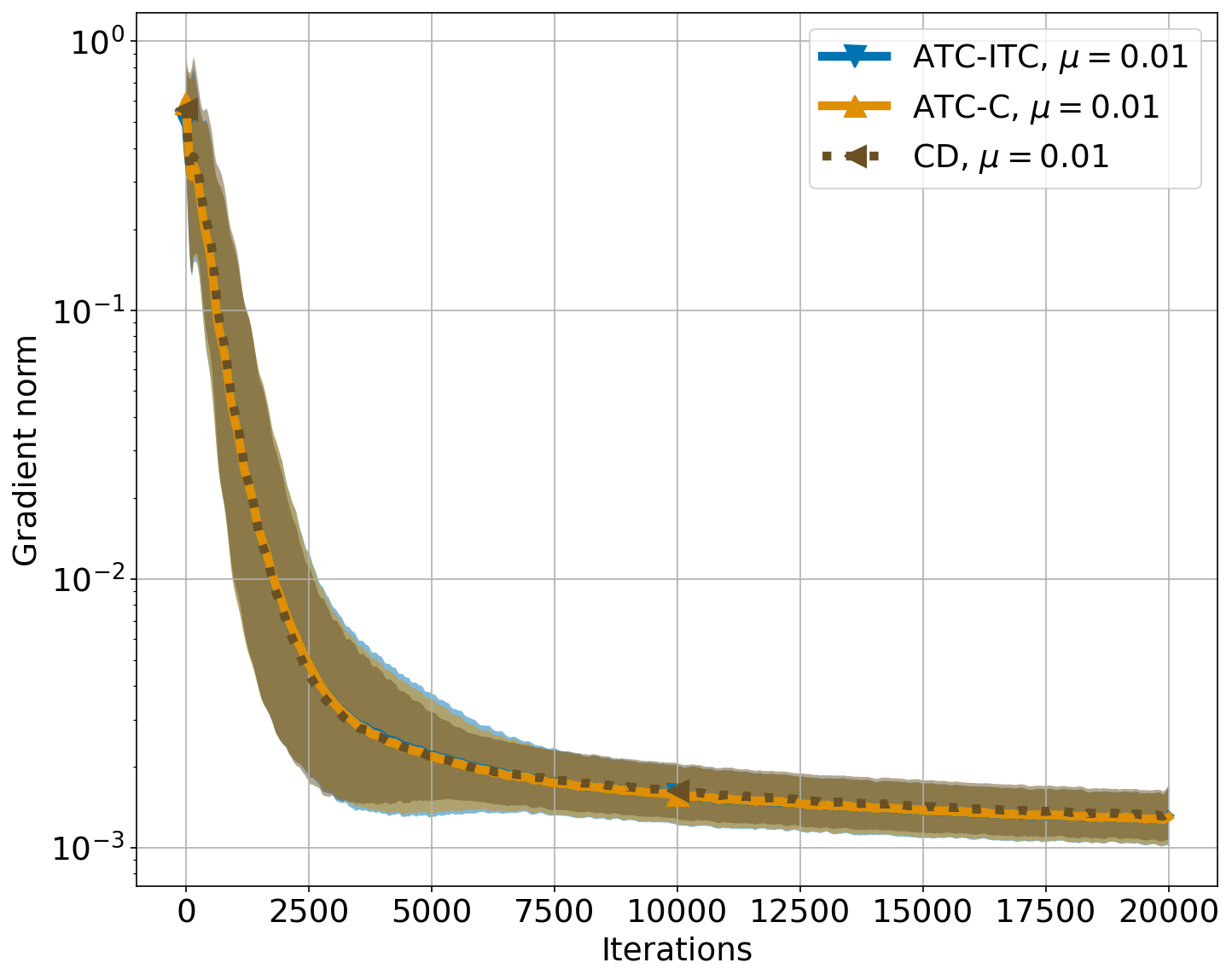}
        \caption{\scalebox{0.9}{$\frac{1}{K} \sum_{k=1}^K \left( \| \widehat{\nabla_x} J_k(\boldsymbol{x}_{k, i}, \boldsymbol{y}_{k, i}) \| + \| \widehat{\nabla_y} J_k(\boldsymbol{x}_{k, i}, \boldsymbol{y}_{k, i}) \| \right)$}}
        \label{fig:example1_mu0.05_grad}
    \end{subfigure}
    \hspace{1em}
    \begin{subfigure}[b]{0.38\linewidth}
        \centering
        \includegraphics[width=\linewidth]{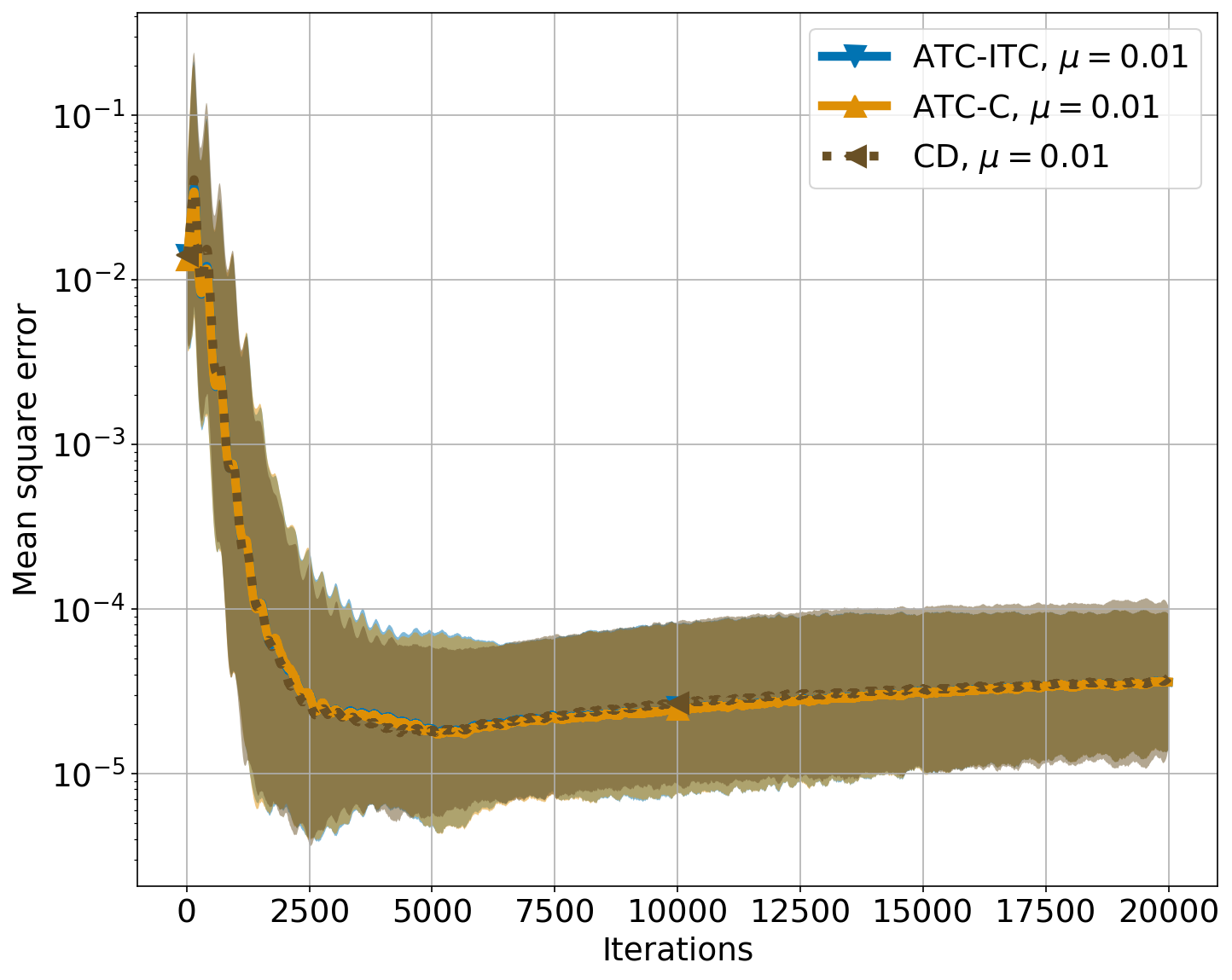}
        \caption{\scalebox{0.9}{$\frac{1}{K} \sum_{k=1}^K \left( \| \hat{\pi}_{k,i} - \pi \|^2 + \| \hat{\sigma}_{k,i} - \sigma \|^2 \right)$}}
        \label{fig:example1_mu0.05_mse}
    \end{subfigure}

    \vspace{1em}
    
    \begin{subfigure}[b]{0.38\linewidth}
        \centering
        \includegraphics[width=\linewidth]{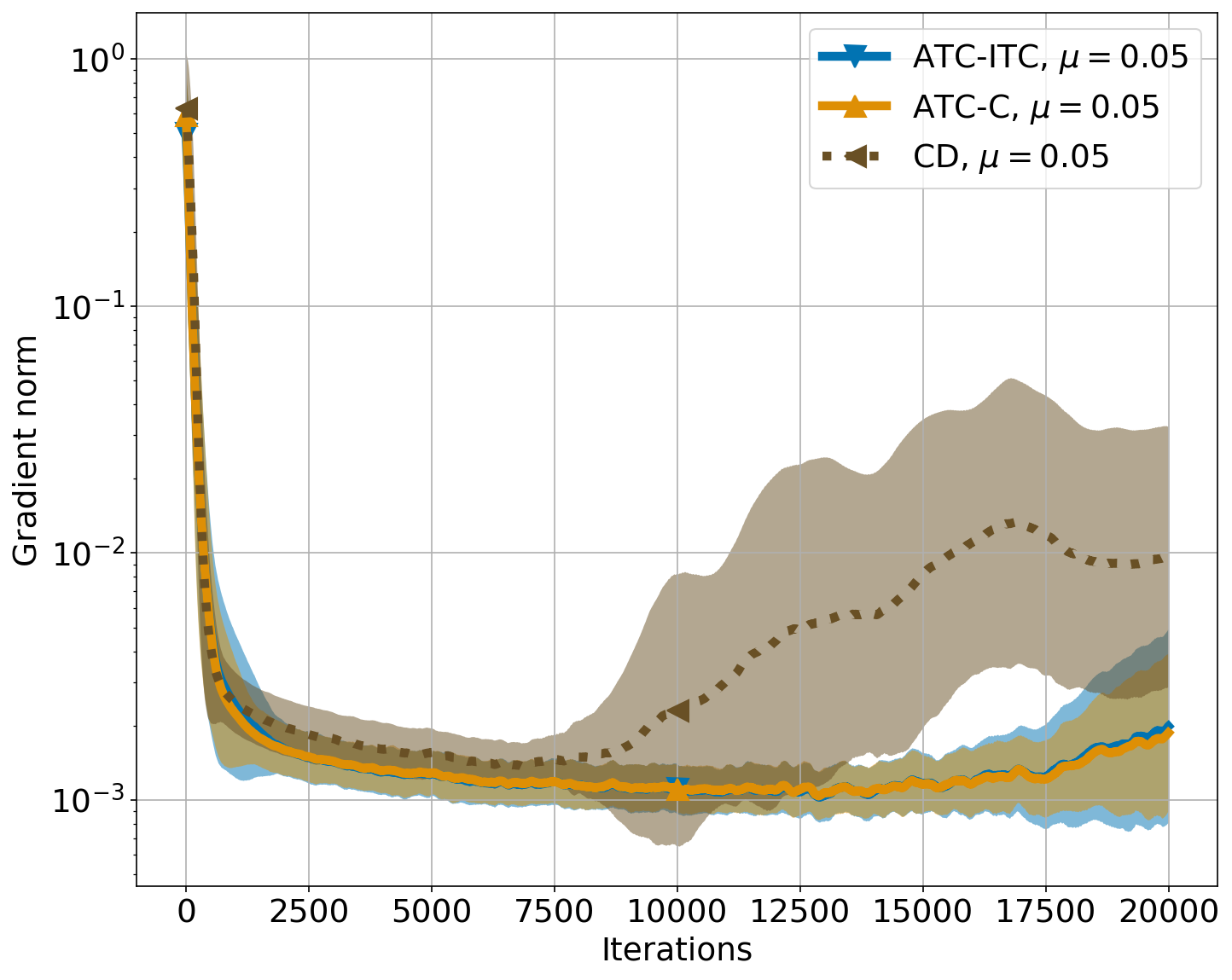}
        \caption{\scalebox{0.9}{$\frac{1}{K} \sum_{k=1}^K \left( \| \widehat{\nabla_x} J_k(\boldsymbol{x}_{k, i}, \boldsymbol{y}_{k, i}) \| + \| \widehat{\nabla_y} J_k(\boldsymbol{x}_{k, i}, \boldsymbol{y}_{k, i}) \| \right)$}}
        \label{fig:example1_mu0.01_grad}
    \end{subfigure}
    \hspace{1em}
    \begin{subfigure}[b]{0.38\linewidth}
        \centering
        \includegraphics[width=\linewidth]{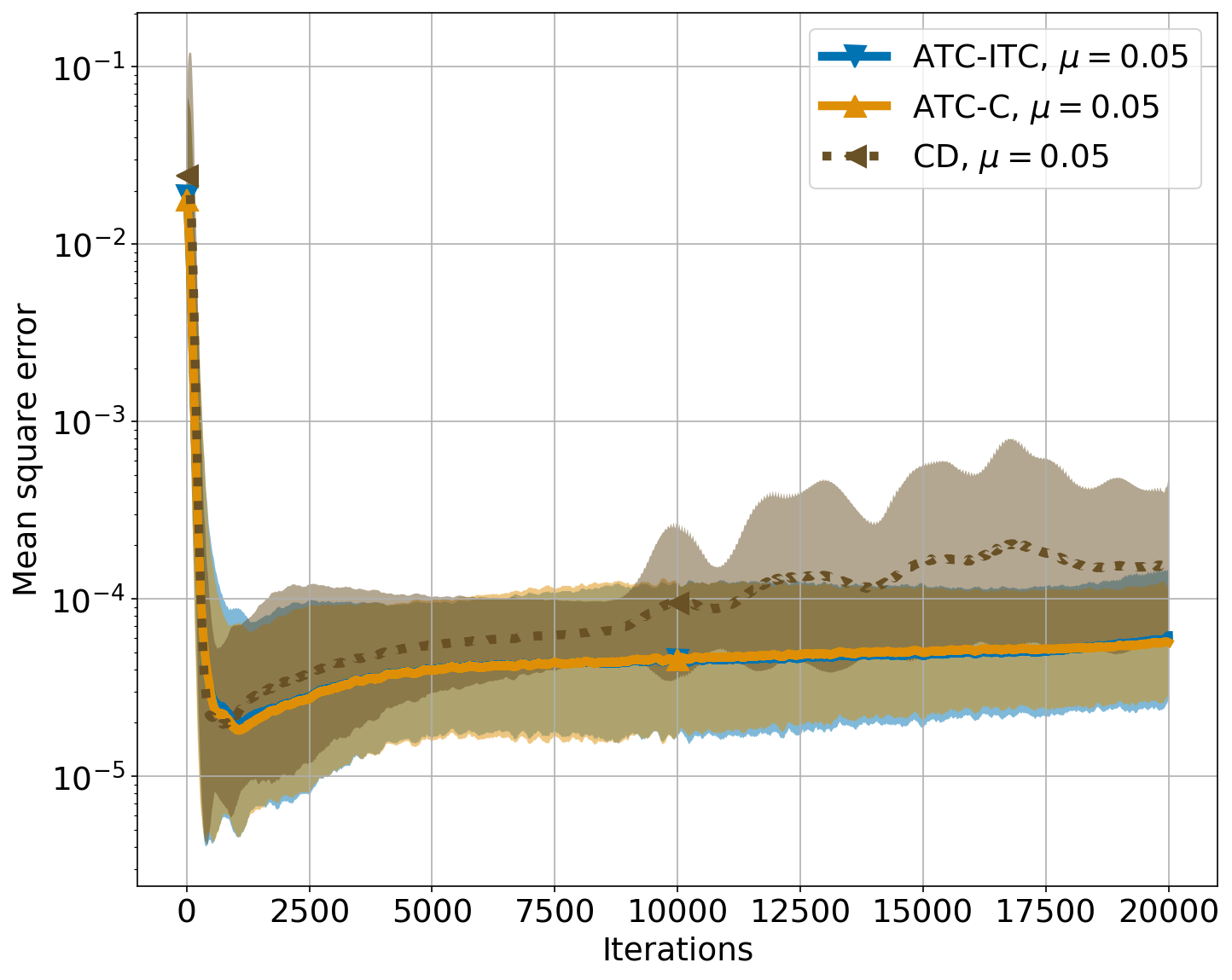}
        \caption{\scalebox{0.9}{$\frac{1}{K} \sum_{k=1}^K \left( \| \hat{\pi}_{k,i} - \pi \|^2 + \| \hat{\sigma}_{k,i} - \sigma \|^2 \right)$}}
        \label{fig:example1_mu0.01_mse}
    \end{subfigure}

    \caption{Evolution of gradient norm and mean square error distance between the true values (i.e. mean $\pi$, standard deviation $\sigma$) and estimated ones (i.e. $\hat{\pi}_{k,i}$, $\hat{\sigma}_{k,i}$):
    In (a), (b), (c), and (d), the true model is given by $\pi = 0, \sigma = 0.01$.}
    \label{fig:sample_figure}
\end{figure*}

Here, $x, y$ are regularized by constants $\lambda_x$ and $\lambda_y$, respectively. The generator $G(x;\boldsymbol{z})$ has a single hidden layer with 5 neurons, and the discriminator is parameterized as $D(y; \boldsymbol{u}) = y_1 \boldsymbol{u} + y_2 \boldsymbol{u}^2$, where $y \in \mathbb{R}^2$. The true samples $\boldsymbol{u}_\ell$ and $\boldsymbol{u}_r$ are drawn from $\mathcal{N}(\pi, \sigma^2)$, and  $\boldsymbol{z}_\ell$ and $\boldsymbol{z}_r$ are drawn from $\mathcal{N}(0, 1)$. In our experiment, we set $\lambda_x = 0.00001, \lambda_y = 0.001$, $\pi = 0$, $\sigma = 0.01$.

 For the network setup, we consider $K_1 = 6$ agents  for team 1 and $K_2 = 4$ agents  for team 2, and thus $K = K_1+K_2 = 10$. We compare the performance of our proposed algorithms, \textbf{ATC-ITC} and \textbf{ATC-C}, against the \textbf{CD} algorithm proposed in \cite{vlaski2021competing}. Two strongly-connect networks are generated for the teams using the averaging rule. For \textbf{ATC-C}, we implement a strong cross-team subgraph, whereas for \textbf{ATC-ITC} and \textbf{CD}, a weak cross-team subgraph is used, with only half of the agents in each team can directly receive information from the adversary team.

In Figures~\ref{fig:example1_mu0.05_grad} and \ref{fig:example1_mu0.05_mse}, we observe that when $\mu = 0.01$, all three algorithms exhibit similar performance in terms of the convergence speed the mean-square-error deviation measure. However, when $\mu = 0.05$, Figures~\ref{fig:example1_mu0.01_grad} and \ref{fig:example1_mu0.01_mse} indicate that the training process of \textbf{CD} is less stable compared to \textbf{ATC-ITC} and \textbf{ATC-C}, leading to degraded performance. 
In other words,
 our proposed \textbf{ATC-ITC} and \textbf{ATC-C} algorithms have a larger stability range of the step size $\mu$ compared to \textbf{CD}. 

Additionally, it is worth noting that, unlike the cooperative decentralized GAN training scheme—where each agent computes gradients with respect to both $x$ and $y$ at each iteration \cite{cai2024diffusion}, \textbf{ATC-C} requires each node to compute gradients with respect to either $x$ or $y$ only (as detailed in Algorithm~\ref{alg:network_learning_general_strong}). This highlights the advantage of \textbf{ATC-C}, which requires less computational burden per agent.

\section{Conclusion}
In this work, we proposed two diffusion learning algorithms, \textbf{ATC-ITC} and  \textbf{ATC-C}, for solving adaptive stochastic competing network problems. We examined two key scenarios: (i) a zero-sum game with weak cross-team subgraphs and (ii) a general game with strong cross-team subgraphs. We proved that both algorithms converge asymptotically to the Nash equilibrium in the mean-square-error sense.
Simulation results on Cournot team competition and decentralized GAN training illustrate their effectiveness. For future work, we plan to extend our study to multicluster games with weak cross-team connections and a broader class of functions to support more applications.


\vspace{-0.5em}
\appendices
\section{Proof of Lemma \ref{lemma:zerosum_within_cross} for \textbf{ATC-ITC}} 
\label{proof:lemma2_atcetc}
In this section, we present the proof of Lemma \ref{lemma:zerosum_within_cross} when running the \textbf{ATC-ITC} algorithm. 
\vspace{-0.5em}
\subsection{Key Lemmas}
Let us define $B^{(x)} \in \mathbb{R}^{K \times K}$ and $B^{(y)} \in \mathbb{R}^{K \times K}$ as follows:
\begin{equation}
    \!\!\!\!B^{(x)} \triangleq 
    \begin{bmatrix}
    {A}^{(1)} & A^{(1)} C^{(12)}\\
    0 & {C^{(2)}} 
\end{bmatrix}, B^{(y)} \triangleq 
    \begin{bmatrix}
    {A}^{(2)} & A^{(2)} C^{(21)}\\
    0 & {C^{(1)}} 
\end{bmatrix}
\end{equation}

\begin{lemma} \label{lemma:BxBy}
Under Assumptions \ref{ass:within_conn}, \ref{ass:cross-conn_weak}, the matrices $B^{(x)}$ and $B^{(y)}$ satisfy the following properties:
\begin{enumerate}[label=(\arabic*), ref=\thelemma(\arabic*)]
    \item \label{lemma:Bx_left} $B^{(x)}$ and $B^{(y)}$ are left-stochastic.
    \item \label{lemma:Bx_simple} 
     For both matrices $B^{(x)}$ and $B^{(y)}$,
 $1$ is the associated eigenvalue which is unique, simple, and largest in magnitude.
    \item \label{lemma:Bx-vector} There exists a unique eigenvector (up to scaling factor) associated with the eigenvalue $1$
    for each matrix, i.e., 
    \begin{equation} \label{eq:eigenvectors}
        p^{(x)} \triangleq \begin{bmatrix}
            p^{(1)} \\
            0
        \end{bmatrix}, \quad 
        p^{(y)} \triangleq \begin{bmatrix}
            p^{(2)} \\
            0
        \end{bmatrix}
    \end{equation}
    where $p^{(x)}$ and $p^{(y)} \in \mathbb{R}^{K \times 1}$ are the eigenvectors of $B^{(x)}$ and $B^{(y)}$, respectively. \hfill\qed
\end{enumerate}
\end{lemma}

\begin{proof}
    \label{proof:matrix}
    We show our proof for $B^{(x)}$; the proof for $B^{(y)}$ follows a similar argument.

    \begin{enumerate}[label=(\arabic*)]
    \item  The left-stochastic property can be verified via
    \begin{equation}
    \begin{aligned}
        \mathbbm{1}_{K}^{\top} B^{(x)} & = \begin{bmatrix}
            \mathbbm{1}_{K_1}^{\top} {A}^{(1)} + \mathbbm{1}_{K_2}^{\top} 0 & \mathbbm{1}_{K_1}^{\top} A^{(1)} C^{(12)} + \mathbbm{1}_{K_2}^{\top} {C^{(2)}}
        \end{bmatrix}\\
        & \stackrel{(a)}{=}
        \begin{bmatrix}
            \mathbbm{1}_{K_1}^{\top} & \mathbbm{1}_{K_1}^{\top} C^{(12)} + \mathbbm{1}_{K_2}^{\top} {C^{(2)}}
        \end{bmatrix}\\
        & \stackrel{(b)}{=}
        \mathbbm{1}^{\top}_{K}
    \end{aligned}
    \end{equation}
    where $(a)$ follows from Assumption \ref{ass:within_conn}; $(b)$ follows from Assumption \ref{ass:cross-conn_weak}.
    
    \item According to \cite{ying2016information}, Assumption \ref{ass:cross-conn_weak} ensures that $\rho\left({C^{(2)}}\right) < 1$, while Assumption \ref{ass:within_conn} guarantees that  1 is the unique eigenvalue on the spectral circle of $A^{(1)}$ and is simple (i.e., multiplicity equal to one). Since $B^{(x)}$ has a block upper-triangular structure, its largest eigenvalue is also 1, and this eigenvalue is simple, with all other eigenvalues having absolute values less than 1 (i.e., $1 > |\lambda_2(B^{(x)})| \geq |\lambda_3(B^{(x)})| \cdots$).

    \item The corresponding unique eigenvector (up to scaling) for the eigenvalue 1 is
    \begin{equation}
        p^{(x)} \triangleq \begin{bmatrix}
            p^{(1)} \\
            0
        \end{bmatrix} \in \mathbb{R}^{K \times 1}
        \label{eq:px_defn}
    \end{equation}
    which satisfies
    \begin{equation}
    \begin{aligned}
        B^{(x)} p^{(x)} &= \begin{bmatrix}
            {A}^{(1)} p^{(1)} + A^{(1)} C^{(12)} 0 \\
            0 p^{(1)} + {C^{(2)}} 0
        \end{bmatrix} \\
        &\stackrel{(a)}{=} \begin{bmatrix}
            p^{(1)} \\
            0
        \end{bmatrix}
    \end{aligned}
    \end{equation}
    where $(a)$ is due to the fact that ${A}^{(1)} p^{(1)} = p^{(1)}$.
    \end{enumerate}
\end{proof}
\vspace{-1.8em}
\begin{remark}
\label{remark:jordan}
     From Lemma \ref{lemma:BxBy} and \cite{sayed2014adaptation}, we find that $B^{(x)}$ admits the following Jordan decomposition: $B^{(x)} = P M P^{-1}$, where
\begin{equation}
\label{eq:Bxjordan}
    P = \begin{bmatrix} p^{(x)} & P_{R} \end{bmatrix}, \ 
    M  = \begin{bmatrix} 1 & 0 \\ 0 & M_{\eta} \end{bmatrix}, \
    P^{-1} = \begin{bmatrix} \mathbbm{1}_K^{\top} \\ P_{L}^{\top} \end{bmatrix}.
\end{equation}
Similarly, for $B^{(y)}$, we have the decomposition $B^{(y)} = Q N Q^{-1}$, where
\begin{equation}
    Q = \begin{bmatrix} p^{(y)} & Q_{R} \end{bmatrix}, \ 
    N  = \begin{bmatrix} 1 & 0 \\ 0 & N_{\iota} \end{bmatrix}, \
    Q^{-1} = \begin{bmatrix} \mathbbm{1}_K^{\top} \\ Q_{L}^{\top} \end{bmatrix}.
\end{equation}
Here, $M_{\eta}, N_\iota$ are 
 block Jordan matrices with the eigenvalues on the diagonal and $\eta, \iota$ on the first lower subdiagonal, respectively.
We define
\begin{align}
   \mathcal{P}_L &\triangleq P_L \otimes I_{M_1}, \hspace{-8pt}
   &\mathcal{P}_R &\triangleq P_R \otimes I_{M_1}, \hspace{-8pt}
   &\mathcal{M}_\eta &\triangleq M_\eta \otimes I_{M_1}\label{eq:PPM}\\
   \mathcal{Q}_L &\triangleq Q_L \otimes I_{M_2}, \hspace{-8pt}
   &\mathcal{Q}_R &\triangleq Q_R \otimes I_{M_2}, \hspace{-8pt}
   &\mathcal{N}_\iota &\triangleq N_\iota \otimes I_{M_2} \label{eq:QQN}
\end{align}
\end{remark}

\begin{lemma}
\label{lemma:X1iX2i}
    Under Assumptions \ref{ass:within_conn}, \ref{ass:cross-conn_weak}, for \textbf{ATC-ITC}, we have:
    \begin{align}
    \!\!\!\!\| \boldsymbol{\mathcal{X}}^1_{i} - \boldsymbol{\mathcal{X}}_{c, i} \|^2 + \| \boldsymbol{\mathcal{X}}^2_{i} - \boldsymbol{\mathcal{X}}^\prime_{c, i} \|^2 &\leq \|\mathcal{P}_L\|^2 \left\| \mathcal{P}_R^{\top} \begin{bmatrix}
    \boldsymbol{\mathcal{X}}^1_{i} \\
    \boldsymbol{\mathcal{X}}^2_{i}
    \end{bmatrix}\right\|^2 \label{eq:X1iX2i}\\
    \!\!\!\!\| \boldsymbol{\mathcal{Y}}^1_{i} - \boldsymbol{\mathcal{Y}}^\prime_{c, i} \|^2 + \| \boldsymbol{\mathcal{Y}}^2_{i} - \boldsymbol{\mathcal{Y}}_{c, i} \|^2 &\leq \|\mathcal{Q}_L\|^2 \left\| \mathcal{Q}_R^{\top} \begin{bmatrix}
    \boldsymbol{\mathcal{Y}}^1_{i} \\
    \boldsymbol{\mathcal{Y}}^2_{i}
\end{bmatrix}\right\|^2 \label{eq:Y1iY2i}
    \end{align}\hfill\qed
\end{lemma}

\begin{proof}
    We establish \eqref{eq:X1iX2i}; a similar argument applies to \eqref{eq:Y1iY2i}.
    Using the property of $p^{(x)}$ in \eqref{eq:px_defn}, the following relationship holds:
\begin{equation}
\begin{aligned}
        \begin{bmatrix}
    \boldsymbol{\mathcal{X}}_{c,i} \\
    \boldsymbol{\mathcal{X}}^\prime_{c,i}
\end{bmatrix} &= \left(\mathbbm{1}_K {p^{(1)}}^{\top} \otimes I_{M_1}\right) \boldsymbol{\mathcal{X}}^1_{i} \\&= \left(\mathbbm{1}_K {p^{(x)}}^{\top} \otimes I_{M_1}\right)
    \begin{bmatrix}
    \boldsymbol{\mathcal{X}}^1_{i} \\
    \boldsymbol{\mathcal{X}}^2_{i}
\end{bmatrix}
\end{aligned}
\end{equation}
Therefore, for $\| \boldsymbol{\mathcal{X}}^1_i - \boldsymbol{\mathcal{X}}_{c, i} \|^2 + \| \boldsymbol{\mathcal{X}}^2_i - \boldsymbol{\mathcal{X}}^\prime_{c, i} \|^2$, we have:
\begin{equation}
\label{eq:centroid_x_bound}
\begin{aligned}
    &\| \boldsymbol{\mathcal{X}}^1_i - \boldsymbol{\mathcal{X}}_{c, i} \|^2 + \| \boldsymbol{\mathcal{X}}^2_i - \boldsymbol{\mathcal{X}}^\prime_{c, i} \|^2\\
    &= \left\|\begin{bmatrix}
    \boldsymbol{\mathcal{X}}^1_{i} \\
    \boldsymbol{\mathcal{X}}^2_{i}
\end{bmatrix} - \begin{bmatrix}
    \boldsymbol{\mathcal{X}}_{c,i} \\
    \boldsymbol{\mathcal{X}}^\prime_{c,i}
\end{bmatrix}\right\|^2 \\
& = \left\| (I_{K} \otimes I_{M_1})\begin{bmatrix}
    \boldsymbol{\mathcal{X}}^1_{i} \\
    \boldsymbol{\mathcal{X}}^2_{i}
\end{bmatrix} - (\mathbbm{1}_K {p^{(x)}}^{\top} \otimes I_{M_1})\begin{bmatrix}
    \boldsymbol{\mathcal{X}}^1_{i} \\
    \boldsymbol{\mathcal{X}}^2_{i}
\end{bmatrix} \right\|^2 \\
& \stackrel{(a)}{=} \left\| \mathcal{P}_L \mathcal{P}_R^{\top} \begin{bmatrix}
    \boldsymbol{\mathcal{X}}^1_{i} \\
    \boldsymbol{\mathcal{X}}^2_{i}
\end{bmatrix}\right\|^2 \\
& \leq 
\|\mathcal{P}_L\|^2 \left\| \mathcal{P}_R^{\top} \begin{bmatrix}
    \boldsymbol{\mathcal{X}}^1_{i} \\
    \boldsymbol{\mathcal{X}}^2_{i}
\end{bmatrix}\right\|^2
\end{aligned}
\end{equation} where $(a)$ follows from the identity (see \eqref{eq:Bxjordan}):
\begin{align}
    (I_{K} \otimes I_{M_1}) &= (P^{-1} \otimes I_{M_1})^{\top} (P \otimes I_{M_1})^{\top} \notag\\
    &= \left(\mathbbm{1}_K {p^{(x)}}^{\top} + P_{L} P_{R}^{\top}\right) \otimes I_{M_1}
\end{align}
\end{proof}

\begin{lemma}
\label{lemma:PRTG}
    With zero-sum objectives \eqref{eq:zero-sum}, under Assumptions \ref{ass:within_conn}, \ref{ass:cross-conn_weak}, \ref{ass:lip}, \ref{ass:bdis}, \ref{ass:gradnoise}, for \textbf{ATC-ITC}, we have:
    \begin{align}
    &\mathbb{E} \left\{\left\|\mathcal{P}_R^{\top} \begin{bmatrix}
        \boldsymbol{\mathcal{G}}^1_{x, i} \\
        \boldsymbol{\mathcal{G}}^2_{x, i}
    \end{bmatrix}\right\|^2\right\} \notag\\
    &\leq 
    8 \|\mathcal{P}_R^{\top}\|^2 L_f^2 \left( \|\mathcal{P}_L\|^2 \mathbb{E} \left\{\left\| \mathcal{P}_R^{\top} \begin{bmatrix}
    \boldsymbol{\mathcal{X}}^1_{i-1} \\
    \boldsymbol{\mathcal{X}}^2_{i-1}
\end{bmatrix}\right\|^2\right\}  \right. \notag\\
& \quad \left. + \|\mathcal{Q}_L\|^2 \mathbb{E} \left\{\left\| \mathcal{Q}_R^{\top} \begin{bmatrix}
    \boldsymbol{\mathcal{Y}}^1_{i-1} \\
    \boldsymbol{\mathcal{Y}}^2_{i-1}
\end{bmatrix}\right\|^2\right\}\right) + 2\|\mathcal{P}_R^{\top}\|^2 K (2G^2 + \sigma^2)
\label{boundPR}
    \end{align}
    Furthermore, 
 $\mathbb{E} \left\{\left\|\mathcal{Q}_R^{\top} \begin{bmatrix}
        \boldsymbol{\mathcal{G}}^1_{y, i} \\
        \boldsymbol{\mathcal{G}}^2_{y, i}
    \end{bmatrix}\right\|^2\right\}$ admits an upper bound similar to the right-hand side of \eqref{boundPR}. \hfill\qed
\end{lemma}

\begin{proof}
    We use $\mathcal{G}^1_{x, i}$ and $\mathcal{G}^2_{x, i}$ to denote the deterministic realization for $\boldsymbol{\mathcal{G}}^1_{x, i}$ and $\boldsymbol{\mathcal{G}}^2_{x, i}$:
    \begin{align}
        \mathcal{G}^1_{x, i} &\triangleq \text{col} \{\nabla_x J_k\left(\boldsymbol{x}_{k, i-1}, \boldsymbol{y}_{k, i-1}\right)\}_{k=1}^{K_1} \\
        \mathcal{G}^2_{x, i} &\triangleq \text{col} \{ - \nabla_x J_k\left(\boldsymbol{x}_{k, i-1}, \boldsymbol{y}_{k, i-1}\right)\}_{k=K_1 + 1}^{K}\end{align} 
        The corresponding stochastic noise is denoted by:
        \begin{align}
        \boldsymbol{\mathcal{S}}_{x, i}^{1} &\triangleq \mathrm{col} \{\widehat{\nabla_x J}_k(\boldsymbol{x}_{k, i-1}, \boldsymbol{y}_{k, i-1})  \notag\\
        & \qquad - \nabla_x J_k\left(\boldsymbol{x}_{k, i-1}, \boldsymbol{y}_{k, i-1}\right)\}_{k=1}^{K_1} \\
        \boldsymbol{\mathcal{S}}_{x, i}^{2} &\triangleq \mathrm{col} \{- \widehat{\nabla_x J}_k(\boldsymbol{x}_{k, i-1}, \boldsymbol{y}_{k, i-1}) \notag\\
        & \qquad + \nabla_x J_k\left(\boldsymbol{x}_{k, i-1}, \boldsymbol{y}_{k, i-1}\right)\}_{k=K_1+1}^{K}
    \end{align} 
    Next note that
    \begin{equation}
        \begin{aligned}
            \left\|\mathcal{P}_R^{\top} \begin{bmatrix}
    \boldsymbol{\mathcal{G}}^1_{x, i} \\
    \boldsymbol{\mathcal{G}}^2_{x, i}
\end{bmatrix}\right\|^2 &= \left\|\mathcal{P}_R^{\top} \left( \begin{bmatrix}
    \mathcal{G}^1_{x, i} \\
    \mathcal{G}^2_{x, i}
\end{bmatrix} + \begin{bmatrix}
    \boldsymbol{\mathcal{S}}_{x, i}^{1} \\
    \boldsymbol{\mathcal{S}}_{x, i}^{2}
\end{bmatrix}
\right) \right\|^2 \\
&\leq 2\left\|\mathcal{P}_R^{\top} \begin{bmatrix}
    \mathcal{G}^1_{x, i} \\
    \mathcal{G}^2_{x, i}
\end{bmatrix}\right\|^2 + 2\left\|\mathcal{P}_R^{\top} \begin{bmatrix}
    \boldsymbol{\mathcal{S}}_{x, i}^{1} \\
    \boldsymbol{\mathcal{S}}_{x, i}^{2}
\end{bmatrix} \right\|^2 \\
            &\leq 2\left\|\mathcal{P}_R^{\top} \begin{bmatrix}
    \mathcal{G}^1_{x, i} \\
    \mathcal{G}^2_{x, i}
\end{bmatrix}\right\|^2 + 2\left\|\mathcal{P}_R^{\top}\right\|^2 \left\| \begin{bmatrix}
    \boldsymbol{\mathcal{S}}_{x, i}^{1} \\
    \boldsymbol{\mathcal{S}}_{x, i}^{2}
\end{bmatrix}\right\|^2 \\
        \end{aligned}
    \end{equation}
Taking the conditional expectation and  since $\boldsymbol{x}_{k, i-1}, \boldsymbol{y}_{k, i-1} \in \boldsymbol{\mathcal{F}}_{i-1}$, we obtain:

    \begin{equation}
    \label{eq:exp_PRTG}
        \mathbb{E} \left\{\left\|\mathcal{P}_R^{\top} \begin{bmatrix}
    \boldsymbol{\mathcal{G}}^1_{x, i} \\
    \boldsymbol{\mathcal{G}}^2_{x, i}
\end{bmatrix}\right\|^2 \middle| \boldsymbol{\mathcal{F}}_{i-1}\right\} \leq 2\left\|\mathcal{P}_R^{\top} \begin{bmatrix}
    \mathcal{G}^1_{x, i} \\
    \mathcal{G}^2_{x, i}
\end{bmatrix}\right\|^2 + 2\left\|\mathcal{P}_R^{\top}\right\|^2 K \sigma^2
    \end{equation}
where we used Assumption \ref{ass:gradnoise}.  Now we get:
\allowdisplaybreaks
    \begin{align}
\label{eq:PRTG}
    &\left\|\mathcal{P}_R^{\top} \begin{bmatrix}\mathcal{G}^1_{x, i} \\\mathcal{G}^2_{x, i}\end{bmatrix}\right\|^2 \notag\\
    &= \left\|\mathcal{P}_R^{\top} \begin{bmatrix}
        \text{col}\{\nabla_x J_k(\boldsymbol{x}_{k, i-1}, \boldsymbol{y}_{k, i-1})\}_{k=1}^{K_1} \\
        \text{col}\{-\nabla_x J_k(\boldsymbol{x}_{k, i-1}, \boldsymbol{y}_{k, i-1})\}_{k=K_1+1}^{K}
    \end{bmatrix}\right\|^2 \notag\\
    &\stackrel{(a)}{=} \Bigg\|\mathcal{P}_R^{\top} \left(\begin{bmatrix}
        \text{col}\{\nabla_x J_k(\boldsymbol{x}_{k, i-1}, \boldsymbol{y}_{k, i-1})\}_{k=1}^{K_1} \\
        \text{col}\{-\nabla_x J_k(\boldsymbol{x}_{k, i-1}, \boldsymbol{y}_{k, i-1})\}_{k=K_1+1}^{K}
    \end{bmatrix} \right. \notag\\
    &\quad\left. - \begin{bmatrix}
        \text{col}\{\nabla_x J_k(\boldsymbol{x}_{c, i-1}, \boldsymbol{y}_{c, i-1})\}_{k=1}^{K_1} \\
        \text{col}\{-\nabla_x J_k(\boldsymbol{x}_{c, i-1}, \boldsymbol{y}_{c, i-1})\}_{k=K_1+1}^{K}
    \end{bmatrix}\right) \notag\\
    &\quad + \mathcal{P}_R^{\top} \left(\begin{bmatrix}
        \text{col}\{\nabla_x J_k(\boldsymbol{x}_{c, i-1}, \boldsymbol{y}_{c, i-1})\}_{k=1}^{K_1} \\
        \text{col}\{-\nabla_x J_k(\boldsymbol{x}_{c, i-1}, \boldsymbol{y}_{c, i-1})\}_{k=K_1+1}^{K}
    \end{bmatrix}  - (\mathbbm{1}_K \right. \notag\\
    &\quad\left. {p^{(x)}}^{\top})\otimes I_{M_1} \begin{bmatrix}
        \text{col}\{\nabla_x J_k(\boldsymbol{x}_{c, i-1}, \boldsymbol{y}_{c, i-1})\}_{k=1}^{K_1} \\
        \text{col}\{-\nabla_x J_k(\boldsymbol{x}_{c, i-1}, \boldsymbol{y}_{c, i-1})\}_{k=K_1+1}^{K}
    \end{bmatrix}\right)\Bigg\|^2 \notag\\
    &\stackrel{(b)}{\leq} 2\|\mathcal{P}_R^{\top}\|^2 \sum_{k=1}^{K} \|\nabla_x J_k(\boldsymbol{x}_{k, i-1}, \boldsymbol{y}_{k, i-1}) - \nabla_x J_k(\boldsymbol{x}_{c, i-1},  \notag\\
    &\quad \boldsymbol{y}_{c, i-1})\|^2 + 2\|\mathcal{P}_R^{\top}\|^2 \sum_{k=1}^{K_1} \|\nabla_x J_k(\boldsymbol{x}_{c, i-1}, \boldsymbol{y}_{c, i-1})  \notag\\
    &\quad - \nabla_x J^{(1)}(\boldsymbol{x}_{c, i-1}, \boldsymbol{y}_{c, i-1})\|^2 + 2\|\mathcal{P}_R^{\top}\|^2 \sum_{k=K_1+1}^{K}  \notag\\
    & \quad \|-\nabla_x J_k(\boldsymbol{x}_{c, i-1}, \boldsymbol{y}_{c, i-1}) + \nabla_x J^{(2)}(\boldsymbol{x}_{c, i-1}, \boldsymbol{y}_{c, i-1})\|^2 \notag\\
    &\stackrel{(c)}{\leq} 2\|\mathcal{P}_R^{\top}\|^2 \sum_{k=1}^{K} L_f^2 (\|\boldsymbol{x}_{k, i-1} - \boldsymbol{x}_{c, i-1}\| + \|\boldsymbol{y}_{k, i-1} - \boldsymbol{y}_{c, i-1}\|)^2 \notag\\
    &\quad + 2\|\mathcal{P}_R^{\top}\|^2 K G^2 \notag\\
    &\stackrel{(d)}{\leq} 2\|\mathcal{P}_R^{\top}\|^2 L_f^2 (2\|\boldsymbol{\mathcal{X}}^1_{i-1} - \boldsymbol{\mathcal{X}}_{c,i-1}\|^2 + 2\|\boldsymbol{\mathcal{Y}}^2_{i-1} - \boldsymbol{\mathcal{Y}}_{c,i-1}\|^2 \notag\\
    &\quad + 2\|\boldsymbol{\mathcal{X}}^2_{i-1} - \boldsymbol{\mathcal{X}}^\prime_{c,i-1}\|^2 + 2\|\boldsymbol{\mathcal{Y}}^1_{i-1} - \boldsymbol{\mathcal{Y}}^\prime_{c,i-1}\|^2) \notag\\
    &\quad + 2\|\mathcal{P}_R^{\top}\|^2 K G^2.
\end{align}
where $(a)$ follows from the fact that $P_R^{\top} \mathbbm{1}_K = 0$; and $(b)$ is due to 
    \begin{equation}
        \begin{aligned}
            &(\mathbbm{1}_K {p^{(x)}}^{\top}) \otimes I_{M_1} \begin{bmatrix}
                \text{col}\{\nabla_x J_k\left(\boldsymbol{x}_{c, i-1}, \boldsymbol{y}_{c, i-1}\right)\}_{k=1}^{K_1} \\ \text{col}\{ - \nabla_x J_k\left(\boldsymbol{x}_{c, i-1}, \boldsymbol{y}_{c, i-1}\right)\}_{k=K_1+1}^{K}
                \end{bmatrix} \\
                &= \mathbbm{1}_K \otimes \sum_{k=1}^{K_1} p_k \nabla_x J_k\left(\boldsymbol{x}_{c, i-1}, \boldsymbol{y}_{c, i-1}\right)\\
                &= \mathbbm{1}_K \otimes \nabla_x \left(\sum_{k=1}^{K_1} p_k J_k\left(\boldsymbol{x}_{c, i-1}, \boldsymbol{y}_{c, i-1}\right)\right)\\
                &= \mathbbm{1}_K \otimes \nabla_x  J^{(1)}\left(\boldsymbol{x}_{c, i-1}, \boldsymbol{y}_{c, i-1}\right)
        \end{aligned}
    \end{equation}
    and the zero-sum setting \eqref{eq:zero-sum}, which implies $\nabla_x J^{(1)}\left(\boldsymbol{x}_{c, i-1}, \boldsymbol{y}_{c, i-1}\right) = - \nabla_x J^{(2)}\left(\boldsymbol{x}_{c, i-1}, \boldsymbol{y}_{c, i-1}\right)$; $(c)$ follows from Assumption \ref{ass:lip} and \ref{ass:bdis}; $(d)$ follows from the inequality $(a+b)^2 \leq 2a^2 + 2b^2$ and the fact that 
    \begin{equation}
        \begin{aligned}
            \sum_{k=1}^{K_1} \|\boldsymbol{x}_{k, i-1} - \boldsymbol{x}_{c, i-1}\|^2 &= \|\boldsymbol{\mathcal{X}}^1_{i-1} - \boldsymbol{\mathcal{X}}_{c, i-1} \|^2 \\
            \sum_{k=K_1+1}^{K_1} \|\boldsymbol{x}_{k, i-1} - \boldsymbol{x}_{c, i-1}\|^2 &= \|\boldsymbol{\mathcal{X}}^2_{i-1} - \boldsymbol{\mathcal{X}}^\prime_{c, i-1} \|^2\\
            \sum_{k=1}^{K_1}\|\boldsymbol{y}_{k, i-1} - \boldsymbol{y}_{c, i-1}\|^2 &= \|\boldsymbol{\mathcal{Y}}^1_{i-1} - \boldsymbol{\mathcal{Y}}^{\prime}_{c,i-1}\|^2\\
            \sum_{k=K_1+1}^{K}\|\boldsymbol{y}_{k, i-1} - \boldsymbol{y}_{c, i-1}\|^2 &= \|\boldsymbol{\mathcal{Y}}^2_{i-1} - \boldsymbol{\mathcal{Y}}_{c,i-1}\|^2   
        \end{aligned}
    \end{equation}
    Moreover, by Lemma \ref{lemma:X1iX2i}, we can further bound \eqref{eq:PRTG} by the following:
    \begin{equation}
    \label{eq:PRTG_2}
        \begin{aligned}
            &\left\|\mathcal{P}_R^{\top} \begin{bmatrix}\mathcal{G}^1_{x, i} \\\mathcal{G}^2_{x, i}\end{bmatrix} \right\|^2\\
            &\leq 
            2\|\mathcal{P}_R^{\top}\|^2 L_f^2 (2\|\boldsymbol{\mathcal{X}}^1_{i-1} - \boldsymbol{\mathcal{X}}_{c,i-1}\|^2 + 2\|\boldsymbol{\mathcal{Y}}^2_{i-1} - \boldsymbol{\mathcal{Y}}_{c,i-1}\|^2 \\
            & \quad + 2\|\boldsymbol{\mathcal{X}}^2_{i-1} - \boldsymbol{\mathcal{X}}^\prime_{c,i-1}\|^2 + 2\|\boldsymbol{\mathcal{Y}}^1_{i-1} - \boldsymbol{\mathcal{Y}}^\prime_{c,i-1}\|^2 )\\
            & \quad 
            + 2\|\mathcal{P}_R^{\top}\|^2 K G^2\\
            & \leq 
            4 \|\mathcal{P}_R^{\top}\|^2 L_f^2 \left( \|\mathcal{P}_L\|^2\left\| \mathcal{P}_R^{\top} \begin{bmatrix}
    \boldsymbol{\mathcal{X}}^1_{i-1} \\
    \boldsymbol{\mathcal{X}}^2_{i-1}
\end{bmatrix}\right\|^2  \right.\\
& \quad \left. + \|\mathcal{Q}_L\|^2 \left\| \mathcal{Q}_R^{\top} \begin{bmatrix}
    \boldsymbol{\mathcal{Y}}^1_{i-1} \\
    \boldsymbol{\mathcal{Y}}^2_{i-1}
\end{bmatrix}\right\|^2\right) + 2\|\mathcal{P}_R^{\top}\|^2 K G^2
        \end{aligned}
    \end{equation}
    Combining \eqref{eq:PRTG_2} with \eqref{eq:exp_PRTG}, we arrive at:
    \begin{align}
    &\mathbb{E} \left\{\left\|\mathcal{P}_R^{\top} \begin{bmatrix}
        \boldsymbol{\mathcal{G}}^1_{x, i} \\
        \boldsymbol{\mathcal{G}}^2_{x, i}
    \end{bmatrix}\right\|^2\right\} \notag\\
    &\leq 
    8 \|\mathcal{P}_R^{\top}\|^2 L_f^2 \left( \|\mathcal{P}_L\|^2 \mathbb{E} \left\{\left\| \mathcal{P}_R^{\top} \begin{bmatrix}
    \boldsymbol{\mathcal{X}}^1_{i-1} \\
    \boldsymbol{\mathcal{X}}^2_{i-1}
\end{bmatrix}\right\|^2\right\}  \right. \notag\\
& \quad \left. + \|\mathcal{Q}_L\|^2 \mathbb{E} \left\{\left\| \mathcal{Q}_R^{\top} \begin{bmatrix}
    \boldsymbol{\mathcal{Y}}^1_{i-1} \\
    \boldsymbol{\mathcal{Y}}^2_{i-1}
\end{bmatrix}\right\|^2\right\}\right) + 2\|\mathcal{P}_R^{\top}\|^2 K (2G^2 + \sigma^2)
    \end{align}
\end{proof}
\vspace{-1.8em}
\subsection{Main Proof}
\label{proof:lemma2}
From Algorithm \ref{alg:network_learning_zerosum}, the recursions for $\boldsymbol{\mathcal{X}}^1_{i}$  and $\boldsymbol{\mathcal{X}}^2_{i}$ are given by:
\begin{align}
    \boldsymbol{\mathcal{X}}^1_{i} &= {\mathcal{A}^{(1)}}^{\top} \left(\boldsymbol{\mathcal{X}}^1_{i-1} - \mu \boldsymbol{\mathcal{G}}^1_{x, i} \right)\\
    \boldsymbol{\mathcal{X}}^2_{i} &= {\mathcal{C}^{(12)}}^{\top} \boldsymbol{\mathcal{X}}_{i}^1 +{\mathcal{C}^{(2)}}^{\top} \left(\boldsymbol{\mathcal{X}}^2_{i-1} - \mu \boldsymbol{\mathcal{G}}^2_{x, i}\right)
\end{align}
In matrix form we write:
    \begin{equation}
        \begin{bmatrix}
            \boldsymbol{\mathcal{X}}^1_{i} \\
            \boldsymbol{\mathcal{X}}^2_{i}
        \end{bmatrix} =  
        \begin{bmatrix}
            {\mathcal{A}^{(1)}}^{\top} & 0 \\
            {\mathcal{C}^{(12)}}^{\top}{\mathcal{A}^{(1)}}^{\top} & {\mathcal{C}^{(2)}}^{\top} 
        \end{bmatrix}
        \begin{bmatrix}
            \boldsymbol{\mathcal{X}}^1_{i-1} - \mu \boldsymbol{\mathcal{G}}^1_{x, i} \\
            \boldsymbol{\mathcal{X}}^2_{i-1} - \mu \boldsymbol{\mathcal{G}}^2_{x, i}
        \end{bmatrix}
    \end{equation} 

We define the transition matrix for $x$ as:
\begin{equation}
    B^{(x)} \triangleq 
    \begin{bmatrix}
    {A}^{(1)} & A^{(1)} C^{(12)}\\
    0 & {C^{(2)}} 
\end{bmatrix} \in \mathbb{R}^{K \times K}
\end{equation}
where the superscript $(x)$ indicates that this matrix is related to the information flow of strategy $x$. From Lemma \ref{lemma:BxBy}, we know that $B^{(x)}$ and $B^{(y)}$ have Jordan decompositions: $B^{(x)} = P M P^{-1}$, $B^{(y)} = Q N Q^{-1}$. Formal definitions of $\mathcal{P}_R, \mathcal{P}_L, \mathcal{M}_\eta, \mathcal{Q}_R, \mathcal{Q}_L, \mathcal{N}_\iota$ can be found in Remark \ref{remark:jordan}.

By Lemma \ref{lemma:X1iX2i}, we can obtain:
\begin{align}
    &\mathbb{E} \left\{\| \boldsymbol{\mathcal{X}}^1_{i} - \boldsymbol{\mathcal{X}}_{c, i} \|^2 + \| \boldsymbol{\mathcal{X}}^2_{i} - \boldsymbol{\mathcal{X}}^\prime_{c, i} \|^2\right\} \notag\\
    &\quad \leq \|\mathcal{P}_L\|^2 \mathbb{E} \left\{\left\| \mathcal{P}_R^{\top} \begin{bmatrix}
    \boldsymbol{\mathcal{X}}^1_{i} \\
    \boldsymbol{\mathcal{X}}^2_{i}
    \end{bmatrix}\right\|^2\right\}
\end{align}
\vspace{-1.3em}
Next,
\vspace{-1em}
\allowdisplaybreaks

\begin{align}
\label{eq:PRTX}
    &\left\| \mathcal{P}_R^{\top} \begin{bmatrix}
            \boldsymbol{\mathcal{X}}^1_{i} \\
            \boldsymbol{\mathcal{X}}^2_{i}
        \end{bmatrix}\right\|^2 \notag \\
    &= \left\| \mathcal{P}_R^{\top} {\mathcal{B}^{(x)}}^{\top} \begin{bmatrix}
    \boldsymbol{\mathcal{X}}^1_{i-1} - \mu \boldsymbol{\mathcal{G}}^1_{x, i} \\
    \boldsymbol{\mathcal{X}}^2_{i-1} - \mu \boldsymbol{\mathcal{G}}^2_{x, i}
\end{bmatrix}\right\|^2 \notag \\
    &\stackrel{(a)}{=} \left\| \mathcal{M}_\eta^{\top} \mathcal{P}_R^{\top} \begin{bmatrix}
    \boldsymbol{\mathcal{X}}^1_{i-1} - \mu \boldsymbol{\mathcal{G}}^1_{x, i} \\
    \boldsymbol{\mathcal{X}}^2_{i-1} - \mu \boldsymbol{\mathcal{G}}^2_{x, i}
\end{bmatrix}\right\|^2 \notag \\
    &\stackrel{(b)}{\leq} \left\| \mathcal{M}_\eta^{\top}\right\|^2 \left(\frac{1}{\left\| \mathcal{M}_\eta^{\top}\right\|}\left\|\mathcal{P}_R^{\top} \begin{bmatrix}
    \boldsymbol{\mathcal{X}}^1_{i-1}  \\
    \boldsymbol{\mathcal{X}}^2_{i-1}
\end{bmatrix}\right\|^2 \right. \notag \\
    &\quad \left. + \frac{1}{1-\left\| \mathcal{M}_\eta^{\top}\right\|} \mu^2 \left\|\mathcal{P}_R^{\top} \begin{bmatrix}
    \boldsymbol{\mathcal{G}}^1_{x, i} \\
    \boldsymbol{\mathcal{G}}^2_{x, i}
\end{bmatrix}\right\|^2\right) \notag \\
    &= \left\| \mathcal{M}_\eta^{\top}\right\| \left\|\mathcal{P}_R^{\top} \begin{bmatrix}
    \boldsymbol{\mathcal{X}}^1_{i-1}  \\
    \boldsymbol{\mathcal{X}}^2_{i-1}
\end{bmatrix}\right\|^2 \notag \\
    &\quad + \frac{\left\| \mathcal{M}_\eta^{\top}\right\|^2}{1-\left\| \mathcal{M}_\eta^{\top}\right\|} \mu^2 \left\|\mathcal{P}_R^{\top} \begin{bmatrix}
    \boldsymbol{\mathcal{G}}^1_{x, i} \\
    \boldsymbol{\mathcal{G}}^2_{x, i}
\end{bmatrix}\right\|^2.
\end{align}
    where  $(a)$ follows from 
    \begin{equation}
    \label{eq:PRB}
        \begin{aligned}
            &(P_{R}^{\top} {B^{(x)}}^{\top}) \otimes I_{M_1} \\
            &= \left(P_{R}^{\top} \begin{bmatrix} \mathbbm{1}_K & P_{L} \end{bmatrix} \begin{bmatrix} 1 & 0 \\ 0 & M_{\eta}^{\top} \end{bmatrix} \begin{bmatrix} {p^{(x)}}^{\top} \\ P_{R}^{\top} \end{bmatrix}\right) \otimes I_{M_1} \\
            &= \left(\begin{bmatrix} 0 & I_{K-1} \end{bmatrix} \begin{bmatrix} 1 & 0 \\ 0 & M_{\eta}^{\top} \end{bmatrix} \begin{bmatrix} {p^{(x)}}^{\top} \\ P_{R}^{\top} \end{bmatrix}\right) \otimes I_{M_1} \\
            &= (M_{\eta}^{\top} P_{R}^{\top}) \otimes I_{M_1}.
        \end{aligned}
    \end{equation}
and $(b)$ follows from the sub-multiplicative property of $\|\cdot\|$, the convexity of $\|\cdot\|^2$, and Jensen's inequality: $\|a+b\|^2 \leq \frac{1}{\alpha} \|a\|^2 + \frac{1}{1-\alpha} \|b\|^2$. 
Combining the results from \eqref{eq:PRTX} and Lemma \ref{lemma:PRTG}, we have:
    \begin{align}
    \label{eq:PRTX_final}
            &\mathbb{E} \left\{\left\| \mathcal{P}_R^{\top} \begin{bmatrix}
            \boldsymbol{\mathcal{X}}^1_{i} \\
            \boldsymbol{\mathcal{X}}^2_{i}
        \end{bmatrix}\right\|^2\right\} \notag\\
            &\stackrel{(a)}{\leq}
            \left\| \mathcal{M}_\eta^{\top}\right\| \mathbb{E} \left\{\left\|\mathcal{P}_R^{\top} \begin{bmatrix}\boldsymbol{\mathcal{X}}^1_{i-1}  \\\boldsymbol{\mathcal{X}}^2_{i-1} \end{bmatrix} \right\|^2\right\} \notag\\
            & \quad + \frac{\left\| \mathcal{M}_\eta^{\top}\right\|^2}{1-\left\| \mathcal{M}_\eta^{\top}\right\|} \mu^2 \mathbb{E} \left\{\left\|\mathcal{P}_R^{\top} \begin{bmatrix}\boldsymbol{\mathcal{G}}^1_{x, i} \\\boldsymbol{\mathcal{G}}^2_{x, i}\end{bmatrix} \right\|^2\right\} \notag\\
            &\stackrel{(b)}{\leq}
            \left\| \mathcal{M}_\eta^{\top}\right\| \mathbb{E} \left\{\left\|\mathcal{P}_R^{\top} \begin{bmatrix}\boldsymbol{\mathcal{X}}^1_{i-1}  \\\boldsymbol{\mathcal{X}}^2_{i-1} \end{bmatrix} \right\|^2\right\} + 8\frac{\left\| \mathcal{M}_\eta^{\top}\right\|^2 \|\mathcal{P}_R^{\top}\|^2 \|\mathcal{P}_L\|^2}{1-\left\| \mathcal{M}_\eta^{\top}\right\|} \notag\\
            & \quad L_f^2 \mu^2 \mathbb{E} \left\{\left\|\mathcal{P}_R^{\top} \begin{bmatrix}\boldsymbol{\mathcal{X}}^1_{i-1}  \\\boldsymbol{\mathcal{X}}^2_{i-1} \end{bmatrix} \right\|^2\right\}
            + 8\frac{\left\| \mathcal{M}_\eta^{\top}\right\|^2 \|\mathcal{P}_R^{\top}\|^2 \|\mathcal{Q}_L\|^2}{1-\left\| \mathcal{M}_\eta^{\top}\right\|} L_f^2 \mu^2 \notag\\
            & \quad \mathbb{E} \left\{\left\| \mathcal{Q}_R^{\top} \begin{bmatrix}
    \boldsymbol{\mathcal{Y}}^1_{i-1} \\
    \boldsymbol{\mathcal{Y}}^2_{i-1}
\end{bmatrix}\right\|^2\right\} + 2\frac{\left\| \mathcal{M}_\eta^{\top}\right\|^2\left\|\mathcal{P}_R^{\top}\right\|^2}{1-\left\| \mathcal{M}_\eta^{\top}\right\|} \mu^2 K (2G ^2 + \sigma^2)
        \end{align}
    where $(a)$ follows from \eqref{eq:PRTX}; $(b)$ follows from Lemma \ref{lemma:PRTG}. Similarly, we get:
    \begin{equation}
    \label{eq:QRTY}
        \begin{aligned}
           &\mathbb{E} \left\{\left\| \mathcal{Q}_R^{\top} \begin{bmatrix}\boldsymbol{\mathcal{Y}}^1_{i} \\\boldsymbol{\mathcal{Y}}^2_{i}\end{bmatrix}\right\|^2\right\} \\
            &\leq
            \left\| \mathcal{N}_\iota^{\top}\right\| \mathbb{E} \left\{\left\| \mathcal{Q}_R^{\top} \begin{bmatrix}\boldsymbol{\mathcal{Y}}^1_{i-1} \\\boldsymbol{\mathcal{Y}}^2_{i-1}\end{bmatrix}\right\|^2\right\} + 8\frac{\left\| \mathcal{N}_\iota^{\top}\right\|^2 \|\mathcal{Q}_R^{\top}\|^2 \|\mathcal{Q}_L\|^2}{1-\left\| \mathcal{N}_\iota^{\top}\right\|} \\
            & \quad L_f^2 \mu^2 \mathbb{E} \left\{\left\| \mathcal{Q}_R^{\top} \begin{bmatrix}\boldsymbol{\mathcal{Y}}^1_{i-1} \\\boldsymbol{\mathcal{Y}}^2_{i-1}\end{bmatrix}\right\|^2\right\}
            + 8\frac{\left\| \mathcal{N}_\iota^{\top}\right\|^2 \|\mathcal{Q}_R^{\top}\|^2 \|\mathcal{P}_L\|^2}{1-\left\| \mathcal{N}_\iota^{\top}\right\|} L_f^2 \mu^2 \\
            & \quad \mathbb{E} \left\{\left\|\mathcal{P}_R^{\top} \begin{bmatrix}\boldsymbol{\mathcal{X}}^1_{i-1}  \\\boldsymbol{\mathcal{X}}^2_{i-1} \end{bmatrix} \right\|^2\right\} + 2\frac{\left\| \mathcal{N}_\iota^{\top}\right\|^2\left\|\mathcal{Q}_R^{\top}\right\|^2}{1-\left\| \mathcal{N}_\iota^{\top}\right\|} \mu^2 K (2G ^2 + \sigma^2)
        \end{aligned}
    \end{equation}
    For compactness, we introduce the following scalar coefficients:
    \begin{align}
    a_1 &= \frac{\left\| \mathcal{M}_\eta^{\top}\right\|^2 \|\mathcal{P}_R^{\top}\|^2 \|\mathcal{P}_L\|^2}{1-\left\| \mathcal{M}_\eta^{\top}\right\|}\\
    b_1 &= \frac{\left\| \mathcal{M}_\eta^{\top}\right\|^2 \|\mathcal{P}_R^{\top}\|^2 \|\mathcal{Q}_L\|^2}{1-\left\| \mathcal{M}_\eta^{\top}\right\|}\\
    c_1 &= \frac{\left\| \mathcal{M}_\eta^{\top}\right\|^2\left\|\mathcal{P}_R^{\top}\right\|^2}{1-\left\| \mathcal{M}_\eta^{\top}\right\|}\\
    a_2 &= \frac{\left\| \mathcal{N}_\iota^{\top}\right\|^2 \|\mathcal{Q}_R^{\top}\|^2 \|\mathcal{Q}_L\|^2}{1-\left\| \mathcal{N}_\iota^{\top}\right\|}\\
    b_2 &= \frac{\left\| \mathcal{N}_\iota^{\top}\right\|^2 \|\mathcal{Q}_R^{\top}\|^2 \|\mathcal{P}_L\|^2}{1-\left\| \mathcal{N}_\iota^{\top}\right\|}\\
    c_2 &= \frac{\left\| \mathcal{N}_\iota^{\top}\right\|^2\left\|\mathcal{Q}_R^{\top}\right\|^2}{1-\left\| \mathcal{N}_\iota^{\top}\right\|}
\end{align}
Adding the results from \eqref{eq:PRTX_final} and \eqref{eq:QRTY}, we obtain:
    \begin{equation}
        \begin{aligned}
            &\mathbb{E} \left\{\left\| \mathcal{P}_R^{\top} \begin{bmatrix}
            \boldsymbol{\mathcal{X}}^1_{i} \\
            \boldsymbol{\mathcal{X}}^2_{i}
        \end{bmatrix}\right\|^2\right\} + \mathbb{E} \left\{\left\| \mathcal{Q}_R^{\top} \begin{bmatrix}\boldsymbol{\mathcal{Y}}^1_{i} \\\boldsymbol{\mathcal{Y}}^2_{i}\end{bmatrix}\right\|^2\right\}\\
            & \leq
            \left(\left\| \mathcal{M}_\eta^{\top}\right\| + 8 a_1 L_f^2 \mu^2 + 8 b_2 L_f^2 \mu^2\right) \mathbb{E} \left\{\left\|\mathcal{P}_R^{\top} \begin{bmatrix}\boldsymbol{\mathcal{X}}^1_{i-1}  \\\boldsymbol{\mathcal{X}}^2_{i-1} \end{bmatrix} \right\|^2\right\} \\
            & \quad
            + \left( \left\|\mathcal{N}_\iota^{\top}\right\| + 8 a_2 L_f^2 \mu^2 + 8 b_1 L_f^2 \mu^2\right) \mathbb{E} \left\{\left\| \mathcal{Q}_R^{\top} \begin{bmatrix}\boldsymbol{\mathcal{Y}}^1_{i-1} \\\boldsymbol{\mathcal{Y}}^2_{i-1}\end{bmatrix}\right\|^2\right\} \\
            & \quad + 2(c_1 + c_2)\mu^2 K (2G ^2 + \sigma^2)\\
            & \stackrel{(a)}{\leq}
            \alpha \left(\mathbb{E} \left\{\left\|\mathcal{P}_R^{\top} \begin{bmatrix}\boldsymbol{\mathcal{X}}^1_{i-1}  \\\boldsymbol{\mathcal{X}}^2_{i-1} \end{bmatrix} \right\|^2\right\} + \mathbb{E} \left\{\left\| \mathcal{Q}_R^{\top} \begin{bmatrix}\boldsymbol{\mathcal{Y}}^1_{i-1} \\\boldsymbol{\mathcal{Y}}^2_{i-1}\end{bmatrix}\right\|^2\right\}\right) \\
            & \quad + 2(c_1 + c_2)\mu^2 K (2G ^2 + \sigma^2)\\
        \end{aligned}   
    \end{equation}
    where  in $(a)$ we define:
    \begin{align}
    \label{eq:zerosum_alpha}
        \alpha &= \max \bigl\{\left\| \mathcal{M}_\eta^{\top}\right\| + 8 a_1 L_f^2 \mu^2 + 8 b_2 L_f^2 \mu^2, \notag\\
        & \qquad \left\|\mathcal{N}_\iota^{\top}\right\| + 8 a_2 L_f^2 \mu^2 + 8 b_1 L_f^2 \mu^2\bigr\} < 1
    \end{align}
    Noting that $\left\| \mathcal{M}_\eta^{\top}\right\| < 1$ is independent of $\mu$, we have $1 - \alpha = O(1)$. Thus, 
\begin{equation}
    \begin{aligned}
        &\mathbb{E} \left\{\left\| \mathcal{P}_R^{\top} \begin{bmatrix}
            \boldsymbol{\mathcal{X}}^1_{i} \\
            \boldsymbol{\mathcal{X}}^2_{i}
        \end{bmatrix}\right\|^2\right\} + \mathbb{E} \left\{\left\| \mathcal{Q}_R^{\top} \begin{bmatrix}\boldsymbol{\mathcal{Y}}^1_{i} \\\boldsymbol{\mathcal{Y}}^2_{i}\end{bmatrix}\right\|^2\right\} \\
        & \stackrel{(a)}{\leq}
        O(\mu^2) + \frac{2(c_1 + c_2)\mu^2 K (2G ^2 + \sigma^2)}{1-\alpha} = O(\mu^2)
    \end{aligned}
\end{equation}
where $(a)$ holds when 
\begin{equation}
    \alpha^i \leq O(\mu^2) \Longleftrightarrow i \geq i_\alpha =\frac{\log \left(O\left(\mu^2\right)\right)}{\log \left(\alpha\right)}
\end{equation}
Therefore, by Lemma \ref{lemma:X1iX2i}, we conclude: 
\begin{equation}
\begin{aligned}
    &\mathbb{E} \{\|\boldsymbol{\mathcal{X}}^1_{i} - \boldsymbol{\mathcal{X}}_{c,i}\|^2 + \|\boldsymbol{\mathcal{Y}}^2_{i} - \boldsymbol{\mathcal{Y}}_{c,i}\|^2 + \\
    & \qquad \|\boldsymbol{\mathcal{X}}^2_{i} - \boldsymbol{\mathcal{X}}^\prime_{c,i}\|^2 + \|\boldsymbol{\mathcal{Y}}^1_{i} - \boldsymbol{\mathcal{Y}}^\prime_{c,i}\|^2\} \\
    &\leq
    \max\{\|\mathcal{P}_L\|^2, \|\mathcal{Q}_L\|^2 \}\\
    &\quad \left(\mathbb{E} \left\{\left\| \mathcal{P}_R^{\top} \begin{bmatrix}
            \boldsymbol{\mathcal{X}}^1_{i} \\
            \boldsymbol{\mathcal{X}}^2_{i}
        \end{bmatrix}\right\|^2\right\} + \mathbb{E} \left\{\left\| \mathcal{Q}_R^{\top} \begin{bmatrix}\boldsymbol{\mathcal{Y}}^1_{i} \\\boldsymbol{\mathcal{Y}}^2_{i}\end{bmatrix}\right\|^2\right\}\right) \\
        &\leq O(\mu^2) 
\end{aligned}
\end{equation}
when $i \geq i_\alpha$.

\section{Proof of Lemma \ref{lemma:zerosum_within_cross} for \textbf{ATC-C}}
\label{proof:lemma2_atcc}
In this section, we prove Lemma \ref{lemma:zerosum_within_cross} when employing \mbox{\textbf{ATC-C}.}

Note that by Assumption \ref{ass:within_conn}, 
$A^{(1)}$
and $A^{(2)}$
have Jordan decompositions 
$A^{(1)} = V J V^{-1}$ and $A^{(2)} = U K U^{-1}$ \cite{sayed2014adaptation}, where 
\begin{equation}
    V = \begin{bmatrix}
        p^{(1)} & V_{R}
    \end{bmatrix}, \
    J = \begin{bmatrix}
        1 & 0 \\
        0 & J_{\varepsilon}
    \end{bmatrix}, \
    V^{-1} = \begin{bmatrix}
        \mathbbm{1}^{\top} \\
        V_{L}^{\top}
    \end{bmatrix}
\end{equation}
\begin{equation}
    U = \begin{bmatrix}
        p^{(2)} & U_{R}
    \end{bmatrix}, \
    K = \begin{bmatrix}
        1 & 0 \\
        0 & K_{\gamma}
    \end{bmatrix}, \
    U^{-1} = \begin{bmatrix}
        \mathbbm{1}^{\top} \\
        U_{L}^{\top}
    \end{bmatrix}
\end{equation}
Here, $J_{\varepsilon}, K_\gamma$ are 
 block Jordan matrices with the eigenvalues on the diagonal and $\varepsilon, \gamma$ on the first lower subdiagonal, respectively. Using the above notation,
we can verify that 
\begin{align}
    \mathbbm{1}^{\top}_{K_1} V_{R} = 0, \mathbbm{1}^{\top}_{K_2} U_{R} = 0, V_{L}^{\top} V_{R} = I, U_{L}^{\top} U_{R} = I
\end{align}
We also define
\begin{align}
\mathcal{V}_{L} &\triangleq V_{L} \otimes I_{M_1},\!\! &
\mathcal{J}_{\varepsilon} &\triangleq J_{\varepsilon} \otimes I_{M_1},\!\! &
\mathcal{V}_{R} &\triangleq V_{R} \otimes I_{M_1} \\
\mathcal{U}_{L} &\triangleq U_{L} \otimes I_{M_2},\!\! &
\mathcal{K}_{\gamma} &\triangleq K_{\gamma} \otimes I_{M_2},\!\! &
\mathcal{U}_{R} &\triangleq U_{R} \otimes I_{M_2}
\end{align}

\subsection{Key Lemmas}
\begin{lemma}
\label{lemma:plain_within}
    Under Assumptions \ref{ass:within_conn}, for \textbf{ATC-C}, we have:
\begin{align}
    & \|\boldsymbol{\mathcal{X}}^1_{i} - \boldsymbol{\mathcal{X}}_{c,i}\|^2 \leq \|\mathcal{V}_{L} \|^2 \| \mathcal{V}_{R}^{\top} \boldsymbol{\mathcal{X}}^1_i \|^2  \\
    & \|\boldsymbol{\mathcal{Y}}^2_{i} - \boldsymbol{\mathcal{Y}}_{c,i}\|^2 \leq
    \|\mathcal{U}_L\|^2  \left\|\mathcal{U}_{R}^{\top}\boldsymbol{\mathcal{Y}}^2_{i}\right\|^2
\end{align}
\end{lemma}

\begin{proof}
    We prove for $\|\boldsymbol{\mathcal{X}}^1_{i} - \boldsymbol{\mathcal{X}}_{c,i}\|^2$; the proof for $\|\boldsymbol{\mathcal{Y}}^2_{i} - \boldsymbol{\mathcal{Y}}_{c,i}\|^2$ follows a similar argument.
    For $\| \boldsymbol{\mathcal{X}}^1_i - \boldsymbol{\mathcal{X}}_{c, i} \|^2$, we have:
    \begin{align}
    \label{eq:lemma1_first}
        &\| \boldsymbol{\mathcal{X}}^1_i - (\mathbbm{1} {p^{(1)}}^{\top} \otimes I_{M_1})\boldsymbol{ \mathcal{X}}^1_{i} \|^2 \notag\\
        &= \| (I_{K_1} \otimes I_{M_1})\boldsymbol{\mathcal{X}}^1_i - (\mathbbm{1} {p^{(1)}}^{\top} \otimes I_{M_1})\boldsymbol{ \mathcal{X}}^1_{i} \|^2 \notag\\
        & \stackrel{(a)}{=} \|\mathcal{V}_{L} \mathcal{V}_{R}^{\top} \boldsymbol{\mathcal{X}}^1_i\|^2 \notag\\
        & \leq \|\mathcal{V}_{L} \|^2 \| \mathcal{V}_{R}^{\top} \boldsymbol{\mathcal{X}}^1_i \|^2
    \end{align}
    where $(a)$ follows from the fact that 
    \begin{align}
        (I_{K_1} \otimes I_{M_1}) &= (V^{-1} \otimes I_{M_1})^{\top} (V \otimes I_{M_1})^{\top} \notag\\
        &= \left(\mathbbm{1} {p^{(1)}}^{\top} + V_{L} V_{R}^{\top}\right) \otimes I_{M_1}
    \end{align}
\end{proof}

\begin{lemma}
\label{lemma:plain_cross}
    Under Assumption \ref{ass:within_conn}, \ref{ass:cross-conn_strong}, for \textbf{ATC-C}, we have:
    \begin{align}    &\|\boldsymbol{\mathcal{X}}^2_{i} - \boldsymbol{\mathcal{X}}^\prime_{c,i}\|^2 \leq K_2\|\mathcal{V}_{L} \|^2 \| \mathcal{V}_{R}^{\top} \boldsymbol{\mathcal{X}}^1_i \|^2\\    &\|\boldsymbol{\mathcal{Y}}^1_{i} - \boldsymbol{\mathcal{Y}}^\prime_{c,i}\|^2 \leq K_1 \|\mathcal{U}_L\|^2  \left\|\mathcal{U}_{R}^{\top}\boldsymbol{\mathcal{Y}}^2_{i}\right\|^2
    \end{align} \hfill\qed
\end{lemma}

\begin{proof}
    We only prove for $\|\boldsymbol{\mathcal{Y}}^1_{i} - \boldsymbol{\mathcal{Y}}^\prime_{c,i}\|^2$ as the proof for $\|\boldsymbol{\mathcal{X}}^2_{i} - \boldsymbol{\mathcal{X}}^\prime_{c,i}\|^2$ is nearly the same. As stated in Assumption \ref{ass:cross-conn_strong}, we know that $C^{(21)}$ is left-stochastic, meaning that ${C^{(21)}}^{\top} \mathbbm{1}_{K_2} = \mathbbm{1}_{K_1}$. Consequently, we can derive the following:
    \begin{equation}
        \begin{aligned}  \boldsymbol{\mathcal{Y}}^\prime_{c,i-1} 
        &= \left(\mathbbm{1}_{K_1} {p^{(2)}}^{\top} \otimes I_{M_2}\right) \boldsymbol{\mathcal{Y}}^2_{i-1} \\
        &= 
        \left({C^{(21)}}^{\top} \mathbbm{1}_{K_2} {p^{(2)}}^{\top} \otimes I_{M_2}\right) \boldsymbol{\mathcal{Y}}^2_{i-1} \\
        &=
        \left({C^{(21)}}^{\top} \otimes I_{M_2}\right) \left(\mathbbm{1}_{K_2} {p^{(2)}}^{\top} \otimes I_{M_2}\right) \boldsymbol{\mathcal{Y}}^2_{i-1} \\
        &=
        {\mathcal{C}^{(21)}}^{\top} \boldsymbol{\mathcal{Y}}_{c,i-1}.
        \end{aligned}
    \end{equation}
    Moreover, from Algorithm \ref{alg:network_learning_general_strong}, we have $\boldsymbol{\mathcal{Y}}^1_{i-1} = {\mathcal{C}^{(21)}}^{\top} \boldsymbol{\mathcal{Y}}^2_{i-1}$. As a result, we get
    \begin{equation}
        \boldsymbol{\mathcal{Y}}^1_{i-1} -\boldsymbol{\mathcal{Y}}^\prime_{c,i-1} = {\mathcal{C}^{(21)}}^{\top} (\boldsymbol{\mathcal{Y}}^2_{i-1} - \boldsymbol{\mathcal{Y}}_{c,i-1})
    \end{equation}
    Then, by Assumption \ref{ass:cross-conn_strong}, we get:
    \begin{align}
    \label{eq:C21}
        \|{\mathcal{C}^{(21)}}^{\top}\|^2 &= \rho (C^{(21)} {C^{(21)}}^{\top}) \leq \|C^{(21)} {C^{(21)}}^{\top}\|_1 \notag\\
        &\leq \sum_{ij} (C^{(21)} {C^{(21)}}^{\top})_{i,j} = \mathbbm{1}^{\top} C^{(21)} {C^{(21)}}^{\top} \mathbbm{1} \notag\\
        & = K_1
    \end{align}
    Therefore, we conclude that:
    \begin{align}      \|\boldsymbol{\mathcal{Y}}^1_{i} - \boldsymbol{\mathcal{Y}}^\prime_{c,i}\|^2 &\leq  \|{\mathcal{C}^{(21)}}^{\top}\|^2 \|\boldsymbol{\mathcal{Y}}^2_{i-1} - \boldsymbol{\mathcal{Y}}_{c,i-1}\|^2 \notag \\
    &\stackrel{(a)}{\leq} K_1 \|\mathcal{U}_L\|^2  \left\|\mathcal{U}_{R}^{\top}\boldsymbol{\mathcal{Y}}^2_{i}\right\|^2
    \end{align}
    where $(a)$ follows from Lemma \ref{lemma:plain_within} and \eqref{eq:C21}. 
\end{proof}

\begin{lemma}
\label{lemma:VRTG}
    Under Assumptions \ref{ass:within_conn}, \ref{ass:cross-conn_strong},  \ref{ass:lip}, \ref{ass:bdis}, \ref{ass:gradnoise}, for \textbf{ATC-C}, we have:
    \begin{align}
        &\mathbb{E} \{\|\mathcal{V}_{R}^{\top} \boldsymbol{\mathcal{G}}^1_{x, i}\|^2\} \notag\\
        &\leq 8 \|\mathcal{V}_{R}^{\top}\|^2 L_f^2 \left(\|\mathcal{V}_{L} \|^2 \mathbb{E} \{\| \mathcal{V}_{R}^{\top} \boldsymbol{\mathcal{X}}^1_{i-1} \|^2\} \right. \notag\\
        &\quad \left. + K_1 \|\mathcal{U}_L\|^2  \mathbb{E} \{\left\|\mathcal{U}_{R}^{\top}\boldsymbol{\mathcal{Y}}^2_{i-1}\right\|^2\}\right) + 2\|\mathcal{V}_{R}^{\top}\|^2 K_1 (2G^2 +\sigma^2)
    \end{align}
    and $\mathbb{E} \{\|\mathcal{U}_{R}^{\top} \boldsymbol{\mathcal{G}}^2_{y, i}\|^2\}$ can be bounded similarly. \hfill\qed
\end{lemma}

\begin{proof}
    For $k \in \mathcal{N}^{(1)}$, we define 
    \begin{align}
        \!\!\!\!\!\boldsymbol{s}_{k, i}^x &\triangleq \widehat{\nabla_x J}_k\left(\boldsymbol{x}_{k, i-1}, \boldsymbol{y}_{k, i-1}\right) - \nabla_x J_k\left(\boldsymbol{x}_{k, i-1}, \boldsymbol{y}_{k, i-1}\right) \\
        \!\!\!\!\!\mathcal{G}^1_{x, i} &\triangleq \text{col} \{\nabla_x J_k\left(\boldsymbol{x}_{k, i-1}, \boldsymbol{y}_{k, i-1}\right)\}_{k=1}^{K_1}
    \end{align}
    We can bound $\|\mathcal{V}_{R}^{\top} \boldsymbol{\mathcal{G}}^1_{x, i}\|^2$ by:
    \begin{equation}
    \label{eq:VRTG}
        \begin{aligned}
            &\|\mathcal{V}_{R}^{\top} \boldsymbol{\mathcal{G}}^1_{x, i}\|^2 \\
            &= \left\|\mathcal{V}_{R}^{\top} \left( \mathcal{G}^1_{x, i} + \text{col}\{\boldsymbol{s}_{k, i}^x\}_{k=1}^{K_1}\right) \right\|^2 \\
            & \leq 2\|\mathcal{V}_{R}^{\top} \mathcal{G}^1_{x, i}\|^2 + 2\|\mathcal{V}_{R}^{\top} \text{col}\{\boldsymbol{s}_{k, i}^x\}_{k=1}^{K_1} \|^2 \\
            &\leq 2\|\mathcal{V}_{R}^{\top} \mathcal{G}^1_{x, i}\|^2 + 2\|\mathcal{V}_{R}^{\top}\|^2 \|\text{col}\{\boldsymbol{s}_{k, i}^x\}_{k=1}^{K_1} \|^2 \\
            &= 2\|\mathcal{V}_{R}^{\top} \mathcal{G}^1_{x, i}\|^2 + 2 \|\mathcal{V}_{R}^{\top}\|^2 \sum_{k=1}^{K_1} \|\boldsymbol{s}_{k, i}^x\|^2
        \end{aligned}
    \end{equation}
    Now we take the conditional expectation of $\|\mathcal{V}_{R}^{\top} \boldsymbol{\mathcal{G}}^1_{x, i}\|^2$, since $\boldsymbol{x}_{k, i-1}, \boldsymbol{y}_{k, i-1} \in \boldsymbol{\mathcal{F}}_{i-1}$:
    \begin{equation}
    \label{EVRG}
        \mathbb{E} \{\|\mathcal{V}_{R}^{\top} \boldsymbol{\mathcal{G}}^1_{x, i}\|^2 \mid \boldsymbol{\mathcal{F}}_{i-1}\} \leq 2\|\mathcal{V}_{R}^{\top} \mathcal{G}^1_{x, i}\|^2 + 2\|\mathcal{V}_{R}^{\top}\|^2 K_1 \sigma^2
    \end{equation}
    which follows from Assumption \ref{ass:gradnoise}. The key lies in how to upper bound $\|\mathcal{V}_{R_1}^{\top} \mathcal{G}^1_{x, i}\|^2$. Recall that $\boldsymbol{x}_{c, i-1}, \boldsymbol{y}_{c, i-1}$ denote the centroids of $x, y$ for Teams 1, 2 respectively. Then,
    \begin{equation}
    \label{eq:hardlemma1}
        \begin{aligned}
            &\|\mathcal{V}_{R}^{\top} \mathcal{G}^1_{x, i}\|^2 \\
            &= \|\mathcal{V}_{R}^{\top} \text{col}\{\nabla_x J_k\left(\boldsymbol{x}_{k, i-1}, \boldsymbol{y}_{k, i-1}\right)\}_{k=1}^{K_1}\|^2 \\
            &\stackrel{(a)}{=} \|\mathcal{V}_{R}^{\top} \left(\text{col}\{\nabla_x J_k\left(\boldsymbol{x}_{k, i-1}, \boldsymbol{y}_{k, i-1}\right)\}_{k=1}^{K_1} - \text{col}\{\nabla_x J_k\left(\boldsymbol{x}_{c, i-1}, \right. \right.\\
            & \quad \left. \left.\boldsymbol{y}_{c, i-1}\right)\}_{k=1}^{K_1} \right) + \mathcal{V}_{R}^{\top} \left( \text{col} \{\nabla_x J_k\left(\boldsymbol{x}_{c, i-1}, \boldsymbol{y}_{c, i-1}\right)\}_{k=1}^{K_1} \right.\\
            & \left. \quad - (\mathbbm{1} {p^{(1)}}^{\top}) \otimes I_{M_1} \text{col} \{\nabla_x J_k\left(\boldsymbol{x}_{c, i-1}, \boldsymbol{y}_{c, i-1}\right)\}_{k=1}^{K_1} \right)\|^2 \\
            &\leq 2\|\mathcal{V}_{R}^{\top}\|^2 \sum_{k=1}^{K_1} \|\nabla_x J_k\left(\boldsymbol{x}_{k, i-1}, \boldsymbol{y}_{k, i-1}\right) - \nabla_x J_k\left(\boldsymbol{x}_{c, i-1}, \right.\\
            & \left. \quad \boldsymbol{y}_{c, i-1}\right)\|^2 + 2\|\mathcal{V}_{R}^{\top}\|^2 \sum_{k=1}^{K_1} \|\nabla_x J_k\left(\boldsymbol{x}_{c, i-1}, \boldsymbol{y}_{c, i-1}\right) \\
            & \quad - \nabla_x J^{(1)}\left(\boldsymbol{x}_{c, i-1}, \boldsymbol{y}_{c, i-1}\right)\|^2 \\
            &\stackrel{(b)}{\leq} 2\|\mathcal{V}_{R}^{\top}\|^2 \sum_{k=1}^{K_1} L_f^2 (\|\boldsymbol{x}_{k, i-1} - \boldsymbol{x}_{c, i-1}\| + \|\boldsymbol{y}_{k, i-1} - \boldsymbol{y}_{c, i-1}\|)^2 \\
            & \quad + 2\|\mathcal{V}_{R}^{\top}\|^2 K_1 G^2\\
            &\stackrel{(c)}{\leq} 2\|\mathcal{V}_{R}^{\top}\|^2 L_f^2 (2\|\boldsymbol{\mathcal{X}}^1_{i-1} - \boldsymbol{\mathcal{X}}_{c, i-1} \|^2 + 2\|\boldsymbol{\mathcal{Y}}^1_{i-1} - \boldsymbol{\mathcal{Y}}^{\prime}_{c,i-1}\|^2) \\
            & \quad + 2\|\mathcal{V}_{R}^{\top}\|^2 K_1 G^2
        \end{aligned}
    \end{equation}
    where $(a)$ follows from the fact that $V_R^{\top} \mathbbm{1} = 0$; $(b)$ follows from Assumption \ref{ass:lip} and \ref{ass:bdis}; $(c)$ follows from the inequality $(a+b)^2 \leq 2a^2 + 2b^2$ and the facts that $\sum_{k=1}^{K_1} \|\boldsymbol{x}_{k, i-1} - \boldsymbol{x}_{c, i-1}\|^2 = \|\boldsymbol{\mathcal{X}}^1_{i-1} - \boldsymbol{\mathcal{X}}_{c, i-1} \|^2$, $\sum_{k=1}^{K_1}\|\boldsymbol{y}_{k, i-1} - \boldsymbol{y}_{c, i-1}\|^2 = \|\boldsymbol{\mathcal{Y}}^1_{i-1} - \boldsymbol{\mathcal{Y}}^{\prime}_{c,i-1}\|^2$.
    By Lemmas \ref{lemma:plain_within} and \ref{lemma:plain_cross}, we can further bound \eqref{eq:hardlemma1} by:
    \begin{align}
    \label{eq:VRTG1}
        &\|\mathcal{V}_{R}^{\top} \mathcal{G}^1_{x, i}\|^2 \notag\\
        &\leq 2\|\mathcal{V}_{R}^{\top}\|^2 L_f^2 (2\|\boldsymbol{\mathcal{X}}^1_{i-1} - \boldsymbol{\mathcal{X}}_{c, i-1} \|^2 + 2\|\boldsymbol{\mathcal{Y}}^1_{i-1} - \boldsymbol{\mathcal{Y}}^{\prime}_{c,i-1}\|^2) \notag\\
        & \quad + 2\|\mathcal{V}_{R}^{\top}\|^2 K_1 G^2 \notag\\
        & \leq 2\|\mathcal{V}_{R}^{\top}\|^2 L_f^2 (2\|\mathcal{V}_{L} \|^2 \| \mathcal{V}_{R}^{\top} \boldsymbol{\mathcal{X}}^1_{i-1} \|^2 + 2K_1 \|\mathcal{U}_L\|^2  \left\|\mathcal{U}_{R}^{\top}\boldsymbol{\mathcal{Y}}^2_{i-1}\right\|^2) \notag\\
        & \quad + 2\|\mathcal{V}_{R}^{\top}\|^2 K_1 G^2
    \end{align}
    Combining \eqref{eq:VRTG1} and \eqref{EVRG}, we get:
    \begin{align}
        &\mathbb{E} \{\|\mathcal{V}_{R}^{\top} \boldsymbol{\mathcal{G}}^1_{x, i}\|^2\} \notag\\
        &\leq 8 \|\mathcal{V}_{R}^{\top}\|^2 L_f^2 \left(\|\mathcal{V}_{L} \|^2 \mathbb{E} \{\| \mathcal{V}_{R}^{\top} \boldsymbol{\mathcal{X}}^1_{i-1} \|^2\} \right. \notag\\
        &\quad \left. + K_1 \|\mathcal{U}_L\|^2  \mathbb{E} \{\left\|\mathcal{U}_{R}^{\top}\boldsymbol{\mathcal{Y}}^2_{i-1}\right\|^2\}\right) + 2\|\mathcal{V}_{R}^{\top}\|^2 K_1 (2G^2 +\sigma^2)
    \end{align}
\end{proof}
\vspace{-1.5em}
\subsection{Main Proof}
\label{proof:lemma4}
This proof below is similar to the proof of Section \ref{proof:lemma2}. 

By Lemmas \ref{lemma:plain_within} and \ref{lemma:plain_cross}, we have the following bound:
\begin{align}
\label{eq:plain_first}
    & \mathbb{E} \{\|\boldsymbol{\mathcal{X}}^1_{i} - \boldsymbol{\mathcal{X}}_{c,i}\|^2 + \|\boldsymbol{\mathcal{Y}}^2_{i} - \boldsymbol{\mathcal{Y}}_{c,i}\|^2 + \notag\\
    & \qquad \|\boldsymbol{\mathcal{X}}^2_{i} - \boldsymbol{\mathcal{X}}^\prime_{c,i}\|^2 + \|\boldsymbol{\mathcal{Y}}^1_{i} - \boldsymbol{\mathcal{Y}}^\prime_{c,i}\|^2\} \notag\\
    & \leq (1+K_1) \|\mathcal{V}_{L} \|^2 \mathbb{E} \{\| \mathcal{V}_{R}^{\top} \boldsymbol{\mathcal{X}}^1_i \|^2\} \notag \\
    & \quad +(1+K_2)\|\mathcal{U}_L\|^2 \mathbb{E} \{ \left\|\mathcal{U}_{R}^{\top}\boldsymbol{\mathcal{Y}}^2_{i}\right\|^2\}
\end{align}
To bound $\| \mathcal{V}_{R}^{\top} \boldsymbol{\mathcal{X}}^1_i \|^2$, we have:
 \begin{equation}
    \begin{aligned}
    \label{initial}
            &\| \mathcal{V}_{R}^{\top} \boldsymbol{\mathcal{X}}^1_i \|^2 \\
            & = \|\mathcal{V}_{R}^{\top} {\mathcal{A}^{(1)}}^{\top} (\boldsymbol{\mathcal{X}}^1_{i-1} - \mu \boldsymbol{\mathcal{G}}^1_{x, i})\|^2 \\
            &\stackrel{(a)}{=} \|\mathcal{J}_{\varepsilon}^{\top} \mathcal{V}_{R}^{\top} (\boldsymbol{\mathcal{X}}^1_{i-1} - \mu \boldsymbol{\mathcal{G}}^1_{x, i})\|^2 \\
            & \leq \|\mathcal{J}_{\varepsilon}^{\top}\|^2 \|\mathcal{V}_{R}^{\top} (\boldsymbol{\mathcal{X}}^1_{i-1} - \mu \boldsymbol{\mathcal{G}}^1_{x, i})\|^2 \\
            & \stackrel{(b)}{\leq} \|\mathcal{J}_{\varepsilon}^{\top}\|^2 \left(\frac{1}{\left\|\mathcal{J}_{\varepsilon}^{\top}\right\|}\left\|\mathcal{V}_{R}^{\top}\boldsymbol{\mathcal{X}}^1_{i-1}\right\|^2 + \frac{1}{1-\left\|\mathcal{J}_{\varepsilon}^{\top}\right\|} \mu^2 \|\mathcal{V}_{R}^{\top} \boldsymbol{\mathcal{G}}^1_{x, i}\|^2\right) \\
            & = \|\mathcal{J}_{\varepsilon}^{\top} \|\left\|\mathcal{V}_{R}^{\top}\boldsymbol{\mathcal{X}}^1_{i-1}\right\|^2 + \frac{\|\mathcal{J}_{\varepsilon}^{\top}\|^2}{1-\left\|\mathcal{J}_{\varepsilon}^{\top}\right\|} \mu^2 \|\mathcal{V}_{R}^{\top} \boldsymbol{\mathcal{G}}^1_{x, i}\|^2
        \end{aligned}
    \end{equation}
    where  $(a)$ follows from steps similar to \eqref{eq:PRB}
    and $(b)$ follows from the convexity of $\|\cdot\|^2$ and Jensen's inequality, i.e. $\|a+b\|^2 \leq \frac{1}{\alpha} \|a\|^2 + \frac{1}{1-\alpha} \|b\|^2$.
    
     Therefore, by Lemma \ref{lemma:VRTG}, we obtain:
    \begin{equation}
    \label{hardlemma1_EVRX}
        \begin{aligned}
            &\mathbb{E} \{\|\mathcal{V}_R^{\top} \boldsymbol{\mathcal{X}}^1_i\|^2 \}\\
            & \leq
            \|\mathcal{J}_{\varepsilon}^{\top} \| \mathbb{E} \{\left\|\mathcal{V}_{R}^{\top}\boldsymbol{\mathcal{X}}^1_{i-1}\right\|^2\} + \frac{\|\mathcal{J}_{\varepsilon}^{\top}\|^2}{1-\left\|\mathcal{J}_{\varepsilon}^{\top}\right\|} \mu^2 \mathbb{E} \{\|\mathcal{V}_{R}^{\top} \boldsymbol{\mathcal{G}}^1_{x, i}\|^2\} \\
            &\leq       \left(\|\mathcal{J}_{\varepsilon}^{\top} \| + 8 \frac{\|\mathcal{V}_{L}^{\top}\|^2 \|\mathcal{J}_{\varepsilon}^{\top}\|^2 \|\mathcal{V}_{R}^{\top}\|^2}{1-\left\|\mathcal{J}_{\varepsilon}^{\top}\right\|} \mu^2  L_f^2\right) \mathbb{E} \{\left\|\mathcal{V}_{R}^{\top}\boldsymbol{\mathcal{X}}^1_{i-1}\right\|^2\} \\
        & \quad
        + 8 \frac{\|\mathcal{J}_{\varepsilon}^{\top}\|^2 \|\mathcal{V}_{R}^{\top}\|^2 \|\mathcal{U}_L\|^2}{1-\left\|\mathcal{J}_{\varepsilon}^{\top}\right\|} \mu^2  L_f^2 K_1 \mathbb{E} \{\|\mathcal{U}_R^{\top} \boldsymbol{\mathcal{Y}}^2_{i-1}\|^2 \} \\
        & \quad + 2 \frac{\|\mathcal{J}_{\varepsilon}^{\top}\|^2 \|\mathcal{V}_{R}^{\top}\|^2}{1-\left\|\mathcal{J}_{\varepsilon}^{\top}\right\|} \mu^2 K_1 (\sigma^2 + 2G^2)
        \end{aligned}
    \end{equation}
  
    Similarly, we obtain the following:
\begin{equation}
    \begin{aligned}
        &\mathbb{E} \{\left\|\mathcal{U}_{R}^{\top}\boldsymbol{\mathcal{Y}}^2_{i}\right\|^2\}\\
        & \leq
        \left(\left\|\mathcal{K}_{\gamma}^{\top}\right\| + 8 \frac{\|\mathcal{U}_{L}^{\top}\|^2 \|\mathcal{K}_{\gamma}^{\top}\|^2 \|\mathcal{U}_{R}^{\top}\|^2}{1-\left\|\mathcal{K}_{\gamma}^{\top}\right\|} \mu^2  L_f^2\right) \mathbb{E} \{\left\|\mathcal{U}_{R}^{\top}\boldsymbol{\mathcal{Y}}^2_{i-1}\right\|^2\} \\
        & \quad
        + 8 \frac{ \|\mathcal{K}_{\gamma}^{\top}\|^2 \|\mathcal{U}_{R}^{\top}\|^2 \|\mathcal{V}_{L}^{\top}\|^2}{1-\left\|\mathcal{K}_{\gamma}^{\top}\right\|} \mu^2  L_f^2 K_2 \mathbb{E} \{\left\|\mathcal{V}_{R}^{\top}\boldsymbol{\mathcal{X}}^1_{i-1}\right\|^2\} \\
        & \quad + 2 \frac{\|\mathcal{K}_{\gamma}^{\top}\|^2 \|\mathcal{U}_{R}^{\top}\|^2}{1-\left\|\mathcal{K}_{\gamma}^{\top}\right\|} \mu^2 K_2 (\sigma^2 + 2G^2) 
    \end{aligned}
\end{equation}
For compactness, we introduce the following scalar coefficients:
\begin{align}
    d_1 &= \frac{\|\mathcal{V}_{L}^{\top}\|^2 \|\mathcal{J}_{\varepsilon}^{\top}\|^2 \|\mathcal{V}_{R}^{\top}\|^2}{1-\left\|\mathcal{J}_{\varepsilon}^{\top}\right\|}\\
    e_1 &= \frac{\|\mathcal{J}_{\varepsilon}^{\top}\|^2 \|\mathcal{V}_{R}^{\top}\|^2 \|\mathcal{U}_L\|^2}{1-\left\|\mathcal{J}_{\varepsilon}^{\top}\right\|}\\
    f_1 &= \frac{\|\mathcal{J}_{\varepsilon}^{\top}\|^2 \|\mathcal{V}_{R}^{\top}\|^2}{1-\left\|\mathcal{J}_{\varepsilon}^{\top}\right\|}\\
    d_2 &= \frac{\|\mathcal{U}_{L}^{\top}\|^2 \|\mathcal{K}_{\gamma}^{\top}\|^2 \|\mathcal{U}_{R}^{\top}\|^2}{1-\left\|\mathcal{K}_{\gamma}^{\top}\right\|}\\
    e_2 &= \frac{ \|\mathcal{K}_{\gamma}^{\top}\|^2 \|\mathcal{U}_{R}^{\top}\|^2 \|\mathcal{V}_{L}^{\top}\|^2}{1-\left\|\mathcal{K}_{\gamma}^{\top}\right\|}\\
    f_2 &= \frac{\|\mathcal{K}_{\gamma}^{\top}\|^2 \|\mathcal{U}_{R}^{\top}\|^2}{1-\left\|\mathcal{K}_{\gamma}^{\top}\right\|}
\end{align}

Therefore, we have:
\begin{equation}
\label{eq:VRXURY}
    \begin{aligned}
        & \mathbb{E} \{\left\|\mathcal{V}_{R}^{\top}\boldsymbol{\mathcal{X}}^1_{i}\right\|^2\} + \mathbb{E} \{\left\|\mathcal{U}_{R}^{\top}\boldsymbol{\mathcal{Y}}^2_{i}\right\|^2\} \\
        &\leq (\|\mathcal{J}_{\varepsilon}^{\top} \| + 8 d_1 \mu^2  L_f^2 + 8 e_2 \mu^2  L_f^2 K_2) \mathbb{E} \{\left\|\mathcal{V}_{R}^{\top}\boldsymbol{\mathcal{X}}^1_{i-1}\right\|^2\} \\
        & \quad + (\left\|\mathcal{K}_{\gamma}^{\top}\right\| + 8 d_2 \mu^2  L_f^2 + 8 e_1 \mu^2 L_f^2 K_1) \mathbb{E} \{\left\|\mathcal{U}_{R}^{\top}\boldsymbol{\mathcal{Y}}^2_{i-1}\right\|^2\}\\
        & \qquad + 2 \mu^2 (f_1  K_1 + f_2 K_2) (\sigma^2 + 2G^2)\\
        & \stackrel{(a)}{\leq}
        \beta (\mathbb{E} \{\left\|\mathcal{V}_{R}^{\top}\boldsymbol{\mathcal{X}}^1_{i-1}\right\|^2\} + \mathbb{E} \{\left\|\mathcal{U}_{R}^{\top}\boldsymbol{\mathcal{Y}}^2_{i-1}\right\|^2\}) \\
        & \quad +  2 \mu^2 (f_1  K_1 + f_2 K_2) (\sigma^2 + 2G^2)
    \end{aligned}
\end{equation}
where $(a)$ uses the fact that, for sufficiently small $\mu$ and properly selected $\epsilon, \gamma$, we have:
\begin{align}
\label{eq:beta}
    \beta &= \max \{\|\mathcal{J}_{\varepsilon}^{\top} \| + 8 \mu^2 (d_1 + e_2 K_2) L_f^2,  \notag\\
    & \quad\left\|\mathcal{K}_{\gamma}^{\top}\right\| + 8 \mu^2 (d_2 + e_1 K_1) L_f^2\} < 1
\end{align}
Noting that $\|\mathcal{J}_{\varepsilon}^{\top} \| < 1$ is independent of $\mu$, we have $1 - \beta = O(1)$. Therefore, 
\begin{equation}
    \begin{aligned}
        &\mathbb{E} \{\left\|\mathcal{V}_{R}^{\top}\boldsymbol{\mathcal{X}}^1_{i}\right\|^2\} + \mathbb{E} \{\left\|\mathcal{U}_{R}^{\top}\boldsymbol{\mathcal{Y}}^2_{i}\right\|^2\} \\
        & \leq
        \beta^{i+1} (\mathbb{E} \{\left\|\mathcal{V}_{R}^{\top}\boldsymbol{\mathcal{X}}_{-1}\right\|^2\} + \mathbb{E} \{\left\|\mathcal{U}_{R}^{\top}\boldsymbol{\mathcal{Y}}_{-1}\right\|^2\}) \\
        & \quad +  2 \mu^2 (f_1  K_1 + f_2 K_2) (\sigma^2 + 2G^2) \sum_{j=0}^{i} \beta^j \\
        & \stackrel{(a)}{\leq}
        O(\mu^2) + \frac{2 \mu^2 (f_1  K_1 + f_2 K_2) (\sigma^2 + 2G^2)}{1-\beta} \\
        & = O(\mu^2)
    \end{aligned}
\end{equation}
where $(a)$ holds when 
\begin{equation}
    \beta^i \leq O(\mu^2) \Longleftrightarrow i \geq i_\beta = \frac{\log \left(O\left(\mu^2\right)\right)}{\log \left(\beta\right)}
\end{equation}
By \eqref{eq:plain_first}, we conclude that:
\begin{align}
    & \mathbb{E} \{\|\boldsymbol{\mathcal{X}}^1_{i} - \boldsymbol{\mathcal{X}}_{c,i}\|^2 + \|\boldsymbol{\mathcal{Y}}^2_{i} - \boldsymbol{\mathcal{Y}}_{c,i}\|^2 + \\
    & \qquad \|\boldsymbol{\mathcal{X}}^2_{i} - \boldsymbol{\mathcal{X}}^\prime_{c,i}\|^2 + \|\boldsymbol{\mathcal{Y}}^1_{i} - \boldsymbol{\mathcal{Y}}^\prime_{c,i}\|^2\} \leq O(\mu^2)
\end{align}
when $i \geq i_\beta$.
\vspace{-0.3em}
\section{Proof of Lemma \ref{lemma:zerosum_learning_dynamic}}
\label{proof:lemma3}

As can be verified and already shown in \eqref{eq:centroid_x_evolution}, we have the following recursions for the centroids:
    \begin{align}
        \boldsymbol{x}_{c,i} &= \boldsymbol{x}_{c,i-1} - \mu \sum_{k \in \mathcal{N}^{(1)}} p_k \widehat{\nabla_x J}_k (\boldsymbol{x}_{k,i-1}, \boldsymbol{y}_{k,i-1})  \label{eq:learning_centroid_x}\\
        \boldsymbol{y}_{c,i} &= \boldsymbol{y}_{c,i-1} - \mu \sum_{k\in \mathcal{N}^{(2)}} p_{k} \widehat{\nabla_y J}_{k} (\boldsymbol{x}_{k,i-1}, \boldsymbol{y}_{k,i-1}) \label{eq:learning_centroid_y}
    \end{align}
    Note that $p = [p^{(1)};p^{(2)}]$.
   Therefore, for $\boldsymbol{x}_{c,i}$ from (\ref{eq:learning_centroid_x}), we have:
   \begin{equation}
   \label{eq:learning_1}
        \begin{aligned}
            \!\!\!\!\!&\boldsymbol{x}_{c,i} = \boldsymbol{x}_{c,i-1} - \mu \sum_{k \in \mathcal{N}^{(1)}} p_k \nabla_x J_k (\boldsymbol{x}_{c,i-1}, \boldsymbol{y}_{c,i-1}) \\
            & 
             +  \underline{\mu \sum_{k \in \mathcal{N}^{(1)}} p_k (\nabla_x J_k (\boldsymbol{x}_{c,i-1}, \boldsymbol{y}_{c,i-1}) - \nabla_x J_k (\boldsymbol{x}_{k,i-1}, \boldsymbol{y}_{k,i-1}))}\\
            & 
            \underline{+  \mu \sum_{k \in \mathcal{N}^{(1)}} p_k (\nabla_x J_k (\boldsymbol{x}_{k,i-1}, \boldsymbol{y}_{k,i-1}) - \widehat{\nabla_x J}_k (\boldsymbol{x}_{k,i-1}, \boldsymbol{y}_{k,i-1}))}
        \end{aligned}   
   \end{equation}
   Let $\boldsymbol{d}^{(x)}_{c,i}$ denote the underlined terms. According to the relationship between $J^{(1)}(\cdot)$ and $J_k(\cdot)$ as shown in \eqref{eq:global_game_ya}, we can rewrite (\ref{eq:learning_1}) as 
   \begin{equation}
      \boldsymbol{x}_{c,i} = \boldsymbol{x}_{c,i-1} - \mu \nabla_x J^{(1)}(\boldsymbol{x}_{c,i-1}, \boldsymbol{y}_{c, i-1}) + \boldsymbol{d}^{(x)}_{c,i} 
   \end{equation}
   Next we show that $\mathbb{E} \{\|\boldsymbol{d}^{(x)}_{c,i}\|^2\} \leq O(\mu^2)$. Indeed:
   \begin{equation}
   \label{eq:d_ki}
       \begin{aligned}
           & \mathbb{E} \{\|\boldsymbol{d}^{(x)}_{c,i}\|^2\}\\
           & \stackrel{(a)}{\leq}
           2 \mu^2 \sum_{k \in \mathcal{N}^{(1)}} p_k \mathbb{E} \{\|\nabla_x J_k (\boldsymbol{x}_{c,i-1}, \boldsymbol{y}_{c,i-1}) - \nabla_x J_k (\boldsymbol{x}_{k,i-1}, \\
           & \quad \boldsymbol{y}_{k,i-1})\|^2\} 
           + 2 \mu^2 \sum_{k \in \mathcal{N}^{(1)}} p_k \mathbb{E} \{\|\nabla_x J_k (\boldsymbol{x}_{k,i-1}, \boldsymbol{y}_{k,i-1}) \\
           & \quad -  \widehat{\nabla_x J}_k (\boldsymbol{x}_{k,i-1}, \boldsymbol{y}_{k,i-1})\|^2\}\\
           & \stackrel{(b)}{\leq}
           2 \mu^2 \sum_{k \in \mathcal{N}^{(1)}} p_k L_f^2 \mathbb{E} \{2 \|\boldsymbol{x}_{c,i-1} - \boldsymbol{x}_{k,i-1}\|^2 + 2 \|\boldsymbol{y}_{c,i-1} \\
           & \quad - \boldsymbol{y}_{k,i-1}\|^2\} 
           + 2 \mu^2 \sum_{k \in \mathcal{N}^{(1)}} p_k \sigma^2 \\
           &\stackrel{(c)}{\leq}
           4 \mu^2 L_f^2 O(\mu^2) +  2 \mu^2 \sigma^2 \\
           &= 
           O(\mu^2)
       \end{aligned}
   \end{equation}
   where $(a)$ follows from Jensen's inequality $\|a+b\|^2 \leq 2\|a\|^2 + 2\|b\|^2$; $(b)$ follows from Assumption \ref{ass:lip} and Assumption \ref{ass:gradnoise}: $\mathbb{E} \{\|\nabla_x J_k (\boldsymbol{x}_{k,i-1}, \boldsymbol{y}_{k,i-1}) - \widehat{\nabla_x J}_k (\boldsymbol{x}_{k,i-1}, \boldsymbol{y}_{k,i-1})\|^2 \mid \boldsymbol{\mathcal{F}}_{i-1}\} \leq \sigma^2$; $(c)$ follows from Lemmas \ref{lemma:zerosum_within_cross}.
   
   We will show similar results for $\boldsymbol{y}_{c,i}$ from (\ref{eq:learning_centroid_y}) in the following. After some calculation, we get:
   \begin{equation}
   \label{eq:learning_2}
        \begin{aligned}
            & \boldsymbol{y}_{c,i-1} = \boldsymbol{y}_{c,i} - \mu \sum_{k\in \mathcal{N}^{(2)}} p_k \nabla_y J_k (\boldsymbol{x}_{c,i-1}, \boldsymbol{y}_{c,i-1}) \\
            &
            + \underline{\mu \sum_{k \in \mathcal{N}^{(2)}} p_k (\nabla_y J_k (\boldsymbol{x}_{c,i-1}, \boldsymbol{y}_{c,i-1}) - \nabla_y J_k (\boldsymbol{x}_{k,i-1}, \boldsymbol{y}_{k,i-1}))}\\
            &
            \underline{+ \mu \sum_{k\in \mathcal{N}^{(2)}} p_k (\nabla_y J_k (\boldsymbol{x}_{k,i-1}, \boldsymbol{y}_{k,i-1}) - \widehat{\nabla_y J}_k (\boldsymbol{x}_{k,i-1}, \boldsymbol{y}_{k,i-1}))}
        \end{aligned}   
   \end{equation}
   Let $\boldsymbol{d}_{c,i}^{(y)}$ denote the underlined terms. From the relationship between $J^{(2)}$ and $J_k$ as shown in \eqref{eq:global_game_yb}, we can easily show: $\boldsymbol{y}_{c,i} = \boldsymbol{y}_{c,i-1} - \mu  \nabla_y J^{(2)}(\boldsymbol{x}_{c,i-1}, \boldsymbol{y}_{c,i-1}) + \boldsymbol{d}_{c,i}^{(y)}$. What remains is to
   show $\mathbb{E} \{\|\boldsymbol{d}_{c,i}^{(y)}\|^2\} \leq O(\mu^2)$:
   \begin{equation}
       \begin{aligned}
           & \mathbb{E} \{\|\boldsymbol{d}_{c,i}^{(y)}\|^2\}\\
           & \leq
           2 \mu^2 \sum_{k\in \mathcal{N}^{(2)}} p_k \mathbb{E} \{\|\nabla_y J_k (\boldsymbol{x}_{c,i-1}, \boldsymbol{y}_{c,i-1}) - \nabla_y J_k (\boldsymbol{x}_{k,i-1}, \\
           & \quad \boldsymbol{y}_{k,i-1})\|^2\}
           + 2 \mu^2 \sum_{k\in \mathcal{N}^{(2)}} p_k \mathbb{E} \{\|\nabla_y J_k (\boldsymbol{x}_{k,i-1}, \boldsymbol{y}_{k,i-1}) \\
           & \quad - \widehat{\nabla_y J}_k (\boldsymbol{x}_{k,i-1}, \boldsymbol{y}_{k,i-1})\|^2\} \\
           & \leq
           2 \mu^2 \sum_{k\in \mathcal{N}^{(2)}} p_k L_f^2 \mathbb{E} \{2 \|\boldsymbol{x}_{c,i-1} - \boldsymbol{x}_{k,i-1}\|^2 + 2 \|\boldsymbol{y}_{c,i-1} \\
           & \quad - \boldsymbol{y}_{k,i-1}\|^2\} + 2 \mu^2 \sum_{k\in \mathcal{N}^{(2)}} p_k \sigma^2 \\
           &\stackrel{(a)}{\leq}
           4 \mu^2 L_f^2 O(\mu^2) +  2 \mu^2 \sigma^2 \\
           &= 
           O(\mu^2)
       \end{aligned}
   \end{equation}
    where $(a)$ follows from Lemma \ref{lemma:zerosum_within_cross}. Now we define 
    \begin{equation}
    \label{definition_Dci}
        \boldsymbol{d}_{c,i} \triangleq \left[\begin{array}{c}
             \boldsymbol{d}^{(x)}_{c,i} \\
             \boldsymbol{d}_{c,i}^{(y)} 
        \end{array}\right]
    \end{equation}
    We arrive at the conclusion that $\mathbb{E} \{\|\boldsymbol{d}_{c,i}\|^2\} \leq O(\mu^2)$, after sufficient iterations (i.e., $i \geq i_\alpha$ for \textbf{ATC-ITC}, $i \geq i_\beta$ for \textbf{ATC-C}).
\vspace{-0.4em}
\section{Proof of Theorem \ref{thm:zerosum_conv}}
\label{proof:thm1}
From Assumption \ref{ass:lip}, we have: 
    \begin{equation}
    \label{eq:operator_lip}
        \|F(z_1) - F(z_2)\| \leq 2L_f \|z_1 -z_2\|
    \end{equation}
    We set $L \triangleq 2L_f$. From Lemma \ref{lemma:zerosum_learning_dynamic}, we have after sufficient iterations:
    \begin{equation}
        \begin{aligned} 
            & \mathbb{E} \{\left\|\boldsymbol{z}_{c,i}-z^{\star}\right\|^2\} \\ 
            & =
            \mathbb{E} \{\left\|\boldsymbol{z}_{c,i-1}-\mu F\left(\boldsymbol{z}_{c,i-1}\right) + \boldsymbol{d}_{c,i} - z^{\star}\right\|^2\} \\
            & =
            \mathbb{E} \{\left\|\boldsymbol{z}_{c,i-1}-\mu F\left(\boldsymbol{z}_{c,i-1}\right)-z^{\star}\right\|^2\} \\
            & \quad + 2\mathbb{E}\{\langle \boldsymbol{z}_{c,i-1}-\mu F\left(\boldsymbol{z}_{c,i-1}\right)-z^{\star}, \boldsymbol{d}_{c,i}\rangle\} + \mathbb{E}\{\left\|\boldsymbol{d}_{c,i}\right\|^2\}
        \end{aligned}\label{eq:convg_nash}
    \end{equation}
    The key lies in upper bounding the term $2\mathbb{E}\{\langle \boldsymbol{z}_{c,i-1}-\mu F\left(\boldsymbol{z}_{c,i-1}\right)-z^{\star}, \boldsymbol{d}_{c,i}\rangle\}$.
    From Lemma \ref{lemma:zerosum_learning_dynamic}, we have:
    
    \begin{equation}
        \begin{aligned}
            &2\mathbb{E}\{\langle \boldsymbol{z}_{c,i-1}-\mu F\left(\boldsymbol{z}_{c,i-1}\right)-z^{\star}, \boldsymbol{d}_{c,i}\rangle\} \\
            & = 
            2\mathbb{E}\{\langle \boldsymbol{x}_{c,i-1}- \mu \nabla_x J^{(1)}(\boldsymbol{x}_{c,i-1}, \boldsymbol{y}_{c,i-1}) - x^{\star}, \boldsymbol{d}^{(x)}_{c,i}\rangle\} \\
            & \quad+ 2\mathbb{E}\{\langle \boldsymbol{y}_{c,i-1} - \mu \nabla_y J^{(2)}(\boldsymbol{x}_{c,i-1}, \boldsymbol{y}_{c,i-1}) - y^{\star}, \boldsymbol{d}_{c,i}^{(y)}\rangle\}
        \end{aligned}
    \end{equation}
    We will only show the derivation for the first term, as the derivation for the second is nearly the same. Recall that
    \begin{equation}
        \begin{aligned}
            \boldsymbol{d}^{(x)}_{c,i} =& \mu \sum_{k \in \mathcal{N}^{(1)}} p_k (\nabla_x J_k (\boldsymbol{x}_{c,i-1}, \boldsymbol{y}_{c,i-1}) -\nabla_x J_k (\boldsymbol{x}_{k,i-1},\\
            &\boldsymbol{y}_{k,i-1})) + \mu \sum_{k \in \mathcal{N}^{(1)}} p_k (\nabla_x J_k (\boldsymbol{x}_{k,i-1}, \boldsymbol{y}_{k,i-1}) \\
            &- \widehat{\nabla_x J}_k (\boldsymbol{x}_{k,i-1}, \boldsymbol{y}_{k,i-1}))
        \end{aligned}
    \end{equation}
    Then, we have:
    \begin{equation}
        \begin{aligned}
            & 2\mathbb{E}\{\langle \boldsymbol{x}_{c,i-1}- \mu \nabla_x J^{(1)}(\boldsymbol{x}_{c,i-1}, \boldsymbol{y}_{c,i-1}) - x^{\star}, \boldsymbol{d}^{(x)}_{c,i}\rangle \mid \boldsymbol{\mathcal{F}}_{i-1}\}\\
            & \stackrel{(a)}{=} 2 \langle \boldsymbol{x}_{c,i-1}- \mu \nabla_x J^{(1)}(\boldsymbol{x}_{c,i-1}, \boldsymbol{y}_{c,i-1}) - x^{\star}, \\
            & \quad \mu \sum_{k \in \mathcal{N}^{(1)}} p_k (\nabla_x J_k (\boldsymbol{x}_{c,i-1}, \boldsymbol{y}_{c,i-1}) - \nabla_x J_k (\boldsymbol{x}_{k,i-1}, \boldsymbol{y}_{k,i-1}))\rangle \\
            & \stackrel{(b)}{=}
            2 \mu \delta \langle \boldsymbol{x}_{c,i-1}- \mu \nabla_x J^{(1)}(\boldsymbol{x}_{c,i-1}, \boldsymbol{y}_{c,i-1}) - x^{\star}, \\
            &\quad \frac{1}{\delta} \sum_{k \in \mathcal{N}^{(1)}} p_k (\nabla_x J_k (\boldsymbol{x}_{c,i-1}, \boldsymbol{y}_{c,i-1}) - \nabla_x J_k (\boldsymbol{x}_{k,i-1}, \boldsymbol{y}_{k,i-1}))\rangle \\
            & \stackrel{(c)}{\leq}
            \mu \delta \|\boldsymbol{x}_{c,i-1}- \mu \nabla_x J^{(1)}(\boldsymbol{x}_{c,i-1}, \boldsymbol{y}_{c,i-1}) - x^{\star}\|^2\\
            & 
            \quad + \frac{\mu}{\delta} \sum_{k \in \mathcal{N}^{(1)}} p_k L_f^2 (\|\boldsymbol{x}_{c,i-1} - \boldsymbol{x}_{k,i-1}\| + \|\boldsymbol{y}_{c,i-1} - \boldsymbol{y}_{k,i-1}\|)^2
        \end{aligned}
    \end{equation}
    where $(a)$ follows from Assumption \ref{ass:gradnoise}; $(b)$ introduces a fixed constant $\delta > 0$ that can be chosen independently from $\mu$, for example $\delta = \nu$; $(c)$ follows from Assumption \ref{ass:lip}. Then, after sufficient iterations:
    \begin{equation}
        \begin{aligned}
            & 2\mathbb{E}\{\langle \boldsymbol{x}_{c,i-1}- \mu \nabla_x J^{(1)}(\boldsymbol{x}_{c,i-1}, \boldsymbol{y}_{c,i-1}) - x^{\star}, \boldsymbol{d}^{(x)}_{c,i}\rangle \} \\
            & \leq
            \mu \delta \mathbb{E}\{\|\boldsymbol{x}_{c,i-1}- \mu \nabla_x J^{(1)}(\boldsymbol{x}_{c,i-1}, \boldsymbol{y}_{c,i-1}) - x^{\star}\|^2\}\\
            & 
            \quad + \frac{\mu}{\delta} \sum_{k \in \mathcal{N}^{(1)}} p_k L_f^2 (2 \mathbb{E}\{\|\boldsymbol{x}_{c,i-1} - \boldsymbol{x}_{k,i-1}\|^2\} \\
            & \quad + 2\mathbb{E}\{\|\boldsymbol{y}_{c,i-1} - \boldsymbol{y}_{k,i-1}\|^2\}) \\
            & \stackrel{(a)}{\leq}
            \mu \delta \mathbb{E}\{\|\boldsymbol{x}_{c,i-1}- \mu \nabla_x J^{(1)}(\boldsymbol{x}_{c,i-1}, \boldsymbol{y}_{c,i-1}) - x^{\star}\|^2\} \\
            & \quad + O\Big(\frac{\mu^3}{\delta}\Big)
        \end{aligned}
    \end{equation}
    where $(a)$ is due to Lemma \ref{lemma:zerosum_within_cross}. Also considering $2\mathbb{E}\{\langle \boldsymbol{y}_{c,i-1} - \mu \nabla_y J^{(2)}(\boldsymbol{x}_{c,i-1}, \boldsymbol{y}_{c,i-1}) - y^{\star}, \boldsymbol{d}_{c,i}^{(y)}\rangle\}$, we arrive at the following:
    \begin{equation}
        \begin{aligned}
            &2\mathbb{E}\{\langle \boldsymbol{z}_{c,i-1}-\mu F\left(\boldsymbol{z}_{c,i-1}\right)-z^{\star}, \boldsymbol{d}_{c,i}\rangle\}\\
            & \leq
            \mu \delta \mathbb{E} \{\|\boldsymbol{z}_{c,i-1}-\mu F\left(\boldsymbol{z}_{c,i-1}\right)-z^{\star}\|^2\} + O\Big(\frac{\mu^3}{\delta}\Big)
        \end{aligned}
    \end{equation}
    In addition, we have $\mathbb{E}\{\left\|\boldsymbol{d}_{c,i}\right\|^2\} \leq O(\mu^2)$. Therefore, going back to $ \mathbb{E} \{\left\|\boldsymbol{z}_{c,i}-z^{\star}\right\|^2\}$, we get:
    \begin{equation}
        \begin{aligned}
            &\mathbb{E} \{\left\|\boldsymbol{z}_{c,i}-z^{\star}\right\|^2\} \\
            & \leq 
            (1+\delta \mu)\mathbb{E} \{\left\|\boldsymbol{z}_{c,i-1}-\mu F\left(\boldsymbol{z}_{c,i-1}\right)-z^{\star}\right\|^2\} + O(\mu^2) \\
            & \leq 
            (1+\delta \mu)\left(1-2 \mu \nu+L^2 \mu^2\right) \mathbb{E}\{\left\|\boldsymbol{z}_{c,i-1}-z^{\star}\right\|^2\} + O\left(\mu^2\right)
        \end{aligned}
    \end{equation}
    where the last step is due to the following:
    \begin{equation}
        \begin{aligned}
            & \left\|\boldsymbol{z}_{c,i-1}-\mu F\left(\boldsymbol{z}_{c,i-1}\right)+\mu F\left(z^{\star}\right)-z^{\star}\right\|^2\\
            & = 
            \left\|\boldsymbol{z}_{c,i-1}-z^{\star}\right\|^2-2 \mu\left\langle F\left(\boldsymbol{z}_{c,i-1}\right)-F\left(z^{\star}\right), \boldsymbol{z}_{c,i-1} - z^{\star}\right\rangle\\
            &\quad +\mu^2\left\|F\left(\boldsymbol{z}_{c,i-1}\right)-F\left(z^{\star}\right)\right\|^2 \\ 
            & \stackrel{(a)}{\leq}
            \left\|\boldsymbol{z}_{c,i-1}-z^{\star}\right\|^2-2 \mu \nu\left\|\boldsymbol{z}_{c,i-1}-z^{\star}\right\|^2 \\
            & \quad +L^2 \mu^2\left\|\boldsymbol{z}_{c,i-1}-z^{\star}\right\|^2 \\ 
            & =
            \left(1-2 \mu \nu+L^2 \mu^2\right)\left\|\boldsymbol{z}_{c,i-1}-z^{\star}\right\|^2
        \end{aligned}
    \end{equation}
    where $(a)$ follows from Assumption 1 and (\ref{eq:operator_lip}). Now our goal is to choose a proper $\delta$ to make $(1+\delta \mu)\left(1-2 \mu \nu+L^2 \mu^2\right) <1$ when $\mu$ is sufficiently small. A reasonable choice is $\delta = \nu$, let $h(\mu) \triangleq (1+\nu \mu)\left(1-2 \mu \nu+L^2 \mu^2\right)$. We find $h(0) = 1$, $h^\prime(0) = -\nu$, which means when $\mu$ is sufficiently small, we can always ensure $(1+\delta \mu)\left(1-2 \mu \nu+L^2 \mu^2\right) <1$. Specifically, we can get: 
    \begin{equation}
        h(\mu) = 1 - \nu \mu + \left(L^2 - 2 \nu^2\right) \mu^2 + \nu L^2 \mu^3
    \end{equation}
    To let $h(\mu) < 1$, we need:
    \begin{equation}
        - \nu \mu + \left(L^2 - 2 \nu^2\right) \mu^2 + \nu L^2 \mu^3 < 0
    \end{equation}
    By simple algebra, the above inequality is satisfied when 
    \begin{equation}
    \label{eq:mu_range}
        \mu \in \left( 0, \; \dfrac{ - L^2 + 2 \nu^2 + \sqrt{ L^4 + 4 \nu^4 } }{ 2 \nu L^2 } \right)
    \end{equation}
    Therefore, just set $\delta = \nu$, we obtain that after sufficient iterations:
    \begin{equation}
        \mathbb{E} \left[\|\boldsymbol{z}_{c,i} - z^\star\|^2\right] \leq 
        d \, \mathbb{E}\left[\|\boldsymbol{z}_{c,i-1} - z^\star\|^2\right] + O\left(\mu^2\right)
    \end{equation}
    where $d \leq (1+\nu \mu)\left(1-2 \mu \nu+L^2 \mu^2\right)$, and we can ensure $d<1$ when $\mu$ is sufficiently small (e.g. \eqref{eq:mu_range}). We can conclude that:
    
\vspace{-2em}
        \begin{align}
            \limsup_{i \rightarrow \infty} \mathbb{E}\{\|\boldsymbol{z}_{c,i} - z^{\star}\|^2\} 
            &\leq 
            \frac{O\left(\mu^2\right)}{1 - (1+\delta \mu)\left(1-2 \mu \nu+L^2 \mu^2\right)}\notag\\
            & =
            \frac{O\left(\mu^2\right)}{\mu \nu - (L^2 - 2\nu^2) \mu^2 - L^2 \nu \mu^3} \notag\\
            & =
            O(\mu)
        \end{align}

\vspace{-1.1em}
\section*{Acknowledgement}   
The authors wish to thank student Vladyslav Shashkov for useful feedback on the topic of the manuscript. 

\vspace{-0.3em}
\bibliographystyle{IEEEbib}
{\footnotesize
\bibliography{refs}
}

\end{document}